\documentclass[12pt]{article}
\usepackage{amsmath,amsfonts,amsbsy,amssymb,amsthm}
\usepackage{mathtools}
\usepackage{graphicx,psfrag,epsf}
\usepackage{natbib}
\usepackage{bbm}
\usepackage{booktabs}  % for nice tables
\usepackage{url}       % used for URLs
\usepackage{enumitem}  % resume enumerate 
\usepackage{hyperref}
\hypersetup{
  colorlinks,
  citecolor=blue,
  linkcolor=blue,
  urlcolor=blue}
\usepackage{multibib}  % separate references for suppl paper
\usepackage{tikz}
\usetikzlibrary{backgrounds}
\usepackage{xcolor}
\usepackage{comment}
\usepackage[margin=1in]{geometry}
\usepackage{caption,color,setspace,sectsty,array, physics, mathrsfs, subcaption}
\usepackage{pgfplots}
\usepackage{adjustbox, rotating}
\usepackage{chngcntr}
\usepackage{soul}
\usepackage{authblk}
\usepackage{algorithm}
\usepackage{algorithmic}
\usepackage{titling}
\newcites{appendix}{References}% use \citeappendix{} and \citepappendix[][]{}
\usepackage{todonotes}
\usepackage{soul}
\usepackage{color}
\newtheoremstyle{mystyleRM}%              % Name
  {}%                                     % Space above
  {}%                                     % Space below
  {\upshape}%                             % Body font
  {}%                                     % Indent amount
  {\bfseries}%                            % Theorem head font
  {.}%                                    % Punctuation after theorem head
  { }%                                    % Space after theorem head, ' ', or \newline
  {\thmname{#1}\thmnumber{ #2}\thmnote{ (#3)}}%                                     % Theorem head spec (can be left empty, meaning `normal')

\theoremstyle{mystyleRM} 

\newtheorem{theorem}{Theorem}[section]

\newtheorem{corollary}{Corollary}[section]

\newtheorem{remark}{Remark}
\newtheorem{example}{Example}

\newcommand{\blind}{0}

\graphicspath{{./Figures/}}

\newcommand{\doublewidehat}[1]{% For greek letters
  \tikz[baseline=(X.base)]{
    \node[inner sep=0pt] (X) {$\widehat{#1\,}$};
    \node[inner sep=0pt, yshift=-.95ex] at (X.north) {$\widehat{\phantom{#1\,}}$};
  }%
}

\newcommand{\doublewidehatCL}[1]{% For capital letters
  \tikz[baseline=(X.base)]{
    \node[inner sep=0pt] (X) {$\widehat{#1\,}$};
    \node[inner sep=0pt, yshift=-.75ex] at (X.north) {$\widehat{\phantom{#1\,}}$};
  }%
}

\newcommand{\doublewidehatSL}[1]{% For small letters
  \tikz[baseline=(X.base)]{
    \node[inner sep=0pt] (X) {$\widehat{#1\,}$};
    \node[inner sep=0pt, yshift=-.4ex] at (X.north) {$\widehat{\phantom{#1\,}}$};
  }%
}

\newcommand*\samethanks[1][\value{footnote}]{\footnotemark[#1]}

\newcolumntype{C}[1]{>{\centering\arraybackslash}p{#1}}

\newcommand\underlay[4]{%
  \stackengine{0pt}%
  {\kern#2\includegraphics[height=#1]{#4}}%
  {\includegraphics[height=#1]{#3}}%
  {O}{l}{F}{F}{L}%
}
\newcommand\addunderlay[4]{%
  \stackengine{0pt}%
  {\kern#2\includegraphics[height=#1]{#4}}%
  {#3}%
  {O}{l}{F}{F}{L}%
}

\newcommand*{\colorboxed}{}
\def\colorboxed#1#{%
  \colorboxedAux{#1}%
}
\newcommand*{\colorboxedAux}[3]{%
  \begingroup
    \colorlet{cb@saved}{.}%
    \color#1{#2}%
    \boxed{%
      \color{cb@saved}%
      #3%
    }%
  \endgroup
}

\NewCommandCopy{\proofqedsymbol}{\qedsymbol}% save the default
\AtBeginEnvironment{proof}{\renewcommand{\qedsymbol}{\proofqedsymbol}}

\theoremstyle{definition}

\AtBeginEnvironment{example}{%
  \pushQED{\qed}\renewcommand{\qedsymbol}{$\blacksquare$}% triangle
}
\AtEndEnvironment{example}{\popQED\endexample}

\date{}

\begin{document}

\def\spacingset#1{\renewcommand{\baselinestretch}%
{#1}\small\normalsize}

%%%%%%%%%%%%%%%%%%%%%%%%%%%%%%%%%%%%%%%%%%%%%%%%%%%%%%%%%%%%%%%%%%%%%

\if0\blind
{
  \title{\bf Making Event Study Plots Honest: A Functional Data Approach to Causal Inference}

  \author{
  %%%%%%%%
  Chencheng Fang\thanks{Corresponding author: \href{mailto:ccfang@uni-bonn.de}{ccfang@uni-bonn.de}} 
  \thanks{Institute of Finance and Statistics, University of Bonn, Adenauerallee 24-42, 53113 Bonn, Germany} 
  \thanks{Hausdorff Center for Mathematics, Endenicher Allee 62, 53115 Bonn, Germany} \hspace{2.5em}
  %%%%%%%
  Dominik Liebl\samethanks[2] \samethanks[3] \hspace{0em} 
  }

  \maketitle
  \vspace*{-3.5em}

} \fi

\if1\blind
{
  \bigskip
  \bigskip
  \bigskip
  \begin{center}
    {\LARGE\bf Making Event Study Plots Honest: A Functional Data Approach to Causal Inference}
\end{center}
  \medskip
} \fi

\vspace{-3em}

\spacingset{1.50} 

\begin{abstract}
Event study plots are the centerpiece of Difference-in-Differences (DiD) analysis, but current plotting methods cannot provide honest causal inference when the parallel trends and/or no-anticipation assumptions fail. We introduce a novel functional data approach to DiD that directly enables honest causal inference via event study plots. Our DiD estimator converges to a Gaussian process in the Banach space of continuous functions, enabling powerful simultaneous confidence bands. This theoretical contribution allows us to turn an event study plot into a rigorous honest causal inference tool through equivalence and relevance testing: Honest reference bands can be validated using equivalence testing in the pre-treatment period, and honest causal effects can be tested using relevance testing in the post-treatment period. We demonstrate the performance of our method in simulations and two case studies. 
\end{abstract}


%%%%%%%%%%%%%%%%%%%%%%%%%%%%%%%%%%%%%%%%%%%%
\section{Introduction}\label{sec:intro}
%%%%%%%%%%%%%%%%%%%%%%%%%%%%%%%%%%%%%%%%%%%%

Difference-in-Differences (DiD) is one of the most widely used methods in the social sciences for estimating causal effects of interventions. Its origins can be traced to the pioneering work of Ignaz Semmelweis in the late 1840s \citep[published later in][]{Semmelweis1861}, who studied maternal mortality, and John Snow, who investigated cholera transmission \citep{snow_1855}. Since then, DiD has been refined in theory and widely applied in practice. For recent reviews, see \cite{roth_2023} and \cite{ArkhangelskyImbens2024}.

A central inference tool in DiD applications is the event study plot, which displays point estimates and confidence intervals of event-time coefficients from a dynamic Two-Way Fixed Effects (TWFE) regression (Figure \ref{fig:ESP}). Popularized by \cite{Jacobson_et_al_1993}, the event study plot has become standard in applied economics, particularly in labor, health, and policy research. Its appeal lies in the dual role: providing inference on treatment effects in the post-treatment period while simultaneously assessing the parallel trends and no-anticipation assumptions in the pre-treatment period. This dual purpose explains its ubiquity in empirical DiD research.

However, conventional event study plots---such as the one in Figure \ref{fig:ESP}---suffer from at least three important limitations. 
First, they typically display pointwise confidence intervals that do not account for multiple testing across event times. 
Second, and of particular practical importance, they give the impression that the parallel trends and no-anticipation assumptions can be validated in the pre-treatment period when showing insignificant pre-treatment estimates. However, this is a classical \textit{argument from ignorance} \citep[cf.][]{Walton_1996} as the failure to reject the null hypothesis of no pre-treatment effects (i.e.~absence of differences in time trends or anticipatory effects) does not imply that the parallel trends and no-anticipation assumptions hold.
Third, and equally important from a practical standpoint, honest inference methods—such as those developed by \citet{rambachan_roth_2023}—cannot be integrated into standard event study plots, limiting their usefulness for credible causal inference under violations of the parallel trends and/or no-anticipation assumptions.

%%%%%
\begin{figure}[!htbp]
	\centering
	\begin{subfigure}[b]{0.48\textwidth}
		\centering
\includegraphics[width=\textwidth]{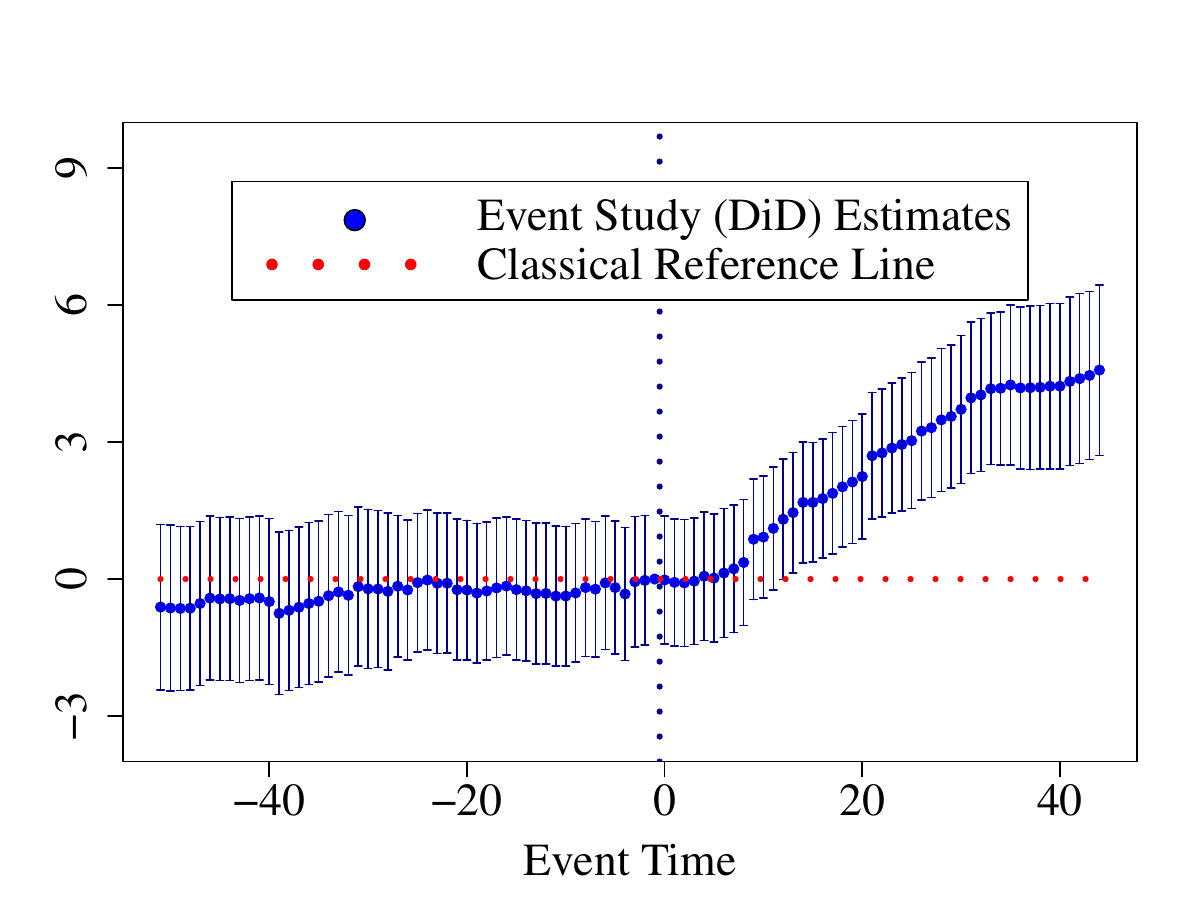} 
		\caption{Classical event study plot with pointwise confidence intervals.}\label{fig:ESP}
	\end{subfigure}\hfill
	\begin{subfigure}[b]{0.48\textwidth} 
		\centering
		\includegraphics[width=\textwidth]{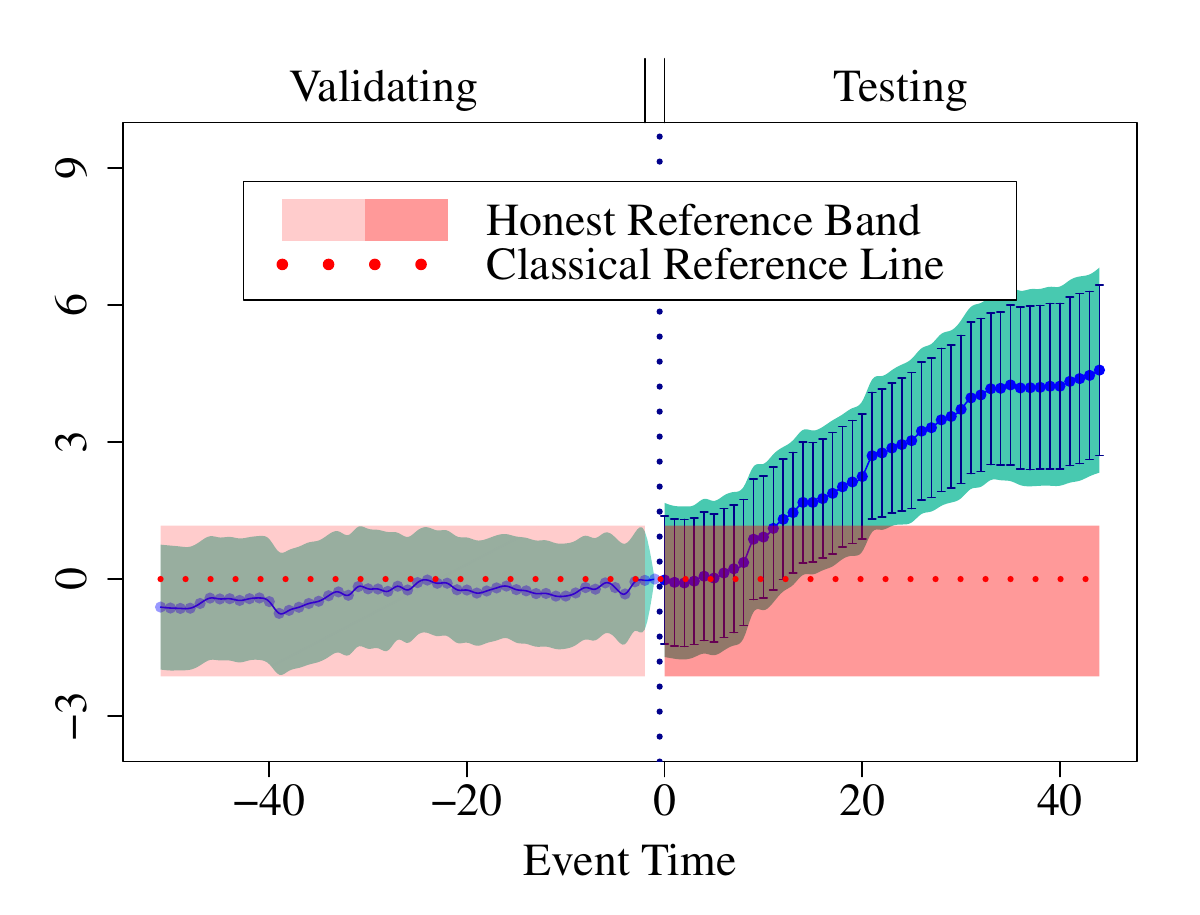} 
		\caption{Honest event study plot with inf- and sup-based simultaneous confidence bands.}\label{fig:ESP_honest}
	\end{subfigure}
\caption{Effects of the realization of working-from-home potential on work-home distance.
Figure \ref{fig:ESP}: Event study plot from \citet{Coskun_2026} with pointwise $95\%$ confidence intervals. 
Figure \ref{fig:ESP_honest}: Infimum-based simultaneous confidence band validates the reference band in the pre-treatment period; supremum-based simultaneous confidence band tests honest post-treatment effects against the validated reference band.}
\label{fig:ESP_compare}
\end{figure}
%%%%%

We address these limitations by introducing a functional-data perspective on DiD. The key idea is to model the underlying time-series processes in continuous time---an assumption already implicit in many empirical DiD studies, where pointwise event-study estimates and confidence intervals are connected by straight lines across event times \citep[cf.][]{moser_voena_2012, card_2024, chen_chen_yang_2025}. Our proposed functional DiD estimator is uniformly consistent and converges to a Gaussian process in the Banach space of continuous functions. Our estimator builds directly on standard panel-data structures and is pointwise identical to the classical panel estimator, making the approach straightforward to implement in empirical applications. We allow for both additional control variables and staggered treatment adoption.

Our main contributions resolve the three limitations of classical event-study plots described above. First, the Gaussian process limit result provides the theoretical foundation for constructing simultaneous confidence bands over the entire continuum of event times. Relative to conventional pointwise inference, this yields a more powerful and statistically credible approach by explicitly accounting for the inherent multiple-testing problem in event-study analyses. Second, infimum-based simultaneous confidence bands enable formal validation of honest reference bands in the pre-treatment period via equivalence testing \citep[cf.][]{Munk_1996, Munk_2004, Rosenbaum_2009, Dette_2018}, allowing researchers to assess the plausibility of the parallel trends and no-anticipation assumptions in a rigorous statistical manner (see Section \ref{sec:ChoosingHRBs} on honest reference bands and Section \ref{sec:validation} on equivalence testing). Third, supremum-based simultaneous confidence bands support honest inference in the post-treatment period through relevance testing, directly integrating existing approaches to honest DiD inference \citep[e.g.,][]{rambachan_roth_2023} into the event-study framework (see Section \ref{sec:honest_testing} on relevance testing). Together, these contributions transform the event-study plot from a descriptive visualization into a rigorous tool for honest causal inference.

Figure \ref{fig:ESP_honest} illustrates how our methodology transforms a classical event study plot (Figure \ref{fig:ESP}) into an honest event study plot (see Section \ref{sec:application_work} for additional details on this application). Unlike the classical reference line, the honest reference band in Figure \ref{fig:ESP_honest} explicitly accounts for potential treatment-anticipation bias.
The selected reference band is validated at the $\alpha=0.05$ level, because the infimum-based $(1-2\alpha)\times 100 \%$ simultaneous confidence band lies entirely \emph{within} the reference band during the pre-anticipation period.
In the post-treatment period $t \in [23,44]$, the causal effect is uniformly and honestly significant at the $\alpha=0.05$ level, because the supremum-based $(1-\alpha)\times 100 \%$ simultaneous confidence band does \emph{not intersect} the statistically validated reference band.

%%% Literature Review (Honest DiD)
The literature on honest and credible causal inference with DiD have seen rapid development in recent years. \citet{rambachan_roth_2023} propose a strategy for robust inference and sensitivity analysis under violations of parallel trends. \cite{freyaldenhoven_hansen_shapiro_2019} introduce a two-stage least squares estimator for the case in which the endogeneity leads to violations of the parallel trends assumption. \cite{abadie_2005} presents a family of semiparametric estimators, considering the case in which the differences in observed covariates create non-parallel outcome dynamics between treated and control groups. \cite{chan_kwok_2022} use a control function approach to adjust for violations of the parallel trends assumption. \cite{Ye_Keele_Hasegawa_Small_2023} propose a partial identification approach requiring two control groups that allow for using a bracketing strategy. \cite{freyaldenhoven_hansen_perez_shapiro_2021} list suggestions to make event study plots more informative and helpful; however, the authors do not provide the corresponding methodological and theoretical foundation. 

%%% Literature on testing parallel trends
The literature on equivalence testing is still sparse---particularly in econometrics. By definition, equivalence testing targets whether a parameter lies \textit{within} a pre-specified region of equivalence, whereas classical null hypothesis significance testing targets whether the parameter lies \textit{outside} the null value \citep{Wellek_2002, Munk_2004}. Hence, the equivalence testing is suitable for addressing the limitation of \textit{argument from ignorance} in classical null hypothesis significance testing, and enables a more rigorous assessment of the parallel trends and no-anticipation assumptions. \cite{Dette_Kokot_2021} develop an equivalence testing approach for functional data. \cite{dette_schumann_2020} contribute an equivalence testing approach to test the parallel trends assumption. Nevertheless, their testing approach cannot be incorporated into event study plots. Our method is similar in spirit to \cite{Fogarty_Small_2014}, who develop a pointwise confidence band for equivalence testing in functional data; however, our infimum-based simultaneous confidence band is substantially tighter, yielding a much more powerful test.

%%% Literature Review (FDA)
The literature on Functional Data Analysis (FDA) offers various methods for constructing simultaneous confidence bands \citep[cf.][]{degras2011simultaneous, pini_vantini_2016, dette_kokot_aue_2020, telschow_schwartzman_2022, Liebl_Reimherr_2023}. In this work, alongside bootstrap approaches, we also use the method of \citet{Liebl_Reimherr_2023} for constructing supremum-based simultaneous confidence bands, as it is computationally efficient and provides practically useful finite sample corrections. Well-known introductory textbooks on FDA include \cite{ramsay_silverman_2005}, \cite{ferraty_vieu_2006}, \cite{hsing_eubank_2015}, and \cite{kokoszka_reimherr_2017}. The application of FDA in econometrics has recently increased \citep[see, for instance,][]{florens2015instrumental, cerovecki2019functional, bugni2021permutation, chang2024autocovariance}. 

%%% Paper outline
The remainder of the paper is organized as follows. Section \ref{sec:theory_methods} presents our core theoretical results and the construction of simultaneous confidence bands. Section \ref{sec:testing} outlines our equivalence and relevance testing frameworks: Section \ref{sec:validation} introduces the use of infimum-based simultaneous bands to validate honest reference bands in the pre-treatment period, and Section \ref{sec:honest_testing} explains how supremum-based simultaneous bands enable honest causal inference in the post-treatment period. Section \ref{sec:sim} reports simulation evidence demonstrating the excellent finite-sample performance of our approach. Section \ref{sec:applications} provides two empirical applications, revisiting \citet{chen_chen_yang_2025} and \citet{lovenheim_willen_2019}. Section \ref{sec:discussion} concludes. All derivations, proofs, algorithms and additional simulations appear in the Online Appendix \citet{fang_liebl_2025}.

%%%%%%%%%%%%%%%%%%%%%%%%%%%%%%%%%%%%%%%%%%%%
\section{Theory and Methods}\label{sec:theory_methods}
%%%%%%%%%%%%%%%%%%%%%%%%%%%%%%%%%%%%%%%%%%%%

%%%%%%%%%%%%%%%%%%%%%%%%%%%%
\subsection{Identification}\label{sec:att_fda}
%%%%%%%%%%%%%%%%%%%%%%%%%%%%

We begin with the non-staggered DiD setting and discuss the extension to staggered adoption in Section \ref{sec:staggered_design}. Let $t=-1$ denote the reference period, conventionally taken as the last observed pre-treatment period. We index the pre-treatment periods by $t \in \{-T_{pre}, \dots, -1\}$ and the post-treatment periods by $t \in \{0, \dots, T_{post}\}$. The treatment occurs in the time period $-1<t<0$.

We model the outcome $Y_{it}\in\mathbb{R}$ of unit $i=1,\dots,n$ at time $t \in \{-T_{pre}, \dots, T_{post}\}$ as a discretely sampled realization 
$$
Y_{it}=Y_i(t)\quad\text{for}\quad t \in \{-T_{pre}, \dots,-1,0,\dots,T_{post}\}.
$$ 
of an underlying stochastic process,
$$
Y_i=\{Y_i(t):t\in[-T_{pre},T_{post}]\},
$$ 
with $[-T_{pre},T_{post}] = [-T_{pre}, -1]\cup(-1,0)\cup[0, T_{post}]$, where the intervals $[-T_{pre}, -1]$ and $[0, T_{post}]$ represent the pre- and post-treatment periods, respectively, over which $Y_i(t)$ is assumed time-continuous. The treatment itself occurs in the time period $(-1,0)$ with a possible discontinuity point at $t\in(-1,0)$ due to the treatment.
The stochastic processes $Y_i$ are assumed to be independently and identically distributed (i.i.d.) across units $i=1,\dots,n$. 

In this paper, we differentiate between an oracle functional data scenario and a practical panel data scenario. The oracle scenario refers to the idealized theoretical setting in which the stochastic processes $Y_i=\{Y_i(t):t\in[-T_{pre},T_{post}]\}$ are fully observable over the entire continuous time span $[-T_{pre},T_{post}]$. In contrast, the practical scenario corresponds to the classical panel data setting, where the processes are only observable at discrete time points, i.e., $Y_{it}=Y_i(t)$ for $t\in\left\{-T_{pre},\dots,T_{post}\right\}$. After establishing our theoretical results in the oracle scenario, we demonstrate that these results can be effectively approximated in practice by interpolating the observed discrete time series, with the resulting approximation error being asymptotically negligible.

Let $D_i\in\{0,1\}$ denote the treatment indicator variable. All units $i$ with $D_i=1$ receive the treatment after the reference period $t=-1$, whereas all units $i$ with $D_i=0$ never receive the treatment. Suppose the observable outcome is given by the following switching equation
$$
Y_{i}(t)=D_iY_{i}(t,1)+(1-D_i)Y_{i}(t,0),
$$ 
where $Y_{i}(t,1)$ and $Y_{i}(t,0)$ represent the potential outcomes for unit $i$ under treatment and non-treatment, respectively. Together with our i.i.d.~assumption, the above switching equation implies the so-called Stable Unit Treatment Value Assumption (SUTVA);  see, for instance, \cite{imbens2015causal}, for a more detailed discussion of SUTVA.

The parameter of interest in this paper is the functional Average Treatment effect on the Treated (ATT) parameter, $\theta_{ATT}=\{\theta_{ATT}(t):t\in[-T_{pre},T_{post}]\}$, defined as
\begin{equation*}
	\theta_{ATT}(t)=\mathbb{E}[Y_{i}(t,1)-Y_{i}(t,0) \mid D_i=1],\quad \text{for all}\quad t\in[-T_{pre},T_{post}]. 
\end{equation*}
The identification of the ATT parameter $\theta_{ATT}(t)$ is not trivial, as $Y_i(t,0)$ is unobservable for the treated units ($D_i=1$) at post-treatment period $t \ge 0$. Hence, to point identify $\theta_{ATT}(t)$ with the estimable DiD parameter 
\begin{align}
\beta(t)
=&\mathbb{E}[Y_{i}(t)-Y_{i}(-1) \mid D_i=1]- \mathbb{E}[Y_{i}(t)-Y_{i}(-1) \mid D_i=0]\label{eq:DiD_parameter_0}\\
\begin{split}\label{eq:DiD_parameter}
=&\underbrace{\mathbb{E}[Y_i(t,1)-Y_i(t,0) \mid D_i=1]}_{\theta_{ATT}(t)} -
\underbrace{\mathbb{E}[Y_i(-1,1)-Y_i(-1,0) \mid D_i=1]}_{\Delta_{TA}=\theta_{ATT}(-1)} + \\
&\underbrace{\mathbb{E}[Y_{i}(t,0)-Y_{i}(-1,0) \mid D_i=1]- \mathbb{E}[Y_{i}(t,0)-Y_{i}(-1,0) \mid D_i=0]}_{\Delta_{DT}(t)},
\end{split}
\end{align} 
one requires the following three identification assumptions (see Online Appendix \ref{app:did_att_relation} for the derivation of \eqref{eq:DiD_parameter}): 
\begin{description}[itemsep=2pt, parsep=0pt, topsep=5pt]
\item[Assumption I (No Anticipation):] $\mathbb{E}[Y_i(-1,1)-Y_i(-1,0) \mid D_i=1]=\Delta_{TA}= 0$ 
\item[Assumption II (Parallel Trends):]\ \\ 
\hspace*{-1cm}$\mathbb{E}[Y_{i}(t,0)\! -\! Y_{i}(-1,0) \!\mid\! D_i\!=\!1] \!-\! \mathbb{E}[Y_{i}(t,0) \!-\! Y_{i}(-1,0) \!\mid\! D_i\!=\!0]=\Delta_{DT}(t)\!=\!0,\, \forall t\in[-T_{pre},T_{post}]$ 
\item[Assumption III (Overlap):] The exists some $\epsilon>0$ such that $\epsilon<P(D_i=1)< 1-\epsilon$.
\end{description}

The estimable functional DiD parameter $\beta(t)$ in \eqref{eq:DiD_parameter_0} represents the expected change in the observable outcome between time $t$ and reference time $t=-1$ for treated and control units. Notably, $\beta(t) \to 0$ as $t \to -1$ by construction. The term $\Delta_{TA} = \theta_{ATT}(-1)$ in \eqref{eq:DiD_parameter} captures the Treatment Anticipation (TA) bias, which arises when Assumption I (No Anticipation) is violated. The term $\Delta_{DT}(t)$ in \eqref{eq:DiD_parameter} captures the Differential Trend (DT) bias at time $t$, which arises when Assumption II (Parallel Trends) is violated. Assumption III (Overlap) ensures the existence of comparable units in both the treated and control groups. 

In a sense, Assumption III (Overlap) is less contentious, as the presence of comparable treated and control units is mostly a prerequisite for attempting causal inference in the first place. However, violations of Assumption I (No Anticipation) and/or Assumption II (Parallel Trends) are more problematic, since estimating the DiD parameter $\beta(t)$ under such violations results in a biased estimate of the causal ATT parameter $\theta_{ATT}(t)$, leading to invalid inference and potentially harmful conclusions. 

Assumption I (No Anticipation) and Assumption II (Parallel Trends) are generally not testable using post-treatment data, $t \in [0, T_{post}]$. However, a key strength of event study plots (Figure \ref{fig:ESP}) is that they allow for an assessment of the plausibility of these two assumptions using the pre-treatment data, $t \in [-T_{pre}, -1]$. If the two assumptions hold, we expect $\beta(t) = \theta_{ATT}(t) = 0$ for $t \in [-T_{pre}, -1]$. Consequently, when the pre-treatment confidence intervals in a conventional event-study plot include zero, as in Figure \ref{fig:ESP}, researchers often interpret this as evidence supporting these identifying assumptions. Such an interpretation, however, is statistically invalid. It constitutes an \emph{argument from ignorance}: failing to reject the null hypothesis of no pre-treatment effects does not provide evidence that the identifying assumptions hold. In Section \ref{sec:validation}, we overcome this limitation by replacing classical null hypothesis significance testing with equivalence testing, which enables a statistically rigorous assessment of the plausibility of the parallel trends and no-anticipation assumptions in the pre-treatment period.

Assumption I (No Anticipation) is violated when treated units ($D_i = 1$) begin responding to the treatment before its formal implementation---a phenomenon known as treatment anticipation. Let $t_A \in [-T_{pre}, -1]$ denote the time point after which treated units start responding. Assumption I holds if $t_A = -1$, corresponding to the last observed pre-treatment period. If instead $t_A < -1$, Assumption I is then violated, as treated units exhibit anticipatory responses prior to the last observed pre-treatment period. This bias is identifiable by the estimable DiD parameter $\beta(t)$ for any time point $t$ in the pre-anticipation period, $t \in [-T_{pre}, t_A]$, where $\theta_{ATT}(t) = 0$, and thus,
\begin{equation}
	\beta(t) = -\Delta_{TA} \quad\text{for all}\quad t\in[-T_{pre},t_A].\label{eq:AnticipationBias}
\end{equation}

Assumption II (Parallel Trends) is violated when, in the absence of treatment, treated units ($D_i = 1$) would have followed a different trend from control units ($D_i = 0$) during the post-treatment period. This bias is likewise identifiable by estimating the DiD parameter $\beta(t)$ for time points in the pre-treatment period, $t \in [-T_{pre}, -1]$, where $\theta_{ATT}(t) = 0$, and thus,
\begin{equation}
\beta(t) = \Delta_{DT}(t) \quad\text{for all}\quad t \in [-T_{pre}, -1]. \label{eq:DTBias}
\end{equation}

If both assumptions, Assumption I (No Anticipation) and Assumption II (Parallel Trends), are violated, we can identify the resulting joint bias by estimating the DiD parameter $\beta(t)$ for time points in the pre-anticipation period, 
\begin{equation}
\beta(t) = -\Delta_{TA} + \Delta_{DT}(t) \quad\text{for all}\quad t \in [-T_{pre}, t_A]. \label{eq:TADTBias}
\end{equation}

Honest inference with DiD refers to drawing inference about the ATT parameter $\theta_{ATT}(t)$ while explicitly accounting for potential violations of the identifying assumptions. One plausible approach is to extrapolate the biases identified in equations \eqref{eq:AnticipationBias}–\eqref{eq:TADTBias} into post-treatment period ($t \in [0, T_{post}]$). In Section \ref{sec:testing}, we clarify how we operationalize this approach for honest inference using event study plots.

%%%%%%%%%%%%%%%%%%%%%%%%%%%%%%%%%%%%%%%%%%%%%%%%%%%%%%%%%%%%%%%%%%%%%%%
\subsection{Models and Estimators}\label{sec:fct_model} 
%%%%%%%%%%%%%%%%%%%%%%%%%%%%%%%%%%%%%%%%%%%%%%%%%%%%%%%%%%%%%%%%%%%%%%%

Under the Panel Data (PD) point of view, the DiD parameter vector $\beta^{PD}\in\mathbb{R}^{T_{pre}+T_{post}+1}$, 
$$
\beta^{PD}=(\beta^{PD}_{-T_{pre}},\dots,\beta^{PD}_{-2},\beta^{PD}_{-1},\beta^{PD}_0,\dots,\beta^{PD}_{T_{post}})^{\top},
$$ 
where $\beta^{PD}_{-1}=0$ by definition (see \eqref{eq:DiD_parameter_0}), is typically estimated using TWFE regression model 
\begin{equation}\label{eq:pd_twfe_reg_main}
    Y_{it}=\sum_{\substack{s=-T_{pre}\\s\neq -1}}^{T_{post}}\beta^{PD}_s D_{its} + u_{it}
    \quad\text{with}\quad
    u_{it} = \lambda_i + \phi_t + \varepsilon_{it},
\end{equation}
for $t\in\{-T_{pre}, \dots,T_{post}\}$ and $i=1,\dots,n$, 
where $D_{its} := D_i \times \mathbbm{1}_{\{t=s\}}$, with $\mathbbm{1}_{\{\cdot\}}$ representing the indicator function. The real error term $u_{it}$ consists of the unobserved individual and temporal fixed effects, $\lambda_i$ and $\phi_t$, and a mean zero error component $\varepsilon_{it}$ with  $\mathbb{E}[\varepsilon_{it} \mid D_i]=0$. The unobserved effects (e.g.~hardly measurable factors like a person's intrinsic motivation) may influence the decision to sort oneself into $(D_i=1)$ or out of $(D_i=0)$ the treatment, resulting in an issue of endogeneity, $\mathbb{E}[u_{it} \mid D_i]\neq 0$ for $t\in\{-T_{pre},\dots,T_{post}\}$. To address this endogeneity issue, one typically uses a two-way transformation to partial out the problematic unobserved individual, $\lambda_i$, and temporal, $\phi_t$, components,
\begin{equation*}
    \ddot{Y}_{it}^{PD}=\sum_{\substack{s=-T_{pre}\\s\neq -1}}^{T_{post}} \beta^{PD}_s \ddot{D}^{PD}_{its} + \ddot{\varepsilon}^{PD}_{it},
\end{equation*}
where $\ddot{Y}^{PD}_{it}=Y_{it}-\frac{1}{T}\sum_{t=-T_{pre}}^{T_{post}}Y_{it}-\frac{1}{n}\sum_{i=1}^{n}Y_{it}+\frac{1}{nT}\sum_{i=1}^{n}\sum_{t=-T_{pre}}^{T_{post}}Y_{it}$ with $T:=T_{pre}+T_{post}+1$; $\ddot{D}^{PD}_{its}$ and $\ddot{\varepsilon}^{PD}_{it}$ are defined analogously. The assumption $\mathbb{E}[\varepsilon_{it} \mid D_i]=0$ implies that $\mathbb{E}[\ddot{\varepsilon}_{it}^{PD} \mid D_i]=0$, which suggests that the TWFE estimator, 
\begin{equation}\label{eq:beta_hat_transf_main}
    \widehat{\beta}^{PD}_n = 
    \left(\frac{1}{nT}\sum_{i=1}^n \sum_{t=-T_{pre}}^{T_{post}} \ddot{D}^{PD}_{it} \ddot{D}^{PD^\top}_{it} \right)^{-1}\left(\frac{1}{nT}\sum_{i=1}^n \sum_{t=-T_{pre}}^{T_{post}} \ddot{D}_{it}^{PD} \ddot{Y}_{it}^{PD}\right),
\end{equation}
is an unbiased estimator, where 
$\ddot{D}_{it}^{PD}=(\ddot{D}^{PD}_{it,-T_{pre}}, \dots ,\ddot{D}^{PD}_{it,-2},\ddot{D}^{PD}_{it,0},\dots,\ddot{D}^{PD}_{it,T_{post}})^{\top}$. To turn the $(T_{pre}+T_{post})$ dimensional estimator $\widehat{\beta}^{PD}_n$ in \eqref{eq:beta_hat_transf_main} into a full-dimensional unbiased estimator of $\beta^{PD}\in\mathbb{R}^{T_{pre}+T_{post}+1}$, we augment it by adding the zero element $\widehat{\beta}^{PD}_{n,-1}=0$ to the reference time period $t=-1$,
\begin{equation}\label{eq:beta_hat_transf_main_1}
\widehat{\beta}^{PD}_n\equiv(\widehat{\beta}^{PD}_{n,-T_{pre}},\dots,\widehat{\beta}^{PD}_{n,-2},0,\widehat{\beta}^{PD}_{n,0},\dots,\widehat{\beta}^{PD}_{n,T_{post}})^{\top}.
\end{equation}

The $(T_{pre}+T_{post}+1)$ dimensional estimator $\widehat{\beta}^{PD}_n$ in \eqref{eq:beta_hat_transf_main_1} is typically used to construct traditional event study plots and is $\sqrt{n}$-consistent and asymptotically normal under standard assumptions. The $\sqrt{n}$-consistency may not be directly obvious from the definition of $\widehat{\beta}^{PD}_n$ in \eqref{eq:beta_hat_transf_main}, which involves both data dimensions $n$ and $T$. However, this becomes obvious from the functional data estimators below, which are equivalent to $\widehat{\beta}^{PD}_{n,t}$ pointwise for all time points $t\in\{-T_{pre},\dots,T_{post}\}$; see Theorem \ref{thm:beta_hat_equiv}. 

From a functional data perspective, the DiD parameter $\beta(t)$ as defined in \eqref{eq:DiD_parameter_0} can be estimated using a TWFE function-on-scalar regression model
\begin{align}\label{eq:fct_data_reg_simple}
	Y_i(t)=\beta(t)D_i+u_{i}(t) \quad\text{with}\quad
	u_{i}(t) = \lambda_i+\phi(t)+\varepsilon_{i}(t),
\end{align}
for $t\in[-T_{pre},T_{post}]$ and $i=1,\dots,n$, where $\beta(t)$ is the function-valued DiD parameter with $\beta(-1)=0$, $D_i\in\{0,1\}$ is the scalar-valued predictor (treatment indicator), and $Y_i(t)$ is the oracle form of the outcome for unit $i$ at time $t$. The function-valued error term $u_{i}(t)$ consists of an unobserved individual $\lambda_i\in\mathbb{R}$ and a function-valued temporal $\phi(t)$ component, as well as a functional mean zero error component $\varepsilon_{i}(t)$ with $\mathbb{E}[\varepsilon_{i}(t) \mid D_i]=0$ for $t\in[-T_{pre},T_{post}]$, which is allowed to be heteroskedastic; i.e.~$\mathbb{E}[\varepsilon(s)\varepsilon(t)|D]=C_{\varepsilon|D}(s,t|D)$, where generally $C_{\varepsilon|D}(s,t|D=d_1)\neq C_{\varepsilon|D}(s,t|D=d_2)$ for $d_1\neq d_2$. The fixed effects components are allowed to correlate with the treatment indicator $D_{i}$, thereby leading to an endogeneity issue $\mathbb{E}[u_{i}(t) \mid D_i]\neq 0$ for $t\in[-T_{pre},T_{post}]$. 

Using a careful adoption of the two-way panel data transformation to the case of functional data, allows us to partial out the unobserved individual and temporal components in the error term, leading to the transformed function-on-scalar model
\begin{equation}\label{eq:fct_data_reg_final}
	\ddot{Y}_i(t)=\underbrace{\left( \beta(t)-\frac{1}{T_{post}+T_{pre}}\int_{-T_{pre}}^{T_{post}} \beta(s)\dd s \right)}_{\gamma(t)}\dot{D}_i + \ddot{\varepsilon}_i(t)
\end{equation}
with $\beta(-1)=0$, where 
\begin{comment}
$\ddot{Y}_i(t) = Y_i(t)-\widetilde{Y_i} -\frac{1}{n} \sum_{i=1}^n Y_i(t)+\frac{1}{n}\sum_{i=1}^n \widetilde{Y}_i$, 
$\ddot{\varepsilon}_i(t) = \varepsilon_i(t)- \widetilde{\varepsilon}_i -\frac{1}{n}\sum_{i=1}^n \varepsilon_i(t)+\frac{1}{n}\sum_{i=1}^n \widetilde{\varepsilon}_i$, 
with 
$\widetilde{Y}_i = \frac{1}{T_{post}+T_{pre}} \int_{-T_{pre}}^{T_{post}} Y_i(t)\dd t$ and 
$\widetilde{\varepsilon}_i=\frac{1}{T_{post}+T_{pre}} \int_{-T_{pre}}^{T_{post}} \varepsilon_i(t)\dd t$,
and 
$\dot{D}_i=D_i-\frac{1}{n}\sum_{i=1}^n D_i$.
\end{comment}
\begin{align*}
	\ddot{Y}_i(t)&=\dot{Y}_i(t)-\frac{1}{T_{post}+T_{pre}} \int_{-T_{pre}}^{T_{post}} Y_i(s)\dd s + \frac{1}{n}\sum_{i=1}^n \frac{1}{T_{post}+T_{pre}}\int_{-T_{pre}}^{T_{post}} Y_i(s)\dd s,\\
	\ddot{\varepsilon}_i(t)&=\dot{\varepsilon}_i(t)-\frac{1}{T_{post}+T_{pre}} \int_{-T_{pre}}^{T_{post}} \varepsilon_i(s)\dd s + \frac{1}{n}\sum_{i=1}^n \frac{1}{T_{post}+T_{pre}}\int_{-T_{pre}}^{T_{post}} \varepsilon_i(s)\dd s,
	%\text{and}\quad\dot{D}_i&=D_i-\frac{1}{n}\sum_{i=1}^n D_i.
\end{align*}
$\dot{D}_i=D_i-n^{-1}\sum_{i=1}^n D_i$, 
$\dot{Y}_i(t)=Y_i(t)- n^{-1}\sum_{i=1}^n Y_i(t)$, and 
$\dot{\varepsilon}_i(t)=\varepsilon_i(t)- n^{-1}\sum_{i=1}^n \varepsilon_i(t)$; the derivation of Model \eqref{eq:fct_data_reg_final} is included in Online Appendix \ref{app:SIMPLIFIED_FCT_MODEL}. Consequently, in \eqref{eq:fct_data_reg_final}, we obtain the DiD parameter as 
\begin{equation}\label{eq:fct_data_reg_final_beta}
\beta(t)=\gamma(t)-\gamma(-1).
\end{equation}

The following theorem shows that the functional DiD parameter $\beta(t)$ in \eqref{eq:fct_data_reg_final_beta} is pointwise equivalent to the panel data DiD parameter $\beta^{PD}_t$ in \eqref{eq:pd_twfe_reg_main} for every $t\in\{-T_{pre},\dots,T_{post}\}$.

\begin{theorem}\label{thm:beta_equiv}
The functional DiD parameter $\beta=\{\gamma(t)-\gamma(-1):t\in[-T_{pre},T_{post}]\}$ in \eqref{eq:fct_data_reg_simple} is pointwise equivalent to the panel data DiD parameter $\beta^{PD}$ in \eqref{eq:pd_twfe_reg_main}, i.e.
\begin{equation*}%\label{eq:beta_equiv}
\beta(t)=\beta^{PD}_t \quad\text{for every}\quad t\in\{-T_{pre},\dots,T_{post}\},\quad\text{with}\quad \beta(-1)=\beta^{PD}_{-1}=0. 
\end{equation*}
\end{theorem}

While Theorem \ref{thm:beta_equiv} is relatively straightforward to prove, its result is nonetheless non-trivial. This is because the functional data model in \eqref{eq:fct_data_reg_simple} is not a direct analogue of the panel data specification in \eqref{eq:pd_twfe_reg_main}. In particular, the predictor variables $D_{its}$ in \eqref{eq:pd_twfe_reg_main} take nonzero values only at isolated time points ($t = s$) for treated units ($D_i = 1$). A naive functional translation of this setup would involve sets of Lebesgue measure zero, over which integration yields zero---posing a fundamental issue in the functional framework.

The least squares estimator $\widehat{\gamma}_n=\{\widehat{\gamma}_n(t):t \in [-T_{pre},T_{post}]\}$ for $\gamma=\{\gamma(t):t \in [-T_{pre},T_{post}]\}$ is
\begin{equation}\label{eq:GammaHat}
	\widehat{\gamma}_n(t) = \left(\frac{1}{n}\sum_{i=1}^n \dot{D}_i^2\right)^{-1}\left(\frac{1}{n}\sum_{i=1}^n \dot{D}_i \ddot{Y}_i(t)\right),\quad t\in[-T_{pre},T_{post}].
\end{equation}
By \eqref{eq:fct_data_reg_final_beta}, the estimator $\widehat{\beta}_n=\{\widehat{\beta}_n(t):t \in [-T_{pre},T_{post}]\}$ for the functional DiD parameter $\beta=\{\beta(t):t \in [-T_{pre},T_{post}]\}$ is therefore, in principle, defined as
\begin{align}\label{eq:BetaHat_0}
\widehat{\beta}_n(t)
	&=\widehat{\gamma}_n(t) -\widehat{\gamma}_n(-1)
	= \left(\frac{1}{n}\sum_{i=1}^n \dot{D}_i^2\right)^{-1}\left(\frac{1}{n}\sum_{i=1}^n \dot{D}_i (\ddot{Y}_i(t) - \ddot{Y}_i(-1))\right).
\end{align}
Note, however, that the terms $\ddot{Y}_i(t)$ and $\ddot{Y}_i(-1)$ involve equal integration operations, which cancel out due to the subtraction in \eqref{eq:BetaHat_0}, yielding the following simplified expression for the functional oracle estimator of the DiD parameter $\beta(t)$:
\begin{align}\label{eq:BetaHat}
\widehat{\beta}_n(t)
 = \left(\frac{1}{n}\sum_{i=1}^n \dot{D}_i^2\right)^{-1}\left(\frac{1}{n}\sum_{i=1}^n \dot{D}_i (\dot{Y}_i(t) - \dot{Y}_i(-1))\right),
\end{align}
where $\dot{Y}_i(t)=Y_i(t)-n^{-1}\sum_{i=1}^n Y_i(t)$ and $\widehat{\beta}_n(-1)=\beta(-1)=0$. 

The following theorem shows that the functional oracle estimator $\widehat{\beta}_n(t)$ in \eqref{eq:BetaHat} is pointwise equivalent to the panel data estimator $\widehat{\beta}_{n,t}^{PD}$ in \eqref{eq:beta_hat_transf_main_1} for every $t\in\{-T_{pre},\dots,T_{post}\}$, which implies that they have pointwise the same distributional properties.

\begin{theorem}\label{thm:beta_hat_equiv}
	The functional DiD estimator $\widehat{\beta}_n(t)$ in \eqref{eq:BetaHat} is pointwise equivalent to the panel data DiD estimator $\widehat{\beta}_{n,t}^{PD}$ in \eqref{eq:beta_hat_transf_main_1}, i.e.
	\begin{equation*}
		\widehat{\beta}_n(t)=\widehat{\beta}^{PD}_{n,t}\quad\text{for every}\quad t\in\{-T_{pre},\dots,T_{post}\}\quad\text{with}\quad \widehat{\beta}_n(-1)=\widehat{\beta}^{PD}_{n,-1}=0. 
	\end{equation*}
\end{theorem}

The results of Theorems \ref{thm:beta_equiv} and \ref{thm:beta_hat_equiv} imply that the functional DiD parameter $\beta(t)$ defined in \eqref{eq:fct_data_reg_simple} and its estimator $\widehat{\beta}_n(t)$ given in \eqref{eq:BetaHat} are pointwise equivalent to the panel data DiD parameter $\beta^{PD}_t$ in \eqref{eq:pd_twfe_reg_main} and its estimator $\widehat{\beta}^{PD}_{n,t}$ in \eqref{eq:beta_hat_transf_main_1} for each $t\in\{-T_{pre}, \dots, T_{post}\}$. This equivalence implies that the functional DiD parameter $\beta(t)$ can be consistently estimated at each time point $t\in\{-T_{pre}, \dots, T_{post}\}$ using standard panel data softwares that implement TWFE estimators.

To estimate the trajectory of the functional DiD parameter $\beta(t)$ in the practical scenario, we propose interpolating the pointwise estimates 
$\widehat{\beta}_n(t)=\widehat{\beta}^{PD}_{n,t}$ using a natural cubic spline interpolation. Notably, in this paper, we allow for a potential discontinuity of the functional DiD parameter at the treatment-point within the interval $(-1,0)$, which is completely unobservable in the practical scenario. Hence, we restrict our focus to the trajectory of the functional parameter over two disjoint intervals: the pre-treatment period $t \in [-T_{pre}, -1]$ and the post-treatment period $t \in [0, T_{post}]$. Let 
\begin{equation}\label{eq:BetaHatHat}
\doublewidehat{\beta}_n=\{\doublewidehat{\beta}_n(t): t\in [-T_{pre}, -1] \cup [0, T_{post}]\},
\end{equation}
denote the natural cubic spline interpolation of the points $(t,\widehat{\beta}_{n}(t))$ for $t\in\{-T_{pre},\dots,-1\}$ and $t\in\{0,\dots,T_{post}\}$ separately. Such interpolations can be conveniently obtained using implementations in standard statistical software routines (e.g. \texttt{splinefun()} in \textsf{R} or \texttt{CubicSpline()} in Python). By the definition of natural cubic splines, the interpolation estimator $\doublewidehat{\beta}_n$ in \eqref{eq:BetaHatHat} satisfies the following properties \citep[cf.][]{Atkinson_1968,de_Boor_2001}:
\begin{enumerate}[label={(\arabic*)}, align=left, leftmargin=*]
    \setlength\itemsep{-0.2em}
    \item $\doublewidehat{\beta}_n(t) = \widehat{\beta}_n(t)$ at every $t\in\{-T_{pre},\dots,T_{post}\}$, 
    \item $\doublewidehat{\beta}_n$ is, at most, cubic on each subinterval $[t-1,t]$ with $t=\{-T_{pre}+1,\dots,-1,1,T_{post}\}$,
    \item $\doublewidehat{\beta}_{n}{}^{\!\!\!\!''}(-T_{pre})=\doublewidehat{\beta}_{n}{}^{\!\!\!\!''}(-1)=\doublewidehat{\beta}_{n}{}^{\!\!\!\!''}(0)=\doublewidehat{\beta}_{n}{}^{\!\!\!\!''}(T_{post})=0$, where $\doublewidehat{\beta}_{n}{}^{\!\!\!\!''}(t)$ denotes the second derivative of $\doublewidehat{\beta}_n$ at $t$, and
    \item $\doublewidehat{\beta}_n$ is, at least, two times continuously differentiable over $[-T_{pre},-1]$ and $[0,T_{post}]$, i.e.~$\doublewidehat{\beta}_n \in C^2([-T_{pre},-1] \cup [0,T_{post}])$.
\end{enumerate}

%%%%%%%%%%%%%%%%%%%%%%%%%%%%%%%%%%%%%%%%%%%%%%%%%%%%
\subsection{Asymptotic Theory}\label{sec:theory}
%%%%%%%%%%%%%%%%%%%%%%%%%%%%%%%%%%%%%%%%%%%%%%%%%%%%
In Sections \ref{sec:pw_inf} and \ref{sec:unif_inf}, we will present theoretical results for our functional DiD estimator $\widehat{\beta}_n$ in \eqref{eq:BetaHat} under the oracle setting. In Section \ref{sec:spline_intrpl}, we will extend these results to the practical setting, involving the spline-based interpolation estimator $\doublewidehat{\beta}_n$ in \eqref{eq:BetaHatHat}. The following list summarizes the assumptions under which these theorems are developed.
%%%

\begin{enumerate}[label=(\arabic*),align=left, leftmargin=*] 
    \setlength\itemsep{-0.2em}
    \item \label{assumption: data_str} $(Y_1,D_1),\dots,(Y_n,D_n)\stackrel{\text{i.i.d.}}{\sim}(Y,D)$\\[-6ex]
        \begin{enumerate}[label=(\arabic{enumi}.\alph*),align=left, leftmargin=*]
            \setlength\itemsep{-0.2em}
            \item \label{assumption: data_str a} Oracle scenario: $Y_1(t),\dots,Y_n(t)$ are completely observable for all $t\in[-T_{pre},T_{post}]$. \\
			Asymptotic scenario: $n\to\infty$. 
            \item \label{assumption: data_str b} Practical scenario: $Y_1(t),\dots,Y_n(t)$ are only observable at $T:=T_{pre}+T_{post}+1$ time points 
			$t\in\{-T_{pre},\dots,-1,0,\dots,T_{post}\}$.\\
			Asymptotic scenario: $n \to \infty$ and $T\equiv T_{n}\to\infty$ as $n\to\infty$ such that $T_{pre,n}/T_{post,n}\to c$ as $n\to\infty$ for a constant $0<c<\infty$.  
        \end{enumerate}
    \item \label{assumption: moments} Moments:\\[-6ex]
        \begin{enumerate}[label=(\arabic{enumi}.\alph*),align=left, leftmargin=*] 
		\setlength\itemsep{-0.2em}
        \item \label{assumption: moments a} $\mathbb{E}[Y(t)^4] < \infty$  for all $t\in[-T_{pre},T_{post}]$ and $\mathbb{E}[D^4] < \infty$.        
        \item \label{assumption: moments b} $\mathbb{E}[\sup_{t \in (-T_{pre},-1) \cup (0, T_{post})} Y'(t)^2] < \infty$, where $Y'(t)$ denotes the first derivative.
        \item \label{assumption: moments c} $\mathbb{E}[\sup_{t \in [-T_{pre},-1]\cup[0,T_{post}]} Y(t)^2] < \infty$.
        \end{enumerate}
    \item \label{assumption: smoothness} Smoothness:\\[-6ex] 
        \begin{enumerate}[label=(\arabic{enumi}.\alph*),align=left, leftmargin=*] 
		\setlength\itemsep{-0.2em}
        \item \label{assumption: smooth a} $\phi,\varepsilon_i \in C^2([-T_{pre},-1]\cup[0,T_{post}])$.
        \item \label{assumption: smooth b} $\beta \in C^2([-T_{pre},-1]\cup[0,T_{post}])$.
        \item \label{assumption: smooth c} $C_{\beta} \in C^2([-T_{pre},-1]\cup[0,T_{post}])^2$, where $C_{\beta}$ denotes the asymptotic covariance of $\widehat{\beta}_n$.
        \item \label{assumption: smooth d} $\exists\;\beta^\star \in C^2([-1,-\frac{1}{T_{pre}}] \cup [0,1])$ s.t.~$\beta(t)=\beta^\star(t^\star(t))$, where $t^\star(t)=\frac{t}{T_{pre}}$ for $t\in[-T_{pre},-1]$ and $t^\star(t)=\frac{t}{T_{post}}$ for $t\in[0,T_{post}]$.
        \item \label{assumption: smooth e} $\exists C_{\beta}^\star \!\in\! C^2([-1,-\frac{1}{T_{pre}}] \cup [0,1])^2$ s.t.~$C_{\beta}(s,t)\!\!=\!\!C_{\beta}^\star(s^\star(s), t^\star(t))$, where $s^\star(s)$ is defined as $t^\star(t)$ above.
    \end{enumerate}
\end{enumerate}

Assumptions \ref{assumption: data_str a} and \ref{assumption: data_str b} describe the two data scenarios—oracle and practical—that we use to develop our theorems. Remember that the error term $\varepsilon(t)$ contained in $Y(t)$ (see \eqref{eq:fct_data_reg_simple}) is allowed to be heteroskedastic.  
%%%
Assumption \ref{assumption: moments a} states typical moment conditions used for showing (multivariate) asymptotic normality. Assumptions \ref{assumption: moments b} and \ref{assumption: moments c} impose further moment conditions allowing us to develop simultaneous inference results across $t\in[-T_{pre}, -1] \cup[0,T_{post}]$.
%%%
Assumptions \ref{assumption: smooth a}--\ref{assumption: smooth c} impose that all functional components in Model \eqref{eq:fct_data_reg_simple} and the asymptotic covariance function of $\widehat{\beta}_n$ are continuous at least up to their second derivative over $t\in[-T_{pre}, -1] \cup[0,T_{post}]$. In particular, the functional DiD parameter $\beta(t)$ is assumed to evolve smoothly over the pre- and post-treatment periods, reflecting the typically gradual changes in underlying factors. We do not impose smoothness assumption over $t\in (-1,0)$ by accounting for a plausible discontinuity at the treatment-point within this interval. Assumptions \ref{assumption: smooth d} and \ref{assumption: smooth e} are standard assumptions in the literature on nonparametric time series \citep[cf.][Ch.~6.2.9]{Fan_Yao_2003} and serve as a technical tool to asymptotically quantify interpolation errors.

%%%%%%%%%%%%%%%%%%%%%%%%%%%%%%%%%%%%%%%%%%%%%%%%%%%%%%%%
\subsubsection{Pointwise Inference}\label{sec:pw_inf}
%%%%%%%%%%%%%%%%%%%%%%%%%%%%%%%%%%%%%%%%%%%%%%%%%%%%%%%%

This subsection is used to introduce the pointwise inference for our functional oracle estimator $\widehat{\beta}_n(t)$ in \eqref{eq:BetaHat}, which, by Theorem \ref{thm:beta_hat_equiv} directly applies also for the panel data estimator $\widehat{\beta}^{PD}_{n,t}$ in \eqref{eq:beta_hat_transf_main_1} at each $t\in\{-T_{pre},\dots,T_{post}\}$.

\begin{theorem}[Pointwise Asymptotic Normality]\label{thm:pw_normality} Under Assumptions \ref{assumption: data_str a} and \ref{assumption: moments a}, we have for the functional oracle estimator $\widehat{\beta}_n(t)$ in \eqref{eq:BetaHat} and the panel data estimator $\widehat{\beta}^{PD}_{n,t}$ in \eqref{eq:beta_hat_transf_main_1} that pointwise\\[-7ex]

\spacingset{1.1}
\begin{align*}
    \begin{array}{lll}
        \sqrt{n}\left(\widehat{\beta}_n(t)-\beta(t)\right)&\stackrel{d}{\to} \mathcal{N}\left(0, C_{\beta}(t,t)\right)\quad&\text{for each}\quad t\in[-T_{pre},T_{post}]\\
        \sqrt{n}\left(\widehat{\beta}^{PD}_{n,t}-\beta(t)\right)&\stackrel{d}{\to} \mathcal{N}\left(0, C_{\beta}(t,t)\right)\quad&\text{for each}\quad t\in\{-T_{pre},\dots,T_{post}\},
    \end{array}
\end{align*}
\spacingset{1.5}
as $n\to\infty$, where 
$C_{\beta}(t,t)=\mathbb{E}[\dot{D}^2 (\dot{\varepsilon}(t)-\dot{\varepsilon}(-1))^2] \mathbb{E}[\dot{D}^2]^{-2}$.
\end{theorem}

Theorem \ref{thm:pw_normality} shows that, for each time point $t\in[-T_{pre},T_{post}]$, $\widehat{\beta}_n(t)$ is pointwise asymptotically normal under standard assumptions. At the reference time point $t=-1$, we have that $\widehat{\beta}_n(-1)=\beta(-1)=0$ such that the variance function $C_{\beta}(t,t)$ is zero at $t=-1$. 

Typically, the variance function $C_{\beta}(t,t)$ is unknown and has to be estimated from the data using 
\begin{equation*}
	\widehat{C}_{\beta,n}(t,t) = \left(\frac{1}{n}\sum_{i=1}^n \dot{D}_i^2 \left( \Delta_0 \dot{Y}_i(t) \right)^2 \right)\left(\frac{1}{n} \sum_{i=1}^n \dot{D}_i^2\right)^{-2},
\end{equation*}
where $\Delta_0 \dot{Y}_i(t)=(\dot{Y}_i(t)-\dot{Y}_i(-1))-\widehat{\beta}_n(t)\dot{D}_i$ with $\widehat{\beta}_n(t)$ defined in \eqref{eq:BetaHat} and $\dot{D}_i=D_i-n^{-1}\sum_{i=1}^n D_i$. It follows from standard arguments (Slutsky's Theorem) that the pointwise asymptotic normality of Theorem \ref{thm:pw_normality} still holds after replacing $C_{\beta}(t,t)$ with $\widehat{C}_{\beta,n}(t,t)$ such that
\begin{equation}\label{eq:T_statistic}
	T_{\beta,n}(t)=\frac{\sqrt{n}(\widehat{\beta}_n(t)-\beta(t))}{\sqrt{\widehat{C}_{\beta,n}(t,t)}}\overset{a}\sim \mathcal{N}(0,1),
\quad\text{as}\quad n\to\infty,
\end{equation}
pointwise for each $t\in[-T_{pre},T_{post}]$. Inverting the above statistic, allows us to construct the pointwise $(1-\alpha)\times 100\%$  confidence band for $\beta(t)$, with significance level $\alpha\in(0,1)$, 
\begin{equation}\label{eq:pw_ci_t}
	\widehat{\operatorname{CI}}_{1-\alpha}(t)=\left[\widehat{\beta}_n(t)\pm t_{1-\alpha/2, \text{df}}\sqrt{\widehat{C}_{\beta,n}(t,t)/n}\right]
\end{equation}
for each $t\in[-T_{pre}, T_{post}]$, where the critical value $t_{1-\alpha/2, \text{df}}$ is given by the ($1-\alpha/2)$-quantile of the $t$-distribution with $\text{df}=n-1$ degrees of freedom, allowing for a finite sample correction and being asymptotically equivalent to the $(1-\alpha/2)$-quantile of the standard normal distribution. 

The pointwise confidence intervals $\widehat{\operatorname{CI}}_{1-\alpha}(t)$ in \eqref{eq:pw_ci_t} are typically displayed in conventional event study plots for selected pre- and post-treatment event times $t\in\{-T_{pre},\dots,T_{post}\}$ \citep[cf.][]{lovenheim_willen_2019,chen_chen_yang_2025}. However, such pointwise intervals do not account for the simultaneous testing problem that arises when interpreting event study plots across multiple event times.

%%%%%%%%%%%%%%%%%%%%%%%%%%%%%%%%%%%%%%%%%%%%%%%%%%%%%%%%%%%%
\subsubsection{Simultaneous Inference (Oracle Scenario)}\label{sec:unif_inf}
%%%%%%%%%%%%%%%%%%%%%%%%%%%%%%%%%%%%%%%%%%%%%%%%%%%%%%%%%%%%

The following theorem generalizes the pointwise asymptotic normality of Theorem \ref{thm:pw_normality} to a uniform asymptotic normality result over the time span $[-T_{pre},-1] \cup [0,T_{post}]$. 

\begin{theorem}[Uniform Asymptotic Normality of the Oracle Estimator (\ref{eq:BetaHat})]\label{thm:gauss}
	Under Assumptions \ref{assumption: data_str a}, \ref{assumption: moments a}, \ref{assumption: moments b}, \ref{assumption: smooth a} and \ref{assumption: smooth b}, we have for the functional oracle estimator $\widehat{\beta}_n(t)$ in \eqref{eq:BetaHat} that 
	\begin{equation*}
		\sqrt{n}\left(\widehat{\beta}_n-\beta\right) \stackrel{d}{\to} \mathcal{GP}(0,C_{\beta}),\quad\text{as}\quad n\to\infty,
	\end{equation*}
	in $C([-T_{pre},-1] \cup [0,T_{post}])$, i.e.~uniformly for all $s,t\in[-T_{pre},-1] \cup [0,T_{post}]$, where $C_{\beta}=\{C_{\beta}(s,t): s,t \in [-T_{pre},-1] \cup [0,T_{post}]\}$ with $C_{\beta}(s, t) =  \mathbb{E}[\dot{D}^2 \left(\dot{\varepsilon}(s)-\dot{\varepsilon}(-1)\right) \left(\dot{\varepsilon}(t)-\dot{\varepsilon}(-1)\right)]\mathbb{E}[\dot{D}^2]^{-2}$.
\end{theorem}

Theorem \ref{thm:gauss} demonstrates that the functional estimator $\widehat{\beta}_n$ is asymptotically a Gaussian process in the Banach space of continuous functions, $C([-T_{pre},-1] \cup [0,T_{post}])$. In contrast to Theorem \ref{thm:pw_normality}, Theorem \ref{thm:gauss} considers the total covariance structure $C_{\beta}(s,t)=\lim_{n\to\infty} n \operatorname{Cov}(\widehat{\beta}_n(s),\widehat{\beta}_n(t))$ of $\widehat{\beta}_n$. At the reference time point $t=-1$, we have that $\widehat{\beta}_n(-1)=\beta(-1)=0$ such that the covariance function $C_{\beta}(s, t)$ is zero at $s=t=-1$. Considering the covariance function $C_{\beta}(s,t)$ instead of only the variance function $C_{\beta}(t,t)$ is important for constructing simultaneous confidence bands for the functional DiD parameter $\beta$. Usually, the covariance function $C_{\beta}(s,t)$ is unknown and has to be estimated from the data using 

\begin{equation}\label{eq:cov_est}
\widehat{C}_{\beta,n}(s,t)=\left(\frac{1}{n}\sum_{i=1}^n \dot{D}_i^2 \left(\Delta_0\dot{Y}_i(s)\right) \left(\Delta_0\dot{Y}_i(t)\right)\right)\left(\frac{1}{n}\sum_{i=1}^n \dot{D}_i^2\right)^{-2},
\end{equation}
where 
$\Delta_0\dot{Y}_i(t) = (\dot{Y}_i(t)-\dot{Y}_i(-1)) -\widehat{\beta}_n(t)\dot{D}_i$ with $\widehat{\beta}_n(t)$ defined in \eqref{eq:BetaHat} and $\dot{D}_i=D_i-n^{-1}\sum_{i=1}^n D_i$. Theorem \ref{thm:consistency_cov} establishes that the covariance estimator $\widehat{C}_{\beta,n}(s,t)$ in \eqref{eq:cov_est} is uniformly consistent to $C_{\beta}(s,t)$ over all $s,t \in [-T_{pre},-1] \cup [0,T_{post}]$. 

\begin{theorem}[Uniform Consistency of Empirical Covariance]\label{thm:consistency_cov}
Under Assumptions \ref{assumption: data_str a}, \ref{assumption: moments}, \ref{assumption: smooth a}, \ref{assumption: smooth b} and \ref{assumption: smooth c}, concerning the covariance function $C_{\beta}$, we have
\begin{equation*}
	\sup_{s,t\in[-T_{pre},-1] \cup [0,T_{post}]}\left| \widehat{C}_{\beta,n}(s,t) - C_{\beta}(s,t) \right| \stackrel{a.s.}{\to} 0 \,,
\end{equation*}
where $\widehat{C}_{\beta,n}=\{\widehat{C}_{\beta,n}(s,t):s,t\in[-T_{pre},-1] \cup [0,T_{post}]\}$ with $\widehat{C}_{\beta,n}(s,t)$ defined in \eqref{eq:cov_est}.
\end{theorem}

Analogous to the pointwise confidence band in \eqref{eq:pw_ci_t}, Theorems \ref{thm:gauss} and \ref{thm:consistency_cov} provide the foundation for constructing Simultaneous Confidence Bands (SCBs). In this work, we consider a two-sided supremum-based $(1-\alpha)\times 100\%$ SCB for relevance testing in the post-treatment period $[0,T_{post}]$ (Section \ref{sec:honest_testing}) and two one-sided infimum-based $(1-\alpha)\times 100\%$ SCBs for equivalence testing in the pre-treatment period $[-T_{pre},-1]$ (Section \ref{sec:validation}),\\[-7ex]

\spacingset{1.2}
\begin{align}
\widehat{\operatorname{SCB}}^{\sup}_{1-\alpha}(t) & = \left[\widehat{\beta}_n(t) \pm u^{\sup}_{1-\alpha/2}\sqrt{\widehat{C}_{\beta,n}(t, t) /n}\right]&& \text{for all} \quad t\in[0,T_{post}],\label{eq:supSCB}\\
\widehat{\operatorname{SCB}}^{\inf,+}_{1-\alpha}(t) & = \left(\left.-\infty, \; \widehat{\beta}_n(t) + u^{\,\inf}_{1-\alpha}\sqrt{\widehat{C}_{\beta,n}(t, t) /n}\right]\right. && \text{for all} \quad t\in[-T_{pre},-1],\label{eq:infSCB_plus}\\
\widehat{\operatorname{SCB}}^{\inf,-}_{1-\alpha}(t) & = \left[\left.\widehat{\beta}_n(t) - u^{\,\inf}_{1-\alpha}\sqrt{\widehat{C}_{\beta,n}(t, t) /n},\; \infty\right)\right. && \text{for all} \quad t\in[-T_{pre},-1],\label{eq:infSCB_minus}
\end{align}
\spacingset{1.5}

\noindent where 
$u_{1-\alpha/2}^{\sup}$ and 
$u_{1-\alpha}^{\,\inf}$ denote the $(1-\alpha/2)$ and $(1-\alpha)$-quantiles of the distributions of the supremum- and infimum-based test statistics $\sup_{t\in[0,T_{post}]}T_{\beta,n}(t)$
and 
$\inf_{t\in[-T_{pre},-1]}T_{\beta,n}(t)$
with 
$T_{\beta,n}(t)=\sqrt{n}(\widehat{\beta}_n(t)-\beta(t))/\left(\widehat{C}_{\beta,n}(t,t)\right)^{1/2}$. Theorem \ref{thm:SCBs} establishes that the simultaneous non-coverage probabilities of the SCBs in \eqref{eq:supSCB}--\eqref{eq:infSCB_minus} converge to their nominal levels, rendering them suitable for relevance and equivalence testing; see Sections \ref{sec:honest_testing} and \ref{sec:validation}.

\begin{theorem}[Non-Coverage Probabilities of SCBs]\label{thm:SCBs}
    Under Assumptions \ref{assumption: data_str a}, \ref{assumption: moments}, \ref{assumption: smooth a}, \ref{assumption: smooth b} and \ref{assumption: smooth c}, we have
\begin{align*}
\lim_{n\to\infty}\quad & P\left(\beta(t)\not\in\widehat{\operatorname{SCB}}^{\sup}_{1-\alpha}(t)\;\text{for at least one}\; t \in [0,T_{post}] \right) = \alpha,\\ 
\lim_{n\to\infty}\quad & P\left(\beta(t)\not\in\widehat{\operatorname{SCB}}^{\inf,\square}_{1-\alpha}(t)\;\text{for all}\; t \in [-T_{pre},-1] \right) = \alpha\quad\text{for each}\quad\square\in\{+,-\}.
\end{align*}
\end{theorem}

There exists no exact formula for the distribution of suprema and infima of general Gaussian processes \citep[see, for instance,][p.~5]{Adler_1990}. Therefore, we follow the usual approach and estimate the quantiles $u_{1-\alpha/2}^{\sup}$ and $u_{1-\alpha}^{\inf}$ using bootstrap  methods. For the case of the supremum-based quantile, we also make use of the Kac-Rice formula approach considered in \cite{Liebl_Reimherr_2023}. Justification for the validity of these approaches follows from the fact that $\sqrt{n}(\widehat{\beta}_n-\beta)$ is asymptotically distributed as a Gaussian process in the Banach space $C([-T_{pre},-1] \cup [0,T_{post}])$ (Theorem \ref{thm:gauss}) and consistency arguments along the lines of, for instance, those in \cite{Belloni_et_al_2017} and Theorem 3.2 in \citet{Liebl_Reimherr_2023}.

%%%%%%%%%%%%%%%%%%%%%%%%%%%%%%%%%%%%%%%%%%%%%%%%%%%%%%%%%%%%%%%%%%%%%%%%%%%%%%
\subsubsection{Simultaneous Inference (Practical Scenario)}\label{sec:spline_intrpl}
%%%%%%%%%%%%%%%%%%%%%%%%%%%%%%%%%%%%%%%%%%%%%%%%%%%%%%%%%%%%%%%%%%%%%%%%%%%%%%

This subsection extends our theoretical results from the oracle setting to the practical panel data context, where the functional DiD parameter $\beta(t)$ is approximated by the interpolation estimator $\doublewidehat{\beta}_n(t)$ defined in \eqref{eq:BetaHatHat}. Theorem \ref{thm:intrpl_consistency} establishes the uniform consistency of $\doublewidehat{\beta}_n(t)$ over the time span $t\in[-T_{pre},-1] \cup [0,T_{post}]$, which is fundamental to deriving the uniform asymptotic normality result in Theorem \ref{thm:intrpl_dist} for the interpolation estimator $\doublewidehat{\beta}_n(t)$.

\begin{theorem}[Uniform Consistency of Interpolation Estimator (\ref{eq:BetaHatHat})]\label{thm:intrpl_consistency}
Under Assumptions \ref{assumption: data_str b}, \ref{assumption: smooth b} and \ref{assumption: smooth d}, we have that
\begin{equation*}
\sup_{t \in [-T_{pre},-1] \cup [0,T_{post}]} \left|\doublewidehat{\beta}_n(t)-\beta(t)\right| \le c_1\max_{t\in\{-T_{pre},\dots,T_{post}\}} \left|\widehat{\beta}_n(t) - \beta(t)\right| +  \frac{c_2 K_{\beta^*}}{T},
\end{equation*}
where $0<c_1,c_2,K_{\beta^*}<\infty$ are finite constants with $K_{\beta^*}=\sup_{t^\star\in(-1,-\frac{1}{T_{pre}}) \cup(0,1)}|\beta^{*\prime}(t^\star)|$. 
\end{theorem}

Theorem \ref{thm:intrpl_consistency} shows that the supremum norm of the estimation error of $\doublewidehat{\beta}_n(t)$ decomposes into two components. The first captures the maximum estimation error of $\widehat{\beta}_n(t)$ at the observable time points $t\in\{-T_{pre},\dots,T_{post}\}$, which, by Theorem \ref{thm:gauss}, is uniformly $O_P(1/\sqrt{n})$. The second term reflects the interpolation error, which decays uniformly at rate $1/T$. For more complex functional DiD parameters $\beta(t)$---i.e., those with larger smoothness constants $K_{\beta^*}$---a larger $T$ may be necessary to sufficiently reduce interpolation bias; see also the simulation evidence in Section \ref{SIM:intrpl_error}. The following theorem generalizes Theorem \ref{thm:gauss} to the case of the practically relevant spline-based interpolation estimator $\doublewidehat{\beta}_n(t)$ in \eqref{eq:BetaHatHat}.

\begin{theorem}[Uniform Asymptotic Normality of the Interpolation Estimator (\ref{eq:BetaHatHat})]\label{thm:intrpl_dist}
Under Assumptions \ref{assumption: data_str b}, \ref{assumption: moments}, \ref{assumption: smooth a}, \ref{assumption: smooth b}, \ref{assumption: smooth d} and the additional condition $\sqrt{n}/T \to 0$, we have for the interpolation estimator $\doublewidehat{\beta}_n(t)$ in \eqref{eq:BetaHatHat} that
    \begin{equation*}
		\sqrt{n}\left(\doublewidehat{\beta}_n-\beta\right) \stackrel{d}{\to} \mathcal{GP}(0,C_{\beta}),\quad\text{as}\quad n\to\infty,
	\end{equation*}
	in $C([-T_{pre},-1]\cup[0,T_{post}])$, i.e.~uniformly for all $s,t\in[-T_{pre},-1]\cup[0,T_{post}]$.
\end{theorem}

Theorem \ref{thm:intrpl_dist} shows that the spline-based interpolation estimator $\doublewidehat{\beta}_n$ has the same uniform asymptotic normal distribution as the functional oracle DiD estimator $\widehat{\beta}_n$ in \eqref{eq:BetaHat}. The additional condition $\sqrt{n}/T \to 0$ is typically not restrictive as $\sqrt{n}$ goes to infinity at a relatively slow rate. 

Similarly to the interpolation estimator $\doublewidehat{\beta}_n(t)$ defined in \eqref{eq:BetaHatHat}, we propose interpolating the pointwise covariance estimates $\widehat{C}_{\beta, n}(s,t)$ at $s,t\in\{-T_{pre}, \dots, T_{post}\}$ using tensor-product natural cubic spline interpolation \citep[cf.][Ch.~17]{de_Boor_2001},
\begin{equation}\label{eq:CovHatHat}
    \doublewidehatCL{C}_{\beta,n}= \{ \doublewidehatCL{C}_{\beta, n}(s,t): s,t \in[-T_{pre},-1]\cup[0,T_{post}] \} \quad\text{with}\quad \doublewidehatCL{C}_{\beta,n}(s,t)=\sum_{l=1}^{N_s}\sum_{k=1}^{N_t} c_{lk}B_{l}(s)B_{k}(t),
\end{equation}
where $B_l(s)$ and $B_k(t)$ denote the $l$-th and $k$-th spline basis function in dimensions $s$ and $t$, and the coefficients $c_{lk}$ are determined algebraically by the pointwise covariance estimates $\widehat{C}_{\beta, n}(s,t)$ at $s,t\in\{-T_{pre},\dots, T_{post}\}$. Such interpolations can be easily done using the \texttt{cov\_spline()} function in our \textsf{R}-package \texttt{fdid}. In parallel to Theorem \ref{thm:intrpl_consistency}, Theorem \ref{thm:cov_intrpl_consistency} establishes the uniform consistency of $\doublewidehatCL{C}_{\beta,n}(s,t)$ over the time surface $s,t \in [-T_{pre},-1]\cup[0,T_{post}]$.

\begin{theorem}[Uniform Consistency of Interpolation Estimator (\ref{eq:CovHatHat})]\label{thm:cov_intrpl_consistency}Under Assumptions \ref{assumption: data_str b}, \ref{assumption: smooth c} and \ref{assumption: smooth e}, we have that 
\begin{equation*}
\sup_{s,t \in [-T_{pre},-1]\cup[0,T_{post}]} \left|\doublewidehatCL{C}_{\beta,n}(s,t)-C_{\beta}(s,t)\right| \le c_3\max_{s,t\in\{-T_{pre},\dots,T_{post}\}} \left|\widehat{C}_{\beta,n}(s,t) - C_{\beta}(s,t)\right| +  \frac{c_4 K_{C_{\beta}^\star}}{T},
\end{equation*}
where $0<c_3,c_4,K_{C_{\beta}^\star}<\infty$ are finite constants.
\end{theorem}

Theorems \ref{thm:intrpl_dist} and \ref{thm:cov_intrpl_consistency} justify the use of the interpolation estimators $\doublewidehat{\beta}_n$ and $\doublewidehatCL{C}_{\beta,n}$ in place of $\widehat{\beta}_n$ and $\widehat{C}_{\beta,n}$ in \eqref{eq:supSCB}--\eqref{eq:infSCB_minus}, leading to the interpolation-based versions of SCBs in \eqref{eq:supSCB_practical}--\eqref{eq:infSCB_minus_practical},\\[-7ex]

\spacingset{1.2}
\begin{align}
    \doublewidehatCL{\operatorname{SCB}}{}^{\sup}_{1-\alpha}(t) & = \left[\doublewidehat{\beta}_n(t) \pm u^{\sup}_{1-\alpha/2}\sqrt{\doublewidehatCL{C}_{\beta,n}(t, t) /n}\,\right]\quad \quad \; \text{for all}\quad t\in[0,T_{post}]\label{eq:supSCB_practical}\\
    \doublewidehatCL{\operatorname{SCB}}{}^{\inf,+}_{1-\alpha}(t) & = \left(\left.-\infty,\; \doublewidehat{\beta}_n(t) + u^{\,\inf}_{1-\alpha}\sqrt{\doublewidehatCL{C}_{\beta,n}(t, t) /n}\,\right]\right.\, \text{for all}\quad t\in[-T_{pre},-1]\label{eq:infSCB_plus_practical}\\
    \doublewidehatCL{\operatorname{SCB}}{}^{\inf,-}_{1-\alpha}(t) & = \left[\left.\doublewidehat{\beta}_n(t) - u^{\,\inf}_{1-\alpha}\sqrt{\doublewidehatCL{C}_{\beta,n}(t, t) /n},\; \infty\right)\right.\quad\, \text{for all}\quad t\in[-T_{pre},-1]\label{eq:infSCB_minus_practical}. 
\end{align}
\spacingset{1.5}

\noindent Theorems \ref{thm:intrpl_dist} and \ref{thm:cov_intrpl_consistency}, together with Continuous Mapping Theorem \citep[Theorem 12 in][]{pollard_1984} imply that also the interpolation-based SCBs in \eqref{eq:supSCB_practical}--\eqref{eq:infSCB_minus_practical} have asymptotically correct uniform non-coverage probabilities (Corollary \ref{cor:SCBs_Interpolation}), rendering them suitable for relevance and equivalence testing; see Sections \ref{sec:honest_testing} and \ref{sec:validation}.

\begin{corollary}[Non-Coverage Probabilities of Interpolation SCBs]\label{cor:SCBs_Interpolation}
    Under Assumptions \ref{assumption: data_str b}, \ref{assumption: moments}, \ref{assumption: smoothness} and the additional condition $\sqrt{n}/T \to 0$, we have
\begin{align*}
\lim_{n\to\infty}\quad & P\left(\beta(t)\not\in \doublewidehatCL{\operatorname{SCB}}{}^{\sup}_{1-\alpha}(t)\;\text{for at least one}\; t \in [0,T_{post}] \right) = \alpha,\\ 
\lim_{n\to\infty}\quad & P\left(\beta(t)\not\in\doublewidehatCL{\operatorname{SCB}}{}^{\inf,\square}_{1-\alpha}(t)\;\text{for all}\; t \in [-T_{pre},-1] \right) = \alpha\quad\text{for each}\quad\square\in\{+,-\}
\end{align*}
\end{corollary}

%%%%%%%%%%%%%%%%%%%%%%%%%%
\subsection{Extension: Control Variables}\label{sec:covariates}
%%%%%%%%%%%%%%%%%%%%%%%%%%
In practice, additional covariates are often used to enhance the efficiency of estimation and inference and to make the identifying assumptions more credible. In DiD models \citep[cf.][]{Lechner_2011, dette_schumann_2020,callaway_santanna_2021}, it is common to include time-invariant, pre-treatment covariates
$W_i = (W_{i1}, W_{i2}, \dots, W_{ik})^{\top} \in \mathbb{R}^k$,
which are not affected by $D_i$; i.e.~which satisfy the exogeneity condition
%\begin{equation*}
$W_i(1)=W_i(0)=W_i$, 
%\end{equation*}
where $W_i(1)$ and $W_i(0)$ denote the potential outcomes of the covariates for treated and untreated units, respectively. Using such covariates $W_i$, we can replace the identification Assumptions I-III by the following weaker, conditional versions:
\begin{description}[itemsep=2pt, parsep=0pt, topsep=5pt]
\item[Assumption $\!\text{I}^*\!\!$ (Conditional \hspace{-0.5em} No \hspace{-0.5em} Anticipation):] $\!\!\!\mathbb{E}[Y_i(-1,1)\!-\!Y_i(-1,0) | D_i\!\!=\!\!1,\!W_i]\!\!=\!\!\Delta_{TA}(W_i\!)\!\!=\!\!0$  
\item[Assumption $\text{II}^*$ (Conditional Parallel Trends):]\ \\ 
\hspace*{-1cm}$\mathbb{E}[Y_{i}(t,0)\!-\!Y_{i}(-1,0) |  D_i\!=\!1,W_i]\!-\! \mathbb{E}[Y_{i}(t,0)\!-\!Y_{i}(-1,0) | D_i\!\!=\!\!0,W_i]\!=\!\Delta_{DT}(W_i)\!=\!0,\;\forall\; t\!\in\![-T_{pre},T_{post}]$ 
\item[Assumption $\text{III}^*$ (Conditional Overlap):]  
        $\exists ~ \epsilon>0$ such that $\epsilon<P(D_i=1\mid W_i)< 1-\epsilon$ 
\end{description}
Model \eqref{eq:fct_data_reg_simple} can then be extended as
\begin{equation}\label{eq:fct_data_CoV_reg}
Y_i(t) = \beta(t) D_i + W_i^{\top} \xi(t) + u_i(t)
\quad \text{with} \quad
u_i(t) = \lambda_i + \phi(t) + \varepsilon_i(t),
\end{equation}
for $t\in[-T_{pre}, T_{post}]$ and $i=1,\dots, n$, where the coefficient $\xi(t)=(\xi_1(t), \xi_2(t), \dots, \xi_k(t))^{\top}$ and $\xi_1(t), \xi_2(t), \dots, \xi_k(t) \in C^2([-T_{pre}, -1] \cup [0, T_{post}])$ are time-varying, and all other model components remain as in Model \eqref{eq:fct_data_reg_simple}. The DiD parameter $\beta$ in Model \eqref{eq:fct_data_CoV_reg} is equivalent to the DiD parameter $\beta$ in the following function-on-scalar model:
\begin{equation}\label{eq:fct_on_scalar_CoV_model}
	\ddot{Y}_i(t)=\underbrace{\left( \beta(t)-\frac{1}{T_{post}+T_{pre}}\int_{-T_{pre}}^{T_{post}} \beta(s)\dd s \right)}_{\gamma(t)}\dot{D}_i +\dot{W}_i^{\top}\breve{\xi}(t) + \ddot{\varepsilon}_i(t)
\end{equation}
with $\beta(-1) = 0$, where 
$\dot{W}_i=(\dot{W}_{i1}, \dot{W}_{i2}, \dots, \dot{W}_{ik})^{\top}$ 
and 
$\breve{\xi}(t)=(\breve\xi_1(t), \breve\xi_2(t), \dots, \breve\xi_k(t))^{\top}$ 
with 
$\dot{W}_{ij} = W_{ij}-n^{-1}\sum_{i=1}^nW_{ij}$, 
and 
$\breve\xi_j(t) = \xi_j(t)-(T_{pre}+T_{post})^{-1}\int_{-T_{pre}}^{T_{post}}\xi_j(s)\dd s$
for 
$j=1,\dots,k$, 
and where 
$\ddot{Y}_i(t)$ and $\ddot{\varepsilon}_i(t)$ are defined as in Model \eqref{eq:fct_data_reg_final}. 

Using the Frisch-Waugh-Lovell (FWL) Theorem (\citealp{frisch_waugh_1933, lovell_1963}) and $\beta(t)=\gamma(t)-\gamma(-1)$, allows estimating $\beta(t)$ in \eqref{eq:fct_data_CoV_reg} by
\begin{align}\label{eq:BetaHat_Covariate}
    \widehat{\beta}_n^{FWL}(t)
    = \left( \frac{1}{n}\sum_{i=1}^n \widetilde{D}_i^2 \right)^{-1} \left( \frac{1}{n}\sum_{i=1}^n \widetilde{D}_i (\widetilde{Y}_i(t)-\widetilde{Y}_i(-1)) \right), 
\end{align}
where $\widetilde{D}_i=\sum_{j=1}^n L_{ij} \dot{D}_j$ and $\widetilde{Y}_i(t)=\sum_{j=1}^n L_{ij} \dot{Y}_j(t)$ and $L_{ij}$ is the $(i,j)$-th entry of ($n \times n$) matrix $L=I-\dot{W}(\dot{W}^{\top}\dot{W})^{-1}\dot{W}^{\top}$ with $I$ as an ($n \times n$) identity matrix and $\dot{W}=(\dot{W}_1, \dot{W}_2, \dots, \dot{W}_n)^{\top}$ as an ($n \times k$) matrix. The asymptotic covariance function of $\widehat{\beta}_n^{FWL}(t)$ can be estimated using
\begin{equation}\label{eq:cov_est_Covariate}
\widehat{C}_{\beta,n}^{FWL}(s,t)= \left(\frac{1}{n}\sum_{i=1}^n \widetilde{D}_i^2 \left(\Delta_0 \widetilde{Y}_i(s)\right)\left( \Delta_0 \widetilde{Y}_i(t)\right) \right)\left(\frac{1}{n}\sum_{i=1}^n \widetilde{D}_i^2\right)^{-2},
\end{equation}
where $\Delta_0 \widetilde{Y}_i(t) = (\widetilde{Y}_i(t)-\widetilde{Y}_i(-1))-\widehat{\beta}^{FWL}_n(t)\widetilde{D}_i$ with $\widehat\beta_n^{FWL}(t)$ defined in \eqref{eq:BetaHat_Covariate}. In the practical panel data scenario, where estimation can only be done at the discretely observable time points $s,t\in\{-T_{pre},\dots,T_{post}\}$, we need to use interpolation estimators $\doublewidehat{\beta}_{n}{}^{\!\!\!\!FWL}$ and $\doublewidehatCL{C}_{\beta,n}{}^{\!\!\!\!\!\!\!\!\!FWL}$ defined exactly in parallel to $\doublewidehat{\beta}_n$ in \eqref{eq:BetaHatHat} and $\doublewidehatCL{C}_{\beta,n}$ in \eqref{eq:CovHatHat}. Our theoretical results in Section \ref{sec:theory} apply to the estimators in this section. For more details, see Online Appendix \ref{app:cov_est_Covariate}.

%%%%%%%%%%%%%%%%%%%%%%%%%%%%%%%%%%%%%%%%%%%%%%%%%%%%%%%%
\subsection{Extension: Staggered Treatment Adoption}\label{sec:staggered_design}
%%%%%%%%%%%%%%%%%%%%%%%%%%%%%%%%%%%%%%%%%%%%%%%%%%%%%%%%

Recent studies show that the TWFE approach faces a ``negative weighting'' problem when treatment effects are heterogeneous across groups with staggered treatment timing \citep{de_Chaisemartin_Haultfoeuille_2020, callaway_santanna_2021, Goodman-Bacon_2021, sun_abraham_2021, de_Chaisemartin_Haultfoeuille_2023, Borusyak_2024}.

Following the convention in Section \ref{sec:att_fda}, we take the last pre-treatment period as the reference period. Each treatment group $g$ is indexed by its reference period, $g\in\mathcal{G} \subset \{-T_{pre}, \dots,-1, 0, \dots, T_{post}\}$, where $\mathcal{G}$ contains at least group $g=-1$, which serves as the time-normalizing group. For group $g=-1$, the first observed post-treatment period is $0$, and for group $g=-3$, if it exists, it is period $-2$. Each unit $i$ belongs to exactly one group, denoted by $G_i \in \bar{\mathcal{G}}=\mathcal{G} \cup \{\infty\}$, where $G_i=\infty$ indicates the never-treated units (control group). The total number of groups $|\bar{\mathcal{G}}|=|\mathcal{G}|+1$ is finite and fixed such that asymptotically $|\mathcal{G}|+1\ll T=|\{-T_{pre}, \dots, -1, 0, \dots, T_{post}\}|$; our asymptotic assumptions in Section \ref{sec:theory} apply then to each group. The group-specific ATT parameter is defined as
$$
\theta_{ATT,g}(e)=\mathbb{E}[Y_i(g+e+1,g)-Y_i(g+e+1, \infty) \mid G_i=g] \quad \text{for all} \quad e\in[-T_{pre}-g-1, T_{post}-g-1],
$$
where $Y_i(g+e+1,g)$ denotes the potential outcome of unit $i$ in group $G_i=g\in\mathcal{G}$ under treatment at event time $e$ measured relative to the first post-treatment period $g+1$, and $Y_i(g+e+1, \infty)$ denotes the potential outcome of unit $i$ under non-treatment. Analogous to the non-staggered DiD setting in Section \ref{sec:att_fda}, we index the pre-treatment period by $e \in [-T_{pre}-g-1, -1]$ and the post-treatment period by $e \in [0, T_{post}-g-1]$. We assume that $-T_{pre}$ and $T_{post}$ are chosen such that there is a non-empty common time window $[-T_{pre,A}, T_{post,A}]:=\bigcap_{g\in\mathcal{G}}[-T_{pre}-g-1, T_{post}-g-1]\cap[-T_{pre}^\infty, T_{post}^\infty]$ observed for all groups, where $[-T_{pre}^\infty, T_{post}^\infty]$ denotes the observed time span of the never-treated units. In staggered design, we replace the identification Assumptions I-III by the following groupwise versions: %identification

\begin{description}[itemsep=2pt, parsep=0pt, topsep=5pt]
\item[Assumption $\!\text{I}^{**}\!\!\!$ (Groupwise \hspace{-0.5em} No \hspace{-0.5em} Anticipation):]\hspace{-0.8em} $\mathbb{E}[Y_i(g,g)\!\!-\!\!Y_i(g, \infty) | G_i\!\!=\!\!g]\!\!=\!\!\Delta_{TA}^g\!\!=\!\!0,\forall g \in \mathcal{G}$
\item[Assumption $\!\text{II}^{**}\!$ (Groupwise Parallel Trends):]\ \\
\hspace*{-1cm}$\mathbb{E}[Y_{i}(g+e+1,\infty)-Y_{i}(g,\infty) \mid G_i=g] - \mathbb{E}[Y_{i}(g+e+1,\infty)-Y_{i}(g,\infty) \mid G_i=\infty]=\Delta_{DT}^g(e)=0, \\ 
\forall\, g\in \mathcal{G} \text{ and } e\in[-T_{pre}-g-1,T_{post}-g-1]$ 
\item[Assumption $\!\text{III}^{**}\!$ (Groupwise Overlap):]\ 
$\exists\,\epsilon>0$ s.t.~$\epsilon<P(G_i=g)< 1-\epsilon,\,\forall g \in \bar{\mathcal{G}}$
\end{description}

Under Assumptions $\text{I}^{**}$-$\text{III}^{**}$, $\theta_{ATT,g}(e)$ can be identified as $\theta_{ATT,g}(e)=\beta_g(e)$, where
\begin{equation}\label{eq:staggered_DiD_ggs_1}
    \beta_g(e)
    =\mathbb{E}[Y_i(g+e+1)-Y_i(g) \mid G_i=g]- \mathbb{E}[Y_{i}(g+e+1)-Y_{i}(g) \mid G_i=\infty],
\end{equation} 
which follows arguments analogous to those in Section \ref{sec:att_fda}; see also Online Appendix \ref{app:did_att_relation}. Often, the group-specific ATT parameters $\theta_{ATT,g}(e)$ are of main interest. One can, however, also consider the aggregated staggered-design ATT parameter such as, for instance, 
\begin{equation}\label{eq:staggered_DiD_parameter}
    \theta_{ATT,A}(e)=\beta_{A}(e)=\sum_{g\in\mathcal{G}} w_g\beta_g(e) \quad \text{for all} \quad e\in[-T_{pre,A}, T_{post,A}],
\end{equation}
where $w_g$ are non-negative group-specific weights satisfying $\sum_{g\in\mathcal{G}} w_g=1$. Restricting the focus to the common time span $[-T_{pre,A}, T_{post,A}]$ avoids hardly interpretable artifacts in $\beta_{A}(e)$ due to groups $g$ being unobserved for some periods $e$. Under Assumptions $\!\text{I}^{**}\!$ (Groupwise No Anticipation) and Assumptions $\!\text{II}^{**}\!$ (Groupwise Parallel Trends), $\beta_{A}(e)=0$ for all $e\leq -1$.

To estimate $\beta_{A}(e)$ in \eqref{eq:staggered_DiD_parameter}, we first estimate the group-specific DiD parameters $\beta_g(e)$ in \eqref{eq:staggered_DiD_ggs_1} using the following group-specific version of the functional DiD estimator in \eqref{eq:BetaHat}, 
\begin{align*}
\widehat{\beta}_{g}(e)
 = \left(\frac{1}{\bar{n}_{g}}\sum_{i\in\bar{\mathcal{I}}_g} \dot{D}_{gi}^2\right)^{-1}\left(\frac{1}{\bar{n}_g}\sum_{i\in\bar{\mathcal{I}}_g} \dot{D}_{gi} (\dot{Y}_i(g+e+1) - \dot{Y}_i(g))\right),
\end{align*}
where 
$\bar{\mathcal{I}}_{g}=\{i=1,\dots,n: G_i=g \text{ or } G_i=\infty\}$ denotes the index set of units in treatment group $g$ and the control group, $\bar{n}_g=|\bar{\mathcal{I}}_g|$ is the corresponding size of the index set, and
$\dot{D}_{gi}=D_{gi}-\bar{n}_g^{-1}\sum_{i\in\bar{\mathcal{I}}_g} D_{gi}$ with 
$D_{gi}=\mathbbm{1}_{\{G_i=g\}}$ indicating the treatment status for units $i\in\bar{\mathcal{I}}_g$. Then, $\beta_{A}(e)$ can be estimated as the weighted average 
$$
\widehat{\beta}_{A}(e)=\sum_{g\in\mathcal{G}}\widehat{\omega}_{g}\widehat{\beta}_{g}(e),\quad e\in[-T_{pre,A}, T_{post,A}],
$$
where $\widehat{\omega}_{g}=n_g/n_{\mathcal{G}}$ with $n_g=\sum_{i=1}^n\mathbbm{1}_{\{G_i=g\}}$ and $n_{\mathcal{G}}=\sum_{i=1}^n\mathbbm{1}_{\{G_i\in\mathcal{G}\}}$ denotes the proportion of treated units in group $g$ among all treated units; see \citet{callaway_santanna_2021} for alternative weighting schemes. The asymptotic covariance function of $\widehat{\beta}_{A}(e)$ is 
\begin{equation}\label{eq:cov_stagger}
    C_{\beta_{A}}(e_1,e_2)=\lim_{n\to\infty}n\operatorname{Cov}(\widehat{\beta}_{A}(e_1),\widehat{\beta}_{A}(e_2))=\sum_{g\in\mathcal{G}}\omega_g^2 \,C_{\beta_g}(e_1,e_2),\quad e_1,e_2\in[-T_{pre,A}, T_{post,A}],
\end{equation}
where $C_{\beta_g}(e_1,e_2)=\lim_{n\to\infty}n\operatorname{Cov}(\widehat{\beta}_{g}(e_1),\widehat{\beta}_{g}(e_2))$ is the asymptotic covariance function of the group-specific DiD estimator $\widehat{\beta}_{g}(e)$; see Online Appendix \ref{app:cov_staggered_did}. $C_{\beta_g}(e_1,e_2)$ can be estimated using a group-specific version of the covariance estimator in \eqref{eq:cov_est},
\begin{equation*}
\widehat{C}_{\beta_g}(e_1,e_2)=\left(\frac{1}{\bar{n}_g}\sum_{i\in\bar{\mathcal{I}}_g} \dot{D}_{gi}^2 \left( \Delta_0\dot{Y}_i(g+e_1+1)\right)\left( \Delta_0\dot{Y}_i(g+e_2+1)\right) \right)\left(\frac{1}{\bar{n}_g}\sum_{i\in\bar{\mathcal{I}}_g} \dot{D}_{gi}^2\right)^{-2},
\end{equation*}
with $\Delta_0\dot{Y}_i(g+e+1) =(\dot{Y}_i(g+e+1)-\dot{Y}_i(g))-\widehat{\beta}_{g}(e)\dot{D}_{gi}$ and $\widehat{\beta}_{g}(e)$ as defined above. The aggregated covariance function $C_{\beta_A}(e_1,e_2)$ in \eqref{eq:cov_stagger} is then estimated as
\begin{equation*}
    \widehat{C}_{\beta_{A}}(e_1,e_2)=\sum_{g\in\mathcal{G}} \widehat{\omega}_{g}^2\, \widehat{C}_{\beta_g}(e_1,e_2),\quad e_1,e_2\in[-T_{pre,A}, T_{post,A}].
\end{equation*}

In practice, when estimation is performed at discrete time points $e_1,e_2\in\{-T_{pre,A},\dots,T_{post,A}\}$, interpolation estimators, 
$\doublewidehat{\beta}_{A}(e)=\sum_{g\in\mathcal{G}}\widehat{\omega}_{g}\doublewidehat{\beta}_{g}(e)$ and 
$\doublewidehatCL{C}_{\beta_{A}}(e_1,e_2)=\sum_{g\in\mathcal{G}}\widehat{\omega}_{g}^2\,\doublewidehatCL{C}_{\beta_{g}}(e_1,e_2)$
are used, where $\doublewidehat{\beta}_{g}$ is defined analogously to $\doublewidehat{\beta}_n$ in \eqref{eq:BetaHatHat} and $\doublewidehatCL{C}_{\beta_{g}}$ parallels $\doublewidehatCL{C}_{\beta, n}$ in \eqref{eq:CovHatHat}.

\begin{remark}
\normalfont In cases where all units are treated, one can view units in the last treated group as controls and remove periods at which the controls are treated \citep{sun_abraham_2021}.
\end{remark}

%%%%%%%%%%%%%%%%%%%%%%%%%%%%%%%%%%%%%%%%%%%%
\section{Honest Causal Inference}\label{sec:testing}
%%%%%%%%%%%%%%%%%%%%%%%%%%%%%%%%%%%%%%%%%%%%

\subsection{Honest Hypothesis Testing in the Post-Treatment Period}\label{sec:honest_testing}

In causal inference, we are usually interested in testing the no-effect null hypothesis that $\theta_{ATT}(t)=0$ in the post-treatment period $t \in[0,T_{post}]$. Under Assumptions I (No Anticipation) and Assumption II (Parallel Trends), we have by \eqref{eq:DiD_parameter} that $\beta(t)=\theta_{ATT}(t)$. Thus, under Assumptions I and II, we can formulate the simultaneous no-effect hypothesis for $\theta_{ATT}(t)$ equivalently in terms of $\beta(t)$ as\\[-7ex]

\spacingset{1}
\begin{align*}
	\begin{array}{rlrl}
        H_0:& \theta_{ATT}(t) = 0, \quad \forall \; t\in[0,T_{post}] &\Leftrightarrow\quad
		H_0:& \beta(t) = 0, \quad \forall \; t\in[0,T_{post}]\\
		H_1:& \exists \; t \in[0,T_{post}] \;\text{ s.t. } \; \theta_{ATT}(t) \neq 0
		&\Leftrightarrow\quad 
		H_1:& \exists \; t \in[0,T_{post}] \; \text{ s.t. } \;\beta(t) \neq 0	.
	\end{array}
\end{align*}%
\spacingset{1.50}%

To conduct a simultaneous hypothesis test for this testing problem, we use $\doublewidehatCL{\operatorname{SCB}}{}^{\sup}_{1-\alpha}(t)$ in \eqref{eq:supSCB_practical}, which, by Corollary \ref{cor:SCBs_Interpolation}, allows controlling the size of the test uniformly over $t\in[0,T_{post}]$: 
\begin{align*}
\lim_{n\to\infty}\quad & P_{H_0}\left(0\not\in \doublewidehatCL{\operatorname{SCB}}{}^{\sup}_{1-\alpha}(t)\;\text{for at least one}\; t \in [0,T_{post}] \right) = \alpha,
\end{align*}
where $P_{H_0}$ denotes the probability under the null hypothesis $H_0\colon \beta(t) = 0$. That is, if there exists at least one time point $t\in[0,T_{post}]$ for which $0\not\in\doublewidehatCL{\operatorname{SCB}}{}^{\sup}_{1-\alpha}(t)$, we can reject the null hypothesis $H_0$ at the significance level $\alpha$. 

If, however, Assumption I (No Anticipation) and/or Assumption II (Parallel Trends) are violated, we have, by \eqref{eq:DiD_parameter}, that
$$
\beta(t)= \theta_{ATT}(t) + \Delta_{DT}(t)- \Delta_{TA} :=  \theta_{ATT}(t) + \Delta(t).
$$ 
Thus, the honest version of the no-effect null hypothesis and its alternative hypothesis are\\[-7ex]

\spacingset{1}
\begin{align*}
	\begin{array}{rlrl}
		H_0:& \theta_{ATT}(t) = 0, \quad \forall \; t\in[0,T_{post}] &\Leftrightarrow\quad
		H_0:& \beta(t) = \Delta(t), \quad \forall \; t\in[0,T_{post}]\\
		H_1:& \exists \; t \in[0,T_{post}] \;\text{ s.t. } \; \theta_{ATT}(t) \neq 0 &\Leftrightarrow\quad 
		H_1:& \exists \; t \in[0,T_{post}] \;\text{ s.t. } \; \beta(t) \neq \Delta(t)	.
	\end{array}
\end{align*}%
\spacingset{1.50}%
\noindent It is generally infeasible to do a statistical hypothesis test in this case, since $\Delta(t)$ is typically unknown. One may, however, use specific domain knowledge to derive a range of values 
$[\Delta_\ell(t), \Delta_u(t)] = [\Delta(t) - \Delta^c_{\ell}(t), \Delta(t) + \Delta^c_{u}(t)]$ for which one can credibly argue that 
$$
\Delta(t) \in [\Delta_\ell(t), \Delta_u(t)]\quad\text{for all}\quad t\in[-T_{pre},T_{post}].
$$ 
Given such a credible honest reference band $[\Delta_\ell(t), \Delta_u(t)]$ of possible bias values, we can rewrite the hypothesis testing problem as a \textit{relevance} testing problem; namely,\\[-7ex]

\spacingset{1}
\begin{align*}
	\begin{array}{rrclcl}
		H_0:& \theta_{ATT}(t) &\in& [\;\;\;0\;\;-\Delta_\ell^c(t),\;\;\;0\;\;+\Delta_u^c(t)]&\;\text{for all}\;&t\in[0,T_{post}]\\
		\Leftrightarrow\quad
		H_0:& \;\;\;\;\;\beta(t) &\in& [\Delta(t) - \Delta^c_{\ell}(t), \Delta(t) + \Delta^c_{u}(t)]&\;\text{for all}\;& t\in[0,T_{post}]\\[2ex]
		H_1:& \theta_{ATT}(t) &\not\in& [-\Delta_\ell^c(t), \Delta_u^c(t)]&\;\text{for at least one}\;&t \in[0,T_{post}]\\
		\Leftrightarrow\quad 
		H_1:& \beta(t) &\not\in& [\Delta_\ell(t), \Delta_u(t)]&\;\text{for at least one}\;&t \in[0,T_{post}].		
	\end{array}
\end{align*}
\spacingset{1.50}

To conduct a simultaneous hypothesis test for this relevance testing problem, we can again use the $\doublewidehatCL{\operatorname{SCB}}{}^{\sup}_{1-\alpha}(t)$, which, by Corollary \ref{cor:SCBs_relevance}, allows controlling the size of the test uniformly over $t\in[0,T_{post}]$ also in the case of compound null hypotheses as used in relevance tests.

\begin{corollary}[Size Control in Relevance Testing]\label{cor:SCBs_relevance}
    Under Assumptions \ref{assumption: data_str b}, \ref{assumption: moments}, \ref{assumption: smoothness} and the additional condition $\sqrt{n}/T \to 0$, we have
\begin{align*}
\lim_{n\to\infty}\quad & \sup_{H_0} P\left([\Delta_\ell(t), \Delta_u(t)]\cap \doublewidehatCL{\operatorname{SCB}}{}^{\sup}_{1-\alpha}(t)=\emptyset\;\text{for at least one}\; t \in [0,T_{post}] \right) \leq \alpha
\end{align*}
\end{corollary}

That is, if there exists at least one time point $t\in[0, T_{post}]$ for which the reference band $ [\Delta_\ell(t), \Delta_u(t)]$ and the simultaneous confidence band $\doublewidehatCL{\operatorname{SCB}}{}^{\sup}_{1-\alpha}(t)$ do \emph{not intersect}, we can reject the null hypothesis $H_0$ at the significance level $\alpha$. This relevance testing approach facilitates honest causal inference in the post-treatment period using event study plots showing simultaneous confidence bands along with reference bands, which take into account violations of the no-anticipation and/or parallel trends assumptions (e.g., Figure \ref{fig:ESP_honest}).

%%%%%%%%%%%%%%%%%%%%%%%%%%%%%%%%%%%%%%%%%%%%%%%%%%%%%%%%%
\subsection{Choosing Honest Reference Bands}\label{sec:ChoosingHRBs} 
%%%%%%%%%%%%%%%%%%%%%%%%%%%%%%%%%%%%%%%%%%%%%%%%%%%%%%%%%%

The choice of the reference band $[\Delta_\ell(t), \Delta_u(t)]$ typically requires domain knowledge about the specific application at hand. The chosen reference band can be constructed over the entire time span $t \in [-T_{pre}, T_{post}]$, with the post-treatment period $[0,T_{post}]$ used for testing causal effects (Section \ref{sec:honest_testing}) and the pre-anticipation period $[-T_{pre},t_A]$ used for validation (Section \ref{sec:validation}). In the following, we propose honest reference bands for common scenarios where either Assumption I (No Anticipation) or Assumption II (Parallel Trends) is violated. These reference bands can be used as a starting point for applied researchers. Typically, one would want to consider multiple reference bands reflecting different scenarios and degrees of violations of the identification assumptions to assess the robustness of causal conclusions.

%%%%%%%%%%%%%%%%%%%%%%%%%%%%%%%%%%%%%%%%%%%%%%%%%%%%%%%%%
\subsubsection*{Honest Reference Band for Violated No Anticipation Assumption}
%%%%%%%%%%%%%%%%%%%%%%%%%%%%%%%%%%%%%%%%%%%%%%%%%%%%%%%%%

If Assumption II (Parallel Trends) holds, but Assumption I (No Anticipation) is violated, we have by \eqref{eq:AnticipationBias} that 
$\Delta(t) = - \Delta_{TA} = \beta(t) \neq 0$ for all $t \in [-T_{pre},t_A]$, where $t_A$ denotes the time point after which treated units begin responding to the treatment (Section \ref{sec:att_fda}). This bias $-\Delta_{TA}$ can be estimated using the average of $\widehat{\beta}_{n}(t)$ in \eqref{eq:BetaHat} over all observable pre-anticipation periods $t \in \{-T_{pre}, \dots, t_A\}$. To account for the estimation uncertainty, we propose the reference band
\begin{align}\label{eq:ta_bounds}
	\begin{split}
        [\widehat\Delta_\ell(t), \widehat\Delta_u(t)] = \left[ \frac{1}{T_A}\sum_{s=-T_{pre}}^{t_A}\widehat{\beta}_{n}(s) - S_\ell , \quad \frac{1}{T_A}\sum_{s=-T_{pre}}^{t_A}\widehat{\beta}_{n}(s) + S_u \right]
	\end{split}
\end{align}
with constant width for all $t\in[-T_{pre},T_{post}]$, where $T_A=|\{-T_{pre}, \dots, t_A\}|$ denotes the total number of observable pre-anticipation periods, and $S_\ell, S_u > 0$ are control parameters.

%%%%%%%%%%%%%%%%%%%%%%%%%%%%%%%%%%%%%%%%%%%%%%%%%%%%%%%%%
\subsubsection*{Honest Reference Band for Violated Parallel Trends Assumption}
%%%%%%%%%%%%%%%%%%%%%%%%%%%%%%%%%%%%%%%%%%%%%%%%%%%%%%%%%

If Assumption I (No Anticipation) holds, but Assumption II (Parallel Trends) is violated, we have by \eqref{eq:DTBias} that 
$
\Delta(t) = \Delta_{DT}(t),  
$
where $\Delta_{DT}(t) \neq 0$ for one or more time points in the post-treatment time period, $t\in[0,T_{post}]$. To define a useful reference band $[\Delta_\ell(t),\Delta_u(t)]$, we suggest extrapolating information from the pre-treatment period $t\in[-T_{pre},-1]$. Under the considered scenario, $\theta_{ATT}(t)=0$ for $t\in[-T_{pre},-1]$, such that $\beta(t)=\Delta_{DT}(t)$ for all $t\in[-T_{pre},-1]$. Thus, we can estimate $\Delta_{DT}(t)$ by $\widehat{\beta}_n(t)$ at all observable pre-treatment periods $t\in\{-T_{pre},\dots,-1\}$. Assuming $\Delta_{DT}(t)$ to be linear, we can further estimate the slope of $\Delta_{DT}(t)$ relative to the reference period $t=-1$ by taking average of the slopes of $\widehat{\beta}_n(t)$ over all $t\in\{-T_{pre},\dots,-2\}$, and extrapolate this information into the post-treatment time period, $t\in[0,T_{post}]$. To do so, we propose the following honest reference band:
\begin{align}\label{eq:rmtrb_bounds}
        [\widehat\Delta_\ell(t), \widehat\Delta_u(t)] = \left[ \left(\frac{1}{T_{pre}-1} \sum_{s=-T_{pre}}^{-2} \frac{\widehat{\beta}_n(s)}{s+1} \right) (t+1)-M_\ell , \quad \left(\frac{1}{T_{pre}-1} \sum_{s=-T_{pre}}^{-2} \frac{\widehat{\beta}_n(s)}{s+1} \right) (t+1)+M_u  \right]
\end{align}
with constant width for all $t\in[-T_{pre}, T_{post}]$, where $M_\ell, M_u > 0$ are control parameters. The honest reference band in \eqref{eq:rmtrb_bounds} essentially forms a band centered around the average slope of the pre-trend periods. 

\begin{remark}
The reference bands in \eqref{eq:ta_bounds} and \eqref{eq:rmtrb_bounds} can be merged to incorporate both anticipation effects and differential trends. Specifically, we are supposed to first construct the band addressing anticipatory effects by \eqref{eq:ta_bounds}, and then further build the bands addressing differential trends by \eqref{eq:rmtrb_bounds} with $t_A$ as the new reference period, respectively on the upper and lower band constructed from \eqref{eq:ta_bounds} at the previous step. The final band consists of the infimum and supremum of the total area defined altogether.
\end{remark}

In Corollary \ref{cor:SCBs_relevance}, the reference band $[\Delta_\ell(t), \Delta_u(t)]$ usually needs to be deterministic to achieve the size control in relevance testing. Estimating a reference band using the same data as used for testing the causal effects can lead to anti-conservative, invalid inference, also known as the double-dipping issue \citep[cf.][]{wang_2024}. However, by Corollary \ref{cor:SCBs_relevance_practical}, our reference bands in \eqref{eq:ta_bounds} and \eqref{eq:rmtrb_bounds} do not suffer from this issue even if we use the same data to estimate the reference band and to test causal effects in practice.

\begin{corollary}[Size Control in Relevance Testing with \eqref{eq:ta_bounds} and \eqref{eq:rmtrb_bounds}]\label{cor:SCBs_relevance_practical}
    Under Assumptions \ref{assumption: data_str b}, \ref{assumption: moments}, \ref{assumption: smoothness} and the additional condition $\sqrt{n}/T \to 0$, we have
\begin{align*}
\lim_{n\to\infty}\quad & \sup_{H_0} P\left([\widehat\Delta_\ell(t), \widehat\Delta_u(t)]\cap \doublewidehatCL{\operatorname{SCB}}{}^{\sup}_{1-\alpha}(t)=\emptyset\;\text{for at least one}\; t \in [0,T_{post}] \right) \leq \alpha,
\end{align*}
where $[\widehat\Delta_\ell(t), \widehat\Delta_u(t)]$ is derived from \eqref{eq:ta_bounds} or \eqref{eq:rmtrb_bounds}.
\end{corollary}
The key to Corollary \ref{cor:SCBs_relevance_practical} is that Assumption \ref{assumption: data_str b} implies $T_{pre} \to \infty$, which makes the average of $\widehat{\beta}_n(t)$ in \eqref{eq:ta_bounds} over the pre-anticipation periods and the average of the slopes of $\widehat{\beta}_n(t)$ in \eqref{eq:rmtrb_bounds} over the pre-treatment periods converge at a faster rate of $1/\sqrt{nT_{pre}}$ than $\doublewidehat{\beta}_n(t)$ at $1/\sqrt{n}$ as shown in Theorem \ref{thm:intrpl_dist}. This validates the use of \eqref{eq:ta_bounds} and \eqref{eq:rmtrb_bounds} to construct the reference band in the practical scenario.

\vspace{-0.5em}

%%%%%%%%%%%%%%%%%%%%%%%%%%%%%%%%%%%%%%%%%%%%%%%%%%%%%%%%%%%%%%%%%%%%%%%%
\subsection{Validating Reference Bands in the Pre-Anticipation Period}\label{sec:validation}
%%%%%%%%%%%%%%%%%%%%%%%%%%%%%%%%%%%%%%%%%%%%%%%%%%%%%%%%%%%%%%%%%%%%%%%%

In event study plots, the pre-anticipation period $[-T_{pre},t_A]$, $t_A \leq -1$, can be used to assess the validity of a chosen honest reference band $[\Delta_\ell(t),\Delta_u(t)]$ via \textit{equivalence} testing; see \cite{Wellek_2002} for an introduction to equivalence testing in case of pointwise testing. If the honest reference band reflects only violations of the parallel trends assumption, then $t_A=-1$. 

Throughout the pre-anticipation period $t\in[-T_{pre},t_A]$, we have $\theta_{ATT}(t)=0$ and hence $\beta(t)=\Delta(t)$. The goal is to verify whether $\beta(t)$, and thus $\Delta(t)$, lies entirely \emph{within} the reference band $[\Delta_\ell(t),\Delta_u(t)]$ for all $t\in[-T_{pre},t_A]$. If so, the reference band can be regarded as validated---at least for the pre-anticipation period. This equivalence problem can be expressed using two one-sided hypotheses; namely,\\[-7ex]

\spacingset{1}  
\begin{align*}
    \begin{array}{rllll}
    H_0^{-}\colon \exists \; t \in[-T_{pre}, t_A] &\text{ s.t. } \; \beta(t) \leq \Delta_\ell(t) & \quad \text{vs.} \quad &  H_1^{-}: \Delta_\ell(t) < \beta(t), & \forall \; t\in[-T_{pre}, t_A]\\
    H_0^{+}\colon \exists \; t \in[-T_{pre}, t_A] &\text{ s.t. } \; \Delta_u(t)\leq \beta(t) & \quad \text{vs.} \quad &  H_1^{+}: \beta(t) < \Delta_u(t), & \forall \; t\in[-T_{pre}, t_A].
    \end{array}
\end{align*}
\spacingset{1.50} 

\noindent The global hypotheses of the combined test represent the hypotheses of the equivalence test\\[-7ex]

\spacingset{1}  
\begin{align*}
    \begin{array}{rlcl}
H_0\colon& \beta(t) \not\in [\Delta_\ell(t),\Delta_u(t)]&\;\text{for at least one}\;&t \in[-T_{pre}, t_A]\\
H_1\colon& \beta(t) \in [\Delta_\ell(t),\Delta_u(t)]&\;\text{for all} \; &t \in[-T_{pre}, t_A].
    \end{array}
\end{align*}
\spacingset{1.50} 

\noindent The global null hypothesis $H_0=H_0^{-}\,\cup\, H_0^{+}$ holds if at least one of the two one-sided null hypotheses $H_0^{-}$ or $H_0^{+}$ is true. An honest reference band $[\Delta_\ell(t),\Delta_u(t)]$ is said to be validated at significance level $\alpha$ if $H_0$ can be rejected, so that the alternative $H_1=H_1^{-}\cap H_1^{+}$ is accepted at level $\alpha$---that is, if \emph{both} one-sided null hypotheses $H_0^{-}$ and $H_0^{+}$ are rejected. 

This Two-One-Sided-Testing (TOST) procedure consitutes an Intersection-Union Test (IUT) \citep[see][Theorem 8.3.23]{Casella_Berger_2024}. As a result, conducting each one-sided test at level $\alpha$ guarantees that the probability of incorrectly validating an invalid reference band is controlled at level $\alpha$.

Note that the hypotheses in equivalence testing revert the roles of the null and alternative hypotheses compared to relevance testing (Section \ref{sec:honest_testing}), which requires switching from supremum-based to infimum-based testing. The one-sided null hypothesis $H_0^{-}$ can be rejected at level $\alpha$ if the infimum-based $(1-\alpha)\times100\%$ simultaneous confidence band $\doublewidehatCL{\operatorname{SCB}}{}^{\inf,-}_{1-\alpha}(t)$ in \eqref{eq:infSCB_minus_practical} is strictly larger than the lower bound of the reference band $\Delta_\ell(t)$ for all $t\in[-T_{pre}, t_A]$,
\begin{equation}
\Delta_\ell(t)<\doublewidehat{\beta}_n(t) - \doublewidehatSL{u}^{\,\inf}_{1-\alpha}\sqrt{\doublewidehatCL{C}_{\beta,n}(t, t) /n}\quad\text{for all}\quad t\in[-T_{pre}, t_A].\label{eq:rejectionRule_minus}
\end{equation}  
Similarly, the one-sided null hypothesis $H_0^{+}$ can be rejected at level $\alpha$ if the one-sided infimum-based $(1-\alpha)\times100\%$ simultaneous confidence band $\doublewidehatCL{\operatorname{SCB}}{}^{\operatorname{inf,+}}_{1-\alpha}(t)$ in \eqref{eq:infSCB_plus_practical} is strictly smaller than the upper bound of the reference band $\Delta_u(t)$ for all $t\in[-T_{pre}, t_A]$,
\begin{equation}
\doublewidehat{\beta}_n(t) + \doublewidehatSL{u}^{\,\inf}_{1-\alpha}\sqrt{\doublewidehatCL{C}_{\beta,n}(t, t) /n} < \Delta_u(t)\quad\text{for all}\quad t\in[-T_{pre}, t_A].\label{eq:rejectionRule_plus}
\end{equation}  
Thus, we can reject the global null hypothesis $H_0$ at level $\alpha$ if the two-sided infimum-based $(1-2\alpha)\times100\%$ simultaneous confidence band lies entirely \emph{within} the reference band $\left[\Delta_\ell(t), \Delta_u(t)\right]$,
\begin{equation}
\doublewidehatCL{\operatorname{SCB}}{}^{\inf}_{1-2\alpha}(t)=\left[\doublewidehat{\beta}_n(t) \pm \doublewidehatSL{u}^{\,\inf}_{1-\alpha}\sqrt{\doublewidehatCL{C}_{\beta,n}(t, t) /n}\right]\subsetneqq \left[\Delta_\ell(t), \Delta_u(t)\right]\;\text{for all}\; t\in[-T_{pre}, t_A],\label{eq:rejectionRuleEquivalenceTesting}
\end{equation}
where the latter rejection rule follows from combining the two rejection rules in \eqref{eq:rejectionRule_minus} and \eqref{eq:rejectionRule_plus}. This generalizes the idea of confidence interval inclusion \citep[cf.][Ch.~3]{Wellek_2002} to the case of simultaneous confidence bands. By Corollary \ref{cor:SCBs_Interpolation} and the IUT principle, the overall size is  controlled at level $\alpha$.

\begin{corollary}[Size Control in Equivalence Testing]\label{cor:SCBs_equivalence}
    Under Assumptions \ref{assumption: data_str b}, \ref{assumption: moments}, \ref{assumption: smoothness} and the additional condition $\sqrt{n}/T \to 0$, we have
$$
\lim_{n\to\infty} \sup_{H_0} P\Big(\doublewidehatCL{\operatorname{SCB}}{}^{\inf}_{1-2\alpha}(t)\subsetneqq \left[\Delta_\ell(t), \Delta_u(t)\right]\quad\text{for all}\quad t\in[-T_{pre}, t_A]\Big) \le \alpha.
$$
\end{corollary}

\begin{remark}
This approach can only validate the reference band in pre-anticipation period; it does not necessarily guarantee the reference band to be valid also in post-treatment period. 
\end{remark} 

%%%%%%%%%%%%%%%%%%%%%%%%%%%%%%%%%%%%%%%%%%%%
\section{Simulations}\label{sec:sim}
%%%%%%%%%%%%%%%%%%%%%%%%%%%%%%%%%%%%%%%%%%%%

We generate data at $T$ equidistant time points $t_k=-T_{pre}+\frac{k-1}{T-1}(T_{pre}+T_{post})$ for $k=1,\dots,T$, with fixed $-T_{pre}=-10$ and $T_{post}=10$, based on the function-on-scalar model\\[-2.5ex]
$$
Y_i(t)=\beta(t)D_i+\lambda_i+\phi(t)+\varepsilon_{i}(t),\quad t\in[-T_{pre},T_{post}],\\[-.5ex]
$$
where $\beta(-1)=0$. Individual fixed effects are drawn from a uniform distribution, $\{\lambda_i\}_{i=1}^n \stackrel{\text{i.i.d.}}{\sim} U[-3,3]$. The temporal fixed effects are given by the polynomial 
$\phi(t)=4^{-5}5^{-4}[2000(t+10)^3-150(t+10)^4+3(t+10)^5]$, 
$t\in[-T_{pre},T_{post}]$. 
Treatment assignment is binary, with $D_i \mid \lambda_i \sim \text{Bernoulli}(\pi_i)$ and
$\pi_i=\pi(\lambda_i)=\frac{\exp(3\lambda_i)}{1+\exp(3\lambda_i)},$ 
$i=1,\dots,n$.  
$D_i$ is independent of $\varepsilon_i$ but depends on $\lambda_i$, and units with larger $\lambda_i$ are more likely to be selected into treatment. 

For the functional DiD parameter, we have 
$\beta(t)=\theta_{ATT}(t)-\Delta_{TA}+\Delta_{DT}(t)$
as shown in \eqref{eq:DiD_parameter}. In Section \ref{SIM:intrpl_error}, we assume Assumption I (No Anticipation) and Assumption II (Parallel Trends) to hold, i.e. $\Delta_{TA}=\Delta_{DT}(t)=0$, when evaluating the accuracy of the interpolation estimator $\doublewidehat{\beta}_n$. In Section \ref{ssec:sim_pt}, we relax Assumption II by defining a differential trend as $\Delta_{DT}(t)=0.1(t+1)$, when analyzing our honest hypothesis testing in the post-treatment period. Whereas, in Section \ref{ssec:validation}, we relax Assumption I by presuming a treatment anticipation as $\Delta_{TA} \neq 0$, when assessing the reference band in the pre-anticipation period. We study two scenarios for $\theta_{ATT}(t)$; see Figure \ref{fig:ATTs}:\\[-7ex]

\spacingset{1}
\begin{align*}
\text{(Simple") } \textbf{ATT1}\colon\; &\theta_{ATT}(t)=a\cdot\left[ \frac{2(t+0.5)}{3+(t+0.5)} \mathbbm{1}_{\{t\in(-0.5,T_{post}]\}}\right], \\ \text{(Complex") } \textbf{ATT2}\colon\; &\theta_{ATT}(t)=a\cdot\left[\left( \frac{2(t+0.5)}{3+(t+0.5)}+0.3\cos(3(t+0.5))-0.3\right)\mathbbm{1}_{\{t\in(-0.5,T_{post}]\}}\right],\\[-5ex]
\end{align*}%
\spacingset{1.50}% 
\noindent where $\mathbbm{1}_{\{\cdot\}}$ is the indicator function. In both scenarios, the treatment occurs at $t=-0.5 \in (-1,0)$, and $\theta_{ATT}(t)=0$ for all $t \le -0.5$. Considering different values of $a\in\mathbb{R}$ allows considering different effect sizes. 

\begin{figure}[!htbp]
\centering
\includegraphics[width=0.98\linewidth]{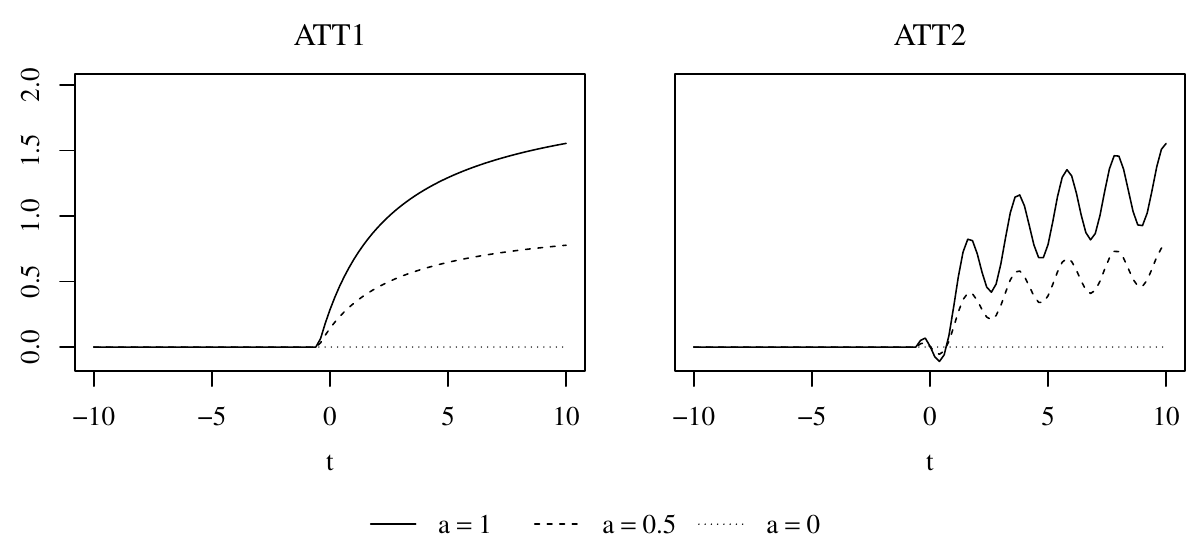}
\caption[]{ATT parameter $\theta_{ATT}(t)$ for the two scenarios. }
\label{fig:ATTs}
\end{figure}

The error term follows a Gaussian process, $\{\varepsilon_i\}_{i=1}^n \stackrel{\text{i.i.d.}}{\sim}\mathcal{GP}(0,C_{\varepsilon})$, with Matérn covariance
$C_{\varepsilon}(s,t)=\sigma^2 (2^{1-\nu}/\Gamma(\nu))(\sqrt{2\nu}|(s-t)/10|)^{\nu}K_{\nu}(\sqrt{2\nu}|(s-t)/10|)$, 
$s,t \in[-T_{pre},T_{post}]$,
where $\sigma^2$ is the variance, $\Gamma$ the Gamma function, $K_{\nu}$ the modified Bessel function of the second kind, and $\nu \ge 0$ controls curve roughness. We consider two cases:\\[-5ex]
\begin{align*}
\text{(Smooth") } \textbf{Cov1}\colon\; \sigma^2=1,\; \nu=3/2, \quad\text{and}\quad \text{(Rough") } \textbf{Cov2}\colon\; \sigma^2=1,\; \nu=2/3.\\[-6ex]
\end{align*}
Cov1 yields smooth, differentiable error processes with high temporal dependence; Cov2 produces rough, non-differentiable paths (violating Assumption \ref{assumption: smooth a}) with low temporal dependence. In simulations, functional curves are observed only at $T\in\{21,41\}$ equidistant time points $t\in\{-T_{pre},\dots,-1,0,\dots,T_{post}\}$, yielding a standard panel data structure with $n\in\{100,200,400,800\}$. Figure \ref{fig:SamplePaths} shows exemplary outcome curves $Y_i=\{Y_i(t):t\in[-T_{pre},T_{post}]\}$ for ATT1 with $a=1$ and $T=21$, along with their actually observed discrete-time panel data points (triangles and dots).

%%%%%%
\begin{figure}[!htbp]
	\centering
	\includegraphics[width=0.98\linewidth]{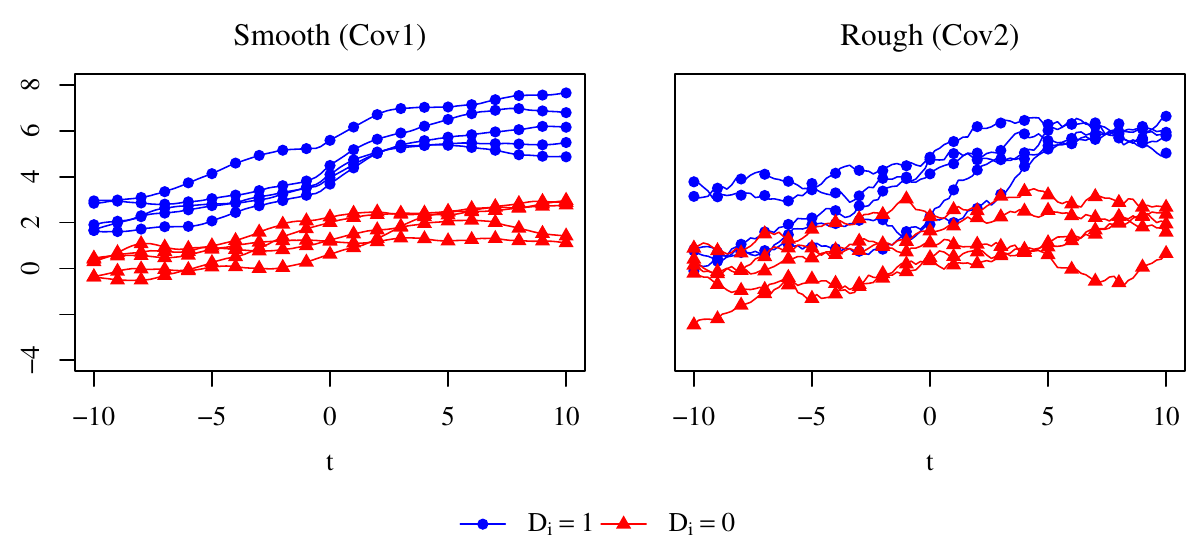}
	\caption[]{Exemplary outcome curves $Y_i$ for ATT1 with $a=1$ and $T=21$, along with their actually observed discrete-time panel data points (triangles and dots).}
	\label{fig:SamplePaths}
\end{figure}
%%%%%%

\subsection{Estimation Accuracy}\label{SIM:intrpl_error}

To evaluate the accuracy of the interpolation estimator $\doublewidehat{\beta}_n=\{\doublewidehat{\beta}_n(t):t\in[-T_{pre},-1] \cup [0,T_{post}]\}$ in \eqref{eq:BetaHatHat}, we use the metric
\begin{equation*}
    \doublewidehat{Q}=\max_{j=1,\dots,101}\left|\doublewidehat{\beta}_n(t_j)-\beta(t_j)\right|
\end{equation*}
based on a grid of $101$ time points over $[-T_{pre},-1]\cup[0,T_{post}]$, chosen equidistantly within the pre- and post-treatment periods separately. This metric captures the uniform distance between $\doublewidehat{\beta}_n$ and the true parameter $\beta$, reflecting both estimation and interpolation errors. For each $(n,T)$-combination, we run $500$ simulations and report the mean of $\doublewidehat{Q}$ with confidence intervals.

%{%
\spacingset{1}
\begin{table}[h!tb]
\centering
\caption{Mean of $\protect\doublewidehat{Q}$ and confidence intervals for ATT1 and ATT2 with $a=1$}  
\label{tbl:ATT1InterpolationError}
\setlength{\tabcolsep}{15pt}
\begin{tabular}{rcccc}
\toprule
& \multicolumn{2}{c}{Smooth (Cov1)} & \multicolumn{2}{c}{Rough (Cov2)} \\ 
\cmidrule(lr){2-3} \cmidrule(lr){4-5}
& $T$=21 & $T$=41 & $T$=21 & $T$=41  \\ 
\cmidrule(lr){2-2} \cmidrule(lr){3-3}\cmidrule(lr){4-4}\cmidrule(lr){5-5}
ATT1\;\;$n$=100 & ${}_{0.24\;\;}0.25_{\;\;0.26}$ 
& ${}_{0.24\;\;}0.25_{\;\;0.26}$ 
& ${}_{0.30\;\;}0.31_{\;\;0.32}$ 
& ${}_{0.30\;\;}0.31_{\;\;0.32}$       \\ 
    $n$=200 & ${}_{0.17\;\;}0.18_{\;\;0.19}$ 
& ${}_{0.17\;\;}0.18_{\;\;0.19}$ 
& ${}_{0.21\;\;}0.22_{\;\;0.22}$ 
& ${}_{0.22\;\;}0.23_{\;\;0.24}$       \\ 
    $n$=400 & ${}_{0.12\;\;}0.13_{\;\;0.13}$ 
& ${}_{0.12\;\;}0.13_{\;\;0.13}$ 
& ${}_{0.15\;\;}0.15_{\;\;0.16}$ 
& ${}_{0.15\;\;}0.16_{\;\;0.16}$      \\ 
    $n$=800 & ${}_{0.08\;\;}0.09_{\;\;0.09}$ 
& ${}_{0.09\;\;}0.09_{\;\;0.09}$ 
& ${}_{0.11\;\;}0.11_{\;\;0.11}$ 
& ${}_{0.11\;\;}0.11_{\;\;0.12}$      \\ 
\midrule
ATT2\;\;$n$=100 & ${}_{0.37\;\;}0.38_{\;\;0.39}$ 
& ${}_{0.25\;\;}0.26_{\;\;0.27}$ 
& ${}_{0.41\;\;}0.42_{\;\;0.43}$ 
& ${}_{0.30\;\;}0.31_{\;\;0.32}$      \\ 
        $n$=200 & ${}_{0.32\;\;}0.32_{\;\;0.33}$ 
& ${}_{0.18\;\;}0.19_{\;\;0.20}$ 
& ${}_{0.34\;\;}0.35_{\;\;0.35}$ 
& ${}_{0.23\;\;}0.23_{\;\;0.24}$         \\ 
        $n$=400 & ${}_{0.28\;\;}0.29_{\;\;0.29}$ 
& ${}_{0.13\;\;}0.14_{\;\;0.14}$ 
& ${}_{0.30\;\;}0.31_{\;\;0.31}$ 
& ${}_{0.16\;\;}0.16_{\;\;0.17}$      \\ 
        $n$=800 & ${}_{0.26\;\;}0.27_{\;\;0.27}$ 
& ${}_{0.10\;\;}0.10_{\;\;0.11}$ 
& ${}_{0.28\;\;}0.28_{\;\;0.28}$ 
& ${}_{0.12\;\;}0.12_{\;\;0.12}$       \\ 
\bottomrule
\end{tabular}
\footnotesize \newline In format ${}_{A\;} B_{\;C}$, where $B$ is the mean of $\doublewidehat{Q}$ over all $500$ simulations and $[A,C]=[B\pm 1.96\cdot \doublewidehatSL\sigma_B/\sqrt{500}]$ denotes the $95\%$ confidence interval with $\doublewidehatSL\sigma_B$ representing the standard deviation of $\doublewidehat{Q}$ over all simulations.
\end{table}
%}
\spacingset{1.50}

In this section, we suppose Assumption I (No Anticipation) and Assumption II (Parallel Trends) to hold, thereby leading to $\beta(t)=\theta_{ATT}(t)$ for all $t\in[-T_{pre},-1] \cup [0,T_{post}]$. The upper panel of Table \ref{tbl:ATT1InterpolationError} shows results for ATT1 with $a=1$. The mean of $\doublewidehat{Q}$ decreases markedly as $n$ grows, while $T$ has little effect since $T=21$ already ensures negligible interpolation error for the simple ATT1 curve. The lower panel reports results for the more complex ATT2 curve, where interpolation errors are larger but vanish as $T$ increases from $T=21$ to $T=41$. These findings align with our uniform consistency result (Theorem \ref{thm:intrpl_consistency}).

%%%%%%%%%%%%%%%%%%%%%%%%%%%%%%%%%%%%%%%%%%%%%%%%%%%%%%%%%%%%%%%%%%%%%%%%%%
\subsection{Honest Hypothesis Testing in the Post-Treatment Period}\label{ssec:sim_pt}
%%%%%%%%%%%%%%%%%%%%%%%%%%%%%%%%%%%%%%%%%%%%%%%%%%%%%%%%%%%%%%%%%%%%%%%%%%

To assess the honest hypothesis testing using our simultaneous confidence band in the post-treatment period, we suppose that Assumption II (Parallel Trends) is violated and define\\[-2.5ex]
$$
\beta(t)=\theta_{ATT}(t)+\Delta_{DT}(t),\quad t\in[-T_{pre},T_{post}],
$$ 
where $\Delta_{DT}(t)=0.1(t+1)$. We are interested in testing the no-effect null hypothesis $\theta_{ATT}(t)=0$ in the post-treatment period $t \in[0,T_{post}]$. Due to the presence of differential trends bias $\Delta_{DT}(t)$, the no-effect null hypothesis can not be directly tested using the classical hypothesis testing (see Online Appendix \ref{SIM:classic_hypothesis_test} for a simulation where classical hypothesis testing is suitable). Instead, by Section \ref{sec:honest_testing}, we can reframe the no-effect null hypothesis as the following relevance hypothesis test:
\begin{equation*}
    H_0: \beta(t) \in [\Delta_{\ell}(t), \Delta_{u}(t)], \; \forall \; t\in[0,T_{post}] \quad \text{vs.} \quad H_1:\exists \; t\in[0,T_{post}] \; \text{s.t.} \; \beta(t) \not\in [\Delta_{\ell}(t), \Delta_{u}(t)],
\end{equation*}
where the reference band $[\Delta_{\ell}(t), \Delta_{u}(t)]$ is estimated using the reference band in \eqref{eq:rmtrb_bounds}, considering control parameter values $M_{\ell}=M_u=M\in\{0.1,0.01\}$. The sample parameters in \eqref{eq:rmtrb_bounds} are determined using separate training data with sample size $n_T=n$ generated in the same way as the testing data, so that the reference band can be kept fixed over all simulations for each $(n,T)$-combination---except for varying the parameter $M$ to investigate the effect of the width of reference band.

Data under $H_0$ are generated by setting $a=0$ for the ATT parameter $\theta_{ATT}(t)$, while values $|a|>0$ generate data under $H_1$ for constructing power curves. For each $(n,T,M)$-combination, we run 500 simulations under the specified data generation process. In each run, we check whether $\doublewidehatCL{\operatorname{SCB}}{}^{\sup}_{1-\alpha}(t)$, with significance level $\alpha=0.05$, intersects with $[\Delta_{\ell}(t), \Delta_{u}(t)]$ at a grid of $101$ equidistant time points over the post-treatment period $[0, T_{post}]$. We reject $H_0$ if they do not intersect for at least one of the $101$ grid points. Otherwise, we retain $H_0$.  The empirical uniform size is the proportion of simulations rejecting $H_0$, when it is true.

We evaluate the performance of the supremum-based simultaneous confidence band 
$\doublewidehatCL{\operatorname{SCB}}{}^{\sup}_{1-\alpha}(t)$ in \eqref{eq:supSCB_practical} using three approaches to estimate the critical value $u^{\sup}_{1-\alpha}$: the Parametric Bootstrap (SCB-PB (Sup); see Online Appendix \ref{ssec:SCB_PB_SUP}), the Multiplier Bootstrap (SCB-MB (Sup); see Online Appendix \ref{ssec:SCB_MB_SUP}), and the Kac-Rice formula (SCB-KR (Sup); see Online Appendix \ref{ssec:SCB_KR}). We compare these three variants with two commonly used benchmarks: the Naive pointwise $t$-band $\doublewidehatCL{\operatorname{CI}}{}^{\!\text{Naive}}_{1-\alpha}(t)$, defined in \eqref{eq:pw_ci_t}, and the Bonferroni-corrected $t$-band
\begin{align*}
\doublewidehatCL{\operatorname{CI}}{}^{\!\text{Bonf}}_{1-\alpha}(t)\!\!=\!\!\left[\doublewidehat{\beta}_n(t)\!\pm\! t_{1-\frac{\alpha/101}{2}, \text{df}}\sqrt{\doublewidehatCL{C}_{\beta,n}(t,t)/n}\;\right],
\end{align*}
where $\text{df}=n-1$ and $\alpha/101$ is the Bonferroni adjustment for testing at $101$ grid points over $[0,T_{post}]$. The Naive pointwise $t$-band is widely used in empirical practice (\citealp{bosch_campos-vazquez_2014, bailey_goodman-bacon_2015, lovenheim_willen_2019}), while the Bonferroni band provides a conservative multiple-testing benchmark.

\begin{figure}[ht!]
	\centering
	\includegraphics[width=0.98\linewidth]{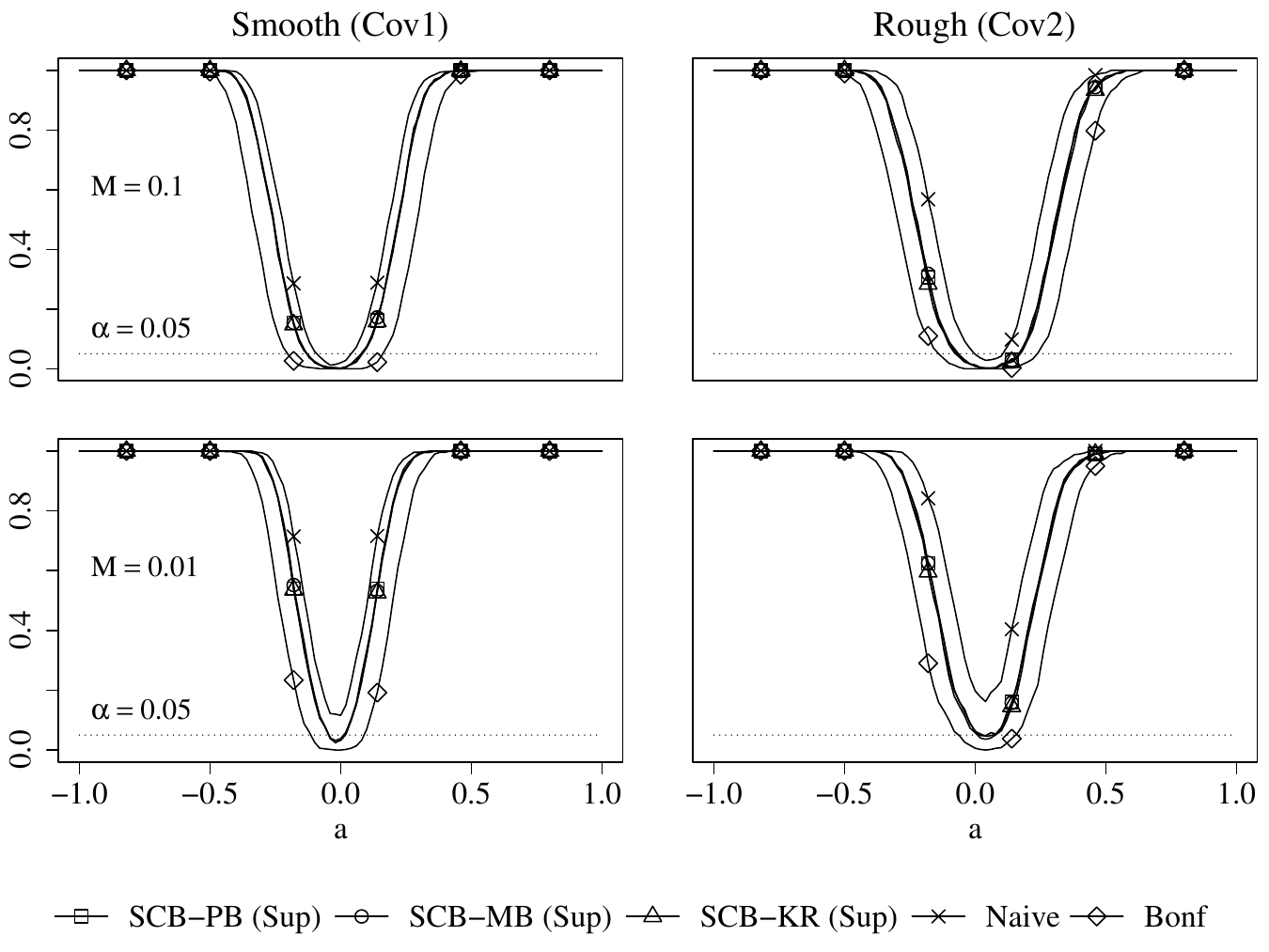}
	\caption[]{Power curves, under violated Assumption II, for ATT1, $n=200$, and $T=21$.}
	\label{fig:N200ATT1PowerCurves_PT}
\end{figure}

Figure \ref{fig:N200ATT1PowerCurves_PT} shows power curves using ATT1 with $n=200$ and $T=21$ for both covariance scenarios and both $M\in\{0.1,0.01\}$. When the reference band is wide ($M=0.1$), all confidence bands are conservative, reflecting the loss of power that arises in honest inference when model assumptions are violated. The three $\doublewidehatCL{\operatorname{SCB}}{}^{\sup}_{1-\alpha}(t)$ bands power-dominate the Bonferroni band, while the Naive band is invalid, which becomes obvious in the narrow reference band scenario ($M=0.01$). When the reference band gets narrower, the Naive band shows its invalidity, failing to control the uniform size under $H_0$ ($a=0$). The Bonferroni band remains overly conservative and uniformly power-dominated by the three $\doublewidehatCL{\operatorname{SCB}}{}^{\sup}_{1-\alpha}(t)$ bands. Results for ATT2 (Online Appendix \ref{SIM:honest_hypothesis_test}) are qualitatively equivalent.

\subsection{Validating Reference Bands in the Pre-Anticipation Period}\label{ssec:validation}

To assess our testing procedure for validating a chosen reference band in the pre-anticipation period, we suppose that Assumption I (No Anticipation) is violated with a treatment anticiaption starting after event time $t_A=-4$. In this case, we consider ATT parameter functions that are non-zero starting at $t=-3.5 \in (-4,-3)$; see Figure \ref{fig:ATTs_Anticipation}: \\[-7ex]

\spacingset{1}
\begin{align*}
\text{(Simple")} \textbf{$\text{ATT1}^*$}\colon\; &\theta_{ATT}(t)\!=\!\!\frac{2(t+3.5)}{3+(t+3.5)} \mathbbm{1}_{\{t\in(-3.5,T_{post}]\}}, \\ 
\text{(Complex")} \textbf{$\text{ATT2}^*$}\colon\;&\theta_{ATT}(t)\!=\!\!\left(\frac{2(t+3.5)}{3+(t+3.5)}+\!0.3\cos(3(t\!+\!3.5))\!-\!0.3\right)\mathbbm{1}_{\{t\in(-3.5,T_{post}]\}}.\\[-5ex]
\end{align*}%
\spacingset{1.50}%
%2(t\!+\!4)^{3/2}/(3\!+\!(t\!+\!4)^{3/2})\!
\begin{figure}[!htbp]
\centering
\includegraphics[trim = 0cm 1.1cm 0cm 0cm, clip,width=0.98\linewidth]{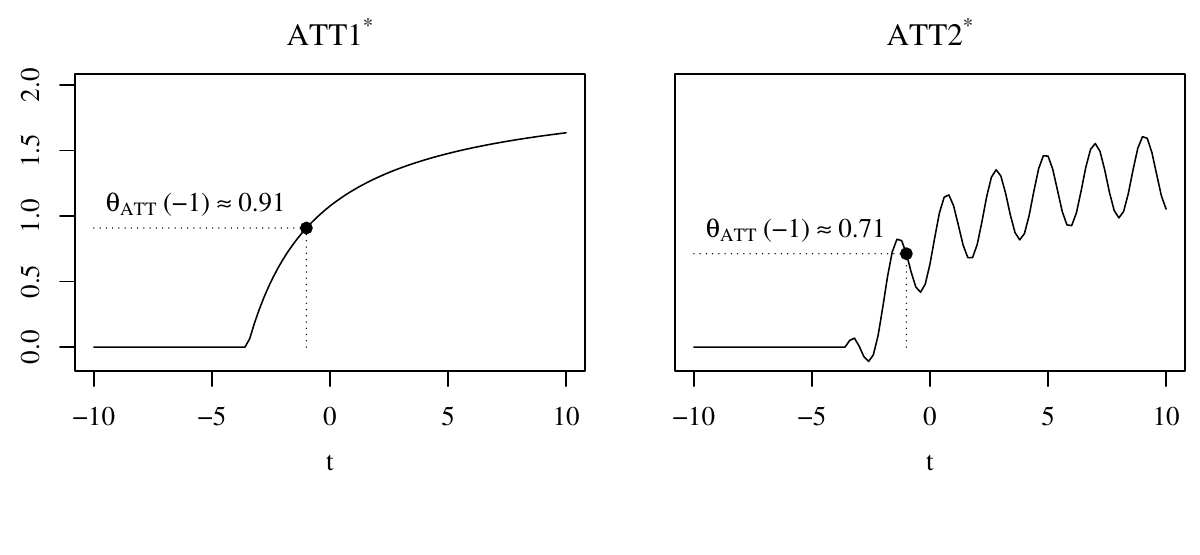}
\caption[]{ATT parameter $\theta_{ATT}(t)$ for the two scenarios under treatment anticipation. }
\label{fig:ATTs_Anticipation}
\end{figure}

By \eqref{eq:DiD_parameter}, we obtain $\beta(t)=\theta_{ATT}(t)-\Delta_{TA}$, $t\in[-T_{pre},T_{post}]$, where $\Delta_{TA}=\theta_{ATT}(-1)$, and $\beta(-1)=0$ by construction. Moreover, following \eqref{eq:AnticipationBias}, we have that $\beta(t)=-\Delta_{TA}$ for any $t\in [-T_{pre}, t_A]$. We use equivalence testing to validate the reference band (Section \ref{sec:validation}) over the pre-anticipation period $[-T_{pre}, t_A]$:\\[-7ex]

\spacingset{1}  
\begin{align*}
    \begin{array}{rlcl}
H_0\colon& \beta(t) \not\in [\Delta_\ell(t),\Delta_u(t)]&\;\text{for at least one}\;&t \in[-T_{pre}, t_A]\\
H_1\colon& \beta(t) \in [\Delta_\ell(t),\Delta_u(t)]&\;\text{for all} \; &t \in[-T_{pre}, t_A],
    \end{array}
\end{align*}
\spacingset{1.50} 

\noindent where the reference band $[\Delta_{\ell}(t),\; \Delta_u(t)]=[-\theta_{ATT}(-1)+1 \pm S]$, and $-\theta_{ATT}(-1)$ is estimated by $-\widehat\theta_{ATT}(-1)$, the average of $\widehat\beta_n(t)$ over all observable pre-anticipation periods $t \in \{-T_{pre},\dots,t_A\}$, as shown in \eqref{eq:ta_bounds}. The sample parameters in \eqref{eq:ta_bounds} are determined using separate training data with sample size $n_T=n$ generated in the same way as the testing data, so that the reference band can be kept fixed over all simulations for each $(n, T)$-combination. The value $-\widehat\theta_{ATT}(-1)+1$ stands for an imperfect estimate of the bias $-\Delta_{TA}=-\theta_{ATT}(-1)\neq 0$, and choosing different values for $S$ allows us to consider the effect of the width of reference band.

Setting $S=1$ leads to data generated under the global null hypothesis, since $\beta(t)=\Delta_{\ell}(t)=-\theta_{ATT}(-1)+1-1$ for all $t\in[-T_{pre}, t_A]$, allowing us to check the empirical size of the equivalence testing procedure. Increasing the control parameter $S>1$ generates data under the alternatives for constructing power curves. For each $(n,T,S)$-combination, we run 500 simulations. Following the rejection rule in \eqref{eq:rejectionRuleEquivalenceTesting}, we reject the global null hypothesis if $\doublewidehatCL{\operatorname{SCB}}{}^{\inf}_{1-2\alpha}(t)$, with $\alpha=0.05$, lies strictly within the chosen reference band $[\Delta_{\ell}(t), \Delta_u(t)]$ for all 101 grid points over the pre-anticipation period $[-T_{pre},t_A]$. Otherwise, we fail to reject the global null hypothesis. For estimating the critical value $u^{\inf}_{1-\alpha}$ in $\doublewidehatCL{\operatorname{SCB}}{}^{\inf}_{1-2\alpha}(t)$, we use the same two bootstrap approaches as in Section \ref{ssec:sim_pt}, but for bootstrapping infimum-based statistics: Parametric Bootstrap (SCB-PB (Inf); see Online Appendix \ref{ssec:SCB_PB_INF}), and Multiplier Bootstrap (SCB-MB (Inf); see Online Appendix \ref{ssec:SCB_MB_INF}).

We compare our two versions of the infimum-based simultaneous confidence band $\doublewidehatCL{\operatorname{SCB}}{}^{\inf}_{1-2\alpha}(t)$ band with two benchmarks: the Naive pointwise $t$-band $\doublewidehatCL{\operatorname{CI}}{}^{\!\text{Naive}}_{1-2\alpha}(t)$, which coincides with the band proposed by \cite{Fogarty_Small_2014} for equivalence testing, and the Bonferroni-corrected $t$-band $\doublewidehatCL{\operatorname{CI}}{}^{\!\text{Bonf}}_{1-2\alpha}(t)$, using a correction of $2\alpha/101$. 

\begin{figure}[!htbp]
	\centering
	\includegraphics[width=0.98\linewidth]{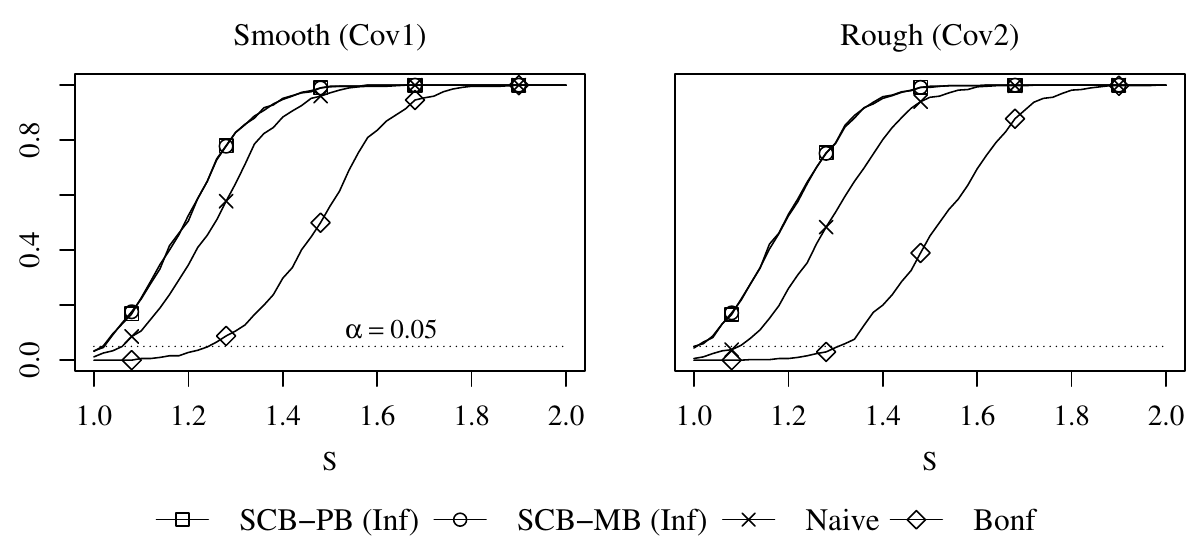}
	\caption[]{Power curves, under violated Assumption I, for $\text{ATT1}^*$, $n=200$, and $T=21$.}
	\label{fig:N200ATT1starPowerCurves_RB}
\end{figure}

Figure \ref{fig:N200ATT1starPowerCurves_RB} shows power curves using $\text{ATT1}^*$ with $n=200$ and $T=21$ for both covariance scenarios. All bands are valid under the global null hypothesis ($S=1$), but the two $\doublewidehatCL{\operatorname{SCB}}{}^{\inf}_{1-2\alpha}(t)$ bands, SCB-PB (Inf) and SCB-MB (Inf), have empirical sizes ($0.032$ and $0.034$ in Cov1; $0.050$ and $0.044$ in Cov2) substantially closer to the nominal level $5\%$ than \cite{Fogarty_Small_2014}'s Naive band ($0.012$ in Cov1; $0.006$ in Cov2) and the Bonferroni band ($0$ in both Cov1 and Cov2). Even when taking the different sizes into account, the two $\doublewidehatCL{\operatorname{SCB}}{}^{\inf}_{1-2\alpha}(t)$ bands clearly power-dominate Naive and Bonferroni bands. Results for $\text{ATT2}^*$ (Online Appendix \ref{SIM:reference_band_validating}) are qualitatively equivalent.

Drawing on the simulation results in Sections \ref{ssec:sim_pt} and \ref{ssec:validation}, there is an important trade-off in choosing the control parameter for honest reference band: a wide reference band facilitates its validation in the pre-anticipation period, but can lead to an overly conservative honest hypothesis testing in the post-treatment period; vice versa for a narrow reference band.

%%%%%%%%%%%%%%%%%%%%%%%%%%%%%%%%%%%%%%%%%%%%
\section{Application}\label{sec:applications} 
%%%%%%%%%%%%%%%%%%%%%%%%%%%%%%%%%%%%%%%%%%%%

In this section we illustrate how researchers can use our honest event study plots in two case studies. Our method is implemented using our \textsf{R}-package \texttt{fdid} \citep{fang_liebl_2025_R}. Users can also adjust their honest reference bands interactively via our Shiny app: \href{https://ccfang2.shinyapps.io/fdidHonestInference/}{https://ccfang2.shinyapps.io/fdidHonestInference/}.

\subsection{The Effect of Working From Home on Work-Home Distance}\label{sec:application_work}

\citet{Coskun_2026} study how the realization of Working-From-Home (WFH) potential during and after the Covid-19 pandemic impacts the labor market and locality choices. As part of their research, \citet{Coskun_2026} employ a dynamic fixed effects model to estimate how much the commuting distance of individuals is dependent on the interaction terms between the WFH potential of occupations and a sequence of dummies indicating the relative event time to the onset of the Covid-19 pandemic.

We augment the classical event study plot (Figure \ref{fig:ESP}) using a supremum-based $(1-\alpha)\times 100\%$ simultaneous confidence band for uniformly testing the causal effects of the WFH potential in the post-treatment period, and an infimum-based $(1-2\alpha)\times 100\%$ simultaneous confidence band for validating the reference band \eqref{eq:ta_bounds} with $S_\ell=1.5$, $S_u=1.8$ and $t_A=-2$ in the pre-anticipation period (Figure \ref{fig:ESP_honest}), where we set $\alpha=0.05$. The chosen reference band takes into account the possible treatment-anticipation bias and can be considered validated at the $\alpha=0.05$ significance level. Using the validated reference band, we find that the treatment effect is honestly and uniformly significant over $t \in [23,44]$. This gives a strong and honest confirmation for the positive effect of the WFH potential on work-home distance.

\subsection{The Effect of Duty-to-Bargain Laws on Employment}\label{sec:application_laws}

In some cases, validating a given reference band can be challenging, as doing so may require selecting a very wide reference band—thereby making subsequent testing in the post-treatment period overly conservative. 
Such non-rejection of the equivalence null hypothesis (Section \ref{sec:validation}) often reflects limited sample size or high variability, and must be viewed as a lack of evidence against the null—not confirmation of it. Thus, a reference band failing to pass the validation can still be used for honest inference when its specification can be supported by domain-specific justification.

%%%%%%
\begin{figure}[!htbp]
	\centering
	\begin{subfigure}[b]{0.48\textwidth} 
		\centering
		\includegraphics[trim={0 0 0 0\linewidth},clip, width=\textwidth]{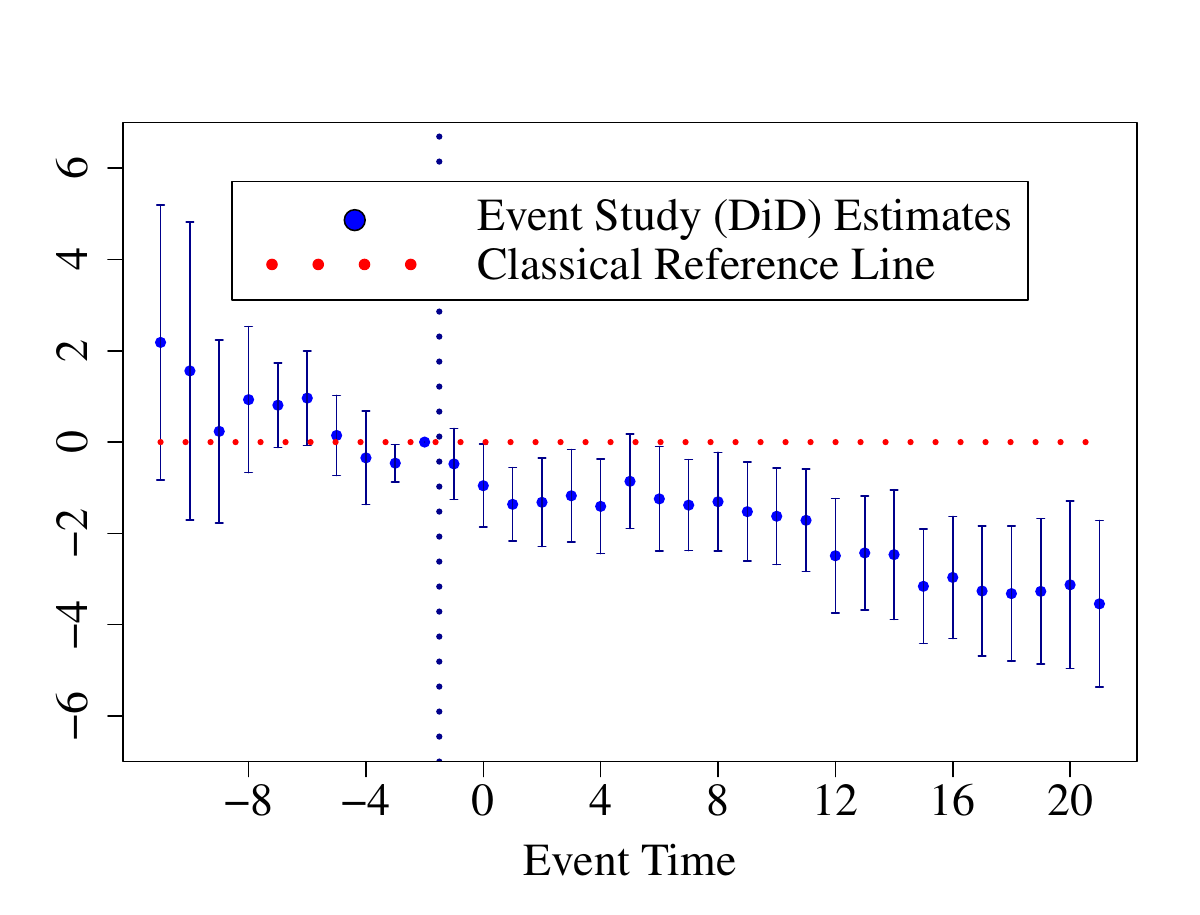} 
		\caption{Classical event study plot with pointwise confidence intervals }\label{fig:ESP_APP2}
	\end{subfigure}\hfill
	\begin{subfigure}[b]{0.48\textwidth} 
		\centering
		\includegraphics[trim={0 0 0 0\linewidth},clip, width=\textwidth]{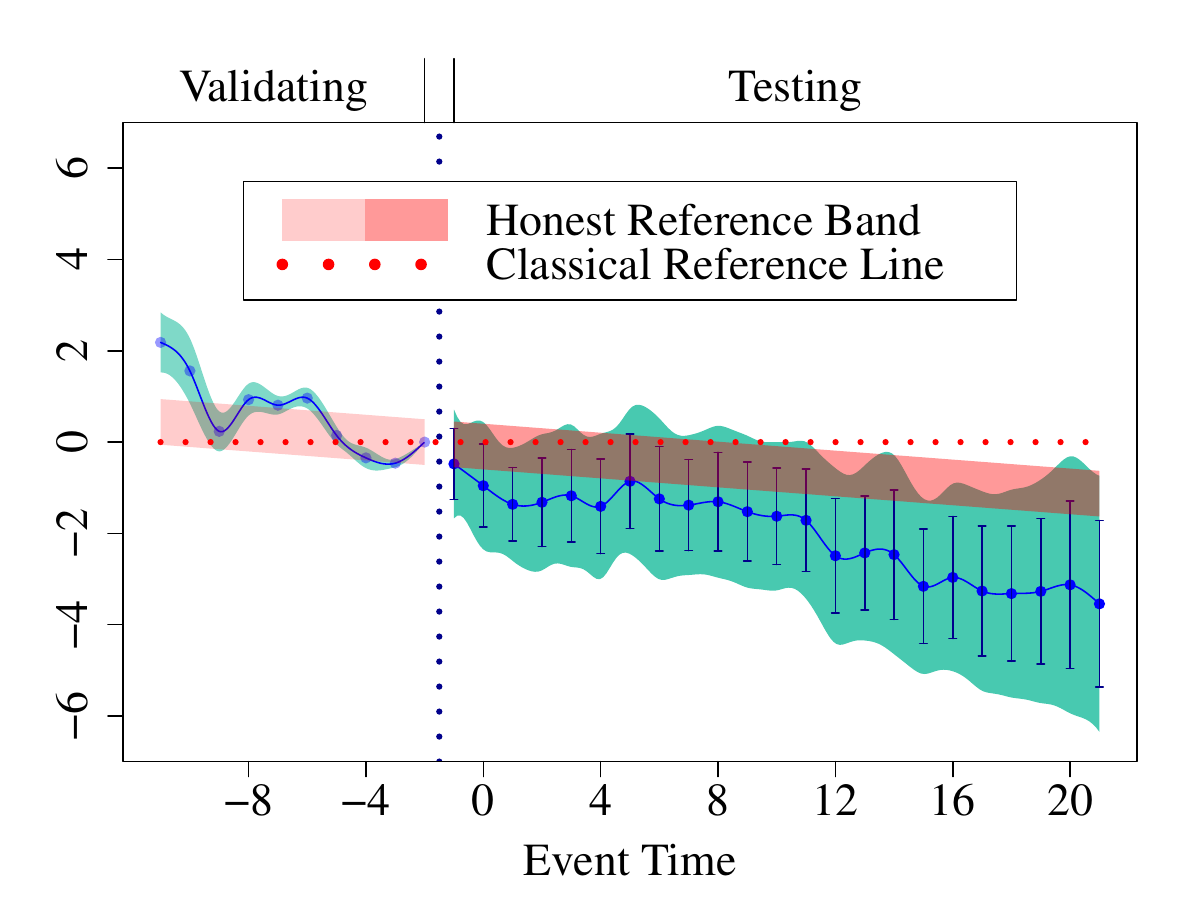} 
		\caption{Honest event study plot with inf- and sup-based simultaneous confidence bands}\label{fig:ESP_honest_APP2}
	\end{subfigure}
        \caption{Effects of DTB laws on employment.}
	\label{fig:ESP_compare_APP2}
\end{figure}
%%%%%%

Figure \ref{fig:ESP_compare_APP2} displays honest causal inference for a study originally published in \citet{lovenheim_willen_2019}, who investigate the effects of Duty-to-Bargain (DTB) laws in the U.S.~education sector on the employment of female workers. In Figure \ref{fig:ESP_APP2}, the original pointwise 95\% confidence intervals suggest that the rollout of DTB laws led to a significant decline in female employment over time, by approximately 1 to 3 percentage points. 

In Figure \ref{fig:ESP_honest_APP2}, we convert the classical event-study plot into an honest event-study plot by adding a supremum-based $(1-\alpha)\times 100\%$ simultaneous confidence band in the post-treatment period and an infimum-based $(1-2\alpha)\times 100\%$ band in the pre-treatment period to validate the reference band \eqref{eq:rmtrb_bounds} with $M_\ell=M_u = 0.5$ (using $\alpha=0.05$). Although the reference band cannot be validated by equivalence testing due to high data variability, it captures the visible downward pre-trend with a width comparable to that of the infimum-based band, providing substantive justification. Using this reference band, we detect no statistically significant treatment effect in the post-treatment period.

%%%%%%%%%%%%%%%%%%%%%%%%%%%%%%%%%%%%%%%%%%%%
\section{Discussion}\label{sec:discussion}
%%%%%%%%%%%%%%%%%%%%%%%%%%%%%%%%%%%%%%%%%%%%

In this paper, we propose reconstructing DiD from a new functional data perspective, which considers the individual time series processes of panel data as smooth processes in continuous time. Our functional estimator of the DiD parameter asymptotically converges to a Gaussian process in the Banach space of continuous functions, facilitating the construction of powerful simultaneous confidence bands. This theoretical result provides a formal foundation for transforming event study plots into rigorous visual tools for honest causal inference via equivalence and relevance testing. In particular, equivalence testing in the pre-anticipation period allows for the validation of the reference bands, while relevance testing in the post-treatment period enables evaluating treatment effects honestly and uniformly.

Notably, we profit from the smoothness assumption on all functional components in our model in many ways. For example, it enables us to interpolate the treatment effect between observable time points, in a belief that the true treatment effect changes smoothly and gradually over time. However, our equivalence and relevance testings can also be easily applied in traditional discrete-data scenario without any smoothness assumption, as indicated in the simulation where the validity of our methodology still holds under rough covariance term.

Our methodology could be extended in several ways. For instance, the assumption of independence between individual realizations of the stochastic processes could be relaxed. In many applications, subjects share a within cluster correlation or are spatially correlated. Adapting our method to such cases would, therefore, be practically relevant. Another non-trivial contribution of the functional data perspective is functional ``registration''. If, for instance, countries have differently fast/slow bureaucratic procedures, there might be heterogeneity in the speed at which treatment effects are realized. Such differences in individual time of the treatment effect development lead to biased treatment effect estimation. Functional registration methods could help to align the functional trajectories. 

\vspace*{1em}

\noindent\textbf{Fundings.} This work was funded by the Deutsche Forschungsgemeinschaft (DFG, German Research Foundation) under Germany's Excellence Strategy -- EXC-2047/1 -- 390685813.

%%%%%%%%%%%%%%%%%%%%%%%%%%%%%
% References
%%%%%%%%%%%%%%%%%%%%%%%%%%%%%
\spacingset{0.8}
\bibliographystyle{Chicago}
\bibliography{bibfile}
%%%%%%%%%%%%%%%%%%%%%%%%%%%%%

% \renewcommand{\thesection}{Appendix~\Roman{section}} 
% \renewcommand{\thesubsection}{\Roman{section}.\arabic{subsection}}
% \setcounter{section}{0}

% \vspace{-1em}
% \spacingset{1.5}

%%%%%%%%%%%%%%%%%%%%%%%%%%%%%%%%%%%%%%%%
% Appendix-Supplement Paper
%%%%%%%%%%%%%%%%%%%%%%%%%%%%%%%%%%%%%%%%
\newpage

\setcounter{page}{1}   
\thispagestyle{plain}
\pagenumbering{Roman}

% \AtBeginDocument{%
  \setlength{\abovedisplayskip}{6pt}
  \setlength{\belowdisplayskip}{6pt}
  \setlength{\abovedisplayshortskip}{5pt}
  \setlength{\belowdisplayshortskip}{5pt}
  \everydisplay{\setstretch{1.3}} % <- THIS is the key line
%}

\appendix

\renewcommand{\thetheorem}{\Roman{theorem}}

\begingroup
\hypersetup{hidelinks, draft}

\counterwithin{figure}{section}
\counterwithin{table}{section}
\counterwithin{equation}{section} 
%\setcounter{equation}{0}
%\vspace*{.5cm}
\spacingset{1.50}

%%%%%%%%%%%%%%%%%%%%%%%%%%%%%%%%%%%%%%%%%%
%\noindent{\large \bf Online supplement for:}\\[2ex]

\if0\blind { \begin{center}
  {\LARGE \bf Online Appendix for\\[1.5ex] ``Making Event Study Plots Honest: A Functional Data Approach to Causal Inference"} 

  \bigskip
  
  \renewcommand*{\thefootnote}{\fnsymbol{footnote}}

  {\large Chencheng Fang\footnotemark[1]\footnotemark[2]\footnotemark[3] \hspace{2.5em} Dominik Liebl\footnotemark[2]\footnotemark[3] \footnotetext[1]{Corresponding author: \href{mailto:ccfang@uni-bonn.de}{ccfang@uni-bonn.de}}
  \footnotetext[2]{Institute of Finance and Statistics, University of Bonn, Adenauerallee 24-42, 53113 Bonn, Germany} \footnotetext[3]{Hausdorff Center for Mathematics, Endenicher Allee 62, 53115 Bonn, Germany}}  

  \renewcommand*{\thefootnote}{\arabic{footnote}}
  \end{center}
} \fi

\if1\blind
{
  \bigskip
  \bigskip
  \bigskip
  \begin{center}
    {\LARGE \bf Supplement to\\[1.5ex] ``Making Event Study Plots Honest: A Functional Data Approach to Causal Inference''}
\end{center}
  \medskip
} \fi

\bigskip

\renewcommand{\thesection}{\Alph{section}}
\renewcommand{\thesubsection}{\Alph{section}.\arabic{subsection}}
\setcounter{section}{0}

%%%%%%%%%%%%%%%%%%%%%%%%%%%%%%%%%%%%%%%%%%%%%%%%%%%%%%%%%%%%%%%%%%%%
%%%%%%%%%%%%%%%%%%%%%%%%%%%%%%%%%%%%%%%%%%%%%%%%%%%%%%%%%%%%%%%%%%%%

This online appendix is a supplement for the paper entitled ``Making Event Study Plots Honest: A Functional Data Approach to Causal Inference". Section \ref{sec:DERIVATIONS} contains our derivations; Sections \ref{sec:PROOFS_Theorems} and \ref{sec:PROOFS_Corollaries} contain proofs; Section \ref{sec:ALGORITHMS} contains our algorithms; and Section \ref{sec:SIM_MORE} contains some additional simulation results. 

%%%%%%%%%%%%%%%%%%%%%%%%%%%%%%%%%%%%%%%%%%%%%%%%%%%%%%%%%%%%%%%%%%
\section{Derivations}\label{sec:DERIVATIONS} 
%%%%%%%%%%%%%%%%%%%%%%%%%%%%%%%%%%%%%%%%%%%%%%%%%%%%%%%%%%%%%%%%%%

%%%%%%%%%%%%%%%%%%%%%%%%%%%%%%%%%%%%%%%%%%%%%%%%%%%%%%%%%%%%%%%%%%
\subsection{Derivation of Equation \eqref{eq:DiD_parameter}}\label{app:did_att_relation} 
%%%%%%%%%%%%%%%%%%%%%%%%%%%%%%%%%%%%%%%%%%%%%%%%%%%%%%%%%%%%%%%%%%
We start with the definition of DiD parameter, and then use the swithing equation to write the observable outcome as potential outcomes. Subsequently, we add and substract some same terms. Finally, we rearrange the terms to derive Equation \eqref{eq:DiD_parameter}.
\begin{align*}
\begin{split}
    \beta(t)
    =&\mathbb{E}[Y_{i}(t)-Y_{i}(-1) \mid D_i=1]- \mathbb{E}[Y_{i}(t)-Y_{i}(-1) \mid D_i=0]\\
    =&\mathbb{E}[Y_{i}(t,1)-Y_{i}(-1,1) \mid D_i=1]- \mathbb{E}[Y_{i}(t,0)-Y_{i}(-1,0) \mid D_i=0] \\
    =&\mathbb{E}[Y_{i}(t,1)-Y_{i}(-1,1) \mid D_i=1]- \mathbb{E}[Y_{i}(t,0)-Y_{i}(-1,0) \mid D_i=0]+ \\
    &\mathbb{E}[Y_{i}(t,0)-Y_{i}(t,0)+Y_{i}(-1,0)-Y_{i}(-1,0)\mid D_i=1] \\
    =&\underbrace{\mathbb{E}[Y_i(t,1)-Y_i(t,0) \mid D_i=1]}_{\theta_{ATT}(t)} -\underbrace{\mathbb{E}[Y_i(-1,1)-Y_i(-1,0) \mid D_i=1]}_{\Delta_{TA}=\theta_{ATT}(-1)}+\\
    &\underbrace{\mathbb{E}[Y_{i}(t,0)-Y_{i}(-1,0) \mid D_i=1]- \mathbb{E}[Y_{i}(t,0)-Y_{i}(-1,0) \mid D_i=0]}_{\Delta_{DT}(t)},
\end{split}
\end{align*} 
where $-\Delta_{TA}=-\theta_{ATT}(-1)$ denotes the bias from the violation of Assumption I (No Anticipation) and $\Delta_{DT}(t)$ denotes the bias from the violation of Assumption II (Parallel Trends).

%%%%%%%%%%%%%%%%%%%%%%%%%%%%%%%%%%%%%%%%%%%%%%%%%%%%%%%%%%%%%%%%%%
\subsection{Derivation of Model \eqref{eq:fct_data_reg_final}}\label{app:SIMPLIFIED_FCT_MODEL} 
%%%%%%%%%%%%%%%%%%%%%%%%%%%%%%%%%%%%%%%%%%%%%%%%%%%%%%%%%%%%%%%%%%

Basically, we only need to apply two-way transformation on Model \eqref{eq:fct_data_reg_simple} to derive Model \eqref{eq:fct_data_reg_final}. 

\textbf{First}, we take average over units on both sides of Model \eqref{eq:fct_data_reg_simple}:
\begin{equation}\label{eq:fct_data_reg_bw}
    \frac{1}{n}\sum_{i=1}^n Y_i(t)=\beta(t)\cdot \frac{1}{n}\sum_{i=1}^n D_i +\frac{1}{n}\sum_{i=1}^n \lambda_i +\phi(t) + \frac{1}{n}\sum_{i=1}^n\varepsilon_i(t).
\end{equation}

\textbf{Second}, we take average over time on both sides of Model \eqref{eq:fct_data_reg_simple}:
\begin{align}\label{eq:fct_data_reg_within}
\begin{split}
    \frac{1}{T_{post}+T_{pre}}\int_{-T_{pre}}^{T_{post}} Y_i(t) \dd t=&\frac{1}{T_{post}+T_{pre}}\int_{-T_{pre}}^{T_{post}} \beta(t)\dd t \cdot D_i+\\
    &\lambda_i+\frac{1}{T_{post}+T_{pre}}\int_{-T_{pre}}^{T_{post}} \phi(t) \dd t+\\
    &\frac{1}{T_{post}+T_{pre}}\int_{-T_{pre}}^{T_{post}} \varepsilon_i(t)\dd t.
\end{split}
\end{align}

\textbf{Third}, we take average over units and time on both sides of Model \eqref{eq:fct_data_reg_simple}:
\begin{align}\label{eq:fct_data_reg_tw}
\begin{split}
    \frac{1}{n}\sum_{i=1}^n\frac{1}{T_{post}+T_{pre}}\int_{-T_{pre}}^{T_{post}} Y_i(t) \dd t=&\frac{1}{T_{post}+T_{pre}}\int_{-T_{pre}}^{T_{post}} \beta(t)\dd t \cdot \frac{1}{n}\sum_{i=1}^nD_i+\\
    &\frac{1}{n}\sum_{i=1}^n\lambda_i+\frac{1}{T_{post}+T_{pre}}\int_{-T_{pre}}^{T_{post}} \phi(t) \dd t+\\
    &\frac{1}{n}\sum_{i=1}^n\frac{1}{T_{post}+T_{pre}}\int_{-T_{pre}}^{T_{post}} \varepsilon_i(t)\dd t.
\end{split}
\end{align}

\textbf{Finally}, by \eqref{eq:fct_data_reg_simple}-\eqref{eq:fct_data_reg_bw}-\eqref{eq:fct_data_reg_within}+\eqref{eq:fct_data_reg_tw}, we have:
\begin{equation*}
    \ddot{Y}_i(t)=\underbrace{\left( \beta(t)-\frac{1}{T_{post}+T_{pre}}\int_{-T_{pre}}^{T_{post}} \beta(s)\dd s \right)}_{\gamma(t)}\dot{D}_i + \ddot{\varepsilon}_i(t), 
\end{equation*}
where 
\begingroup
\allowdisplaybreaks
\begin{align*}
	\ddot{Y}_i(t)&=\dot{Y}_i(t)-\frac{1}{T_{post}+T_{pre}} \int_{-T_{pre}}^{T_{post}} Y_i(s)\dd s+\frac{1}{n}\sum_{i=1}^n \frac{1}{T_{post}+T_{pre}}\int_{-T_{pre}}^{T_{post}} Y_i(s)\dd s,\\
	\ddot{\varepsilon}_i(t)&=\dot{\varepsilon}_i(t)-\frac{1}{T_{post}+T_{pre}} \int_{-T_{pre}}^{T_{post}} \varepsilon_i(s)\dd s+\frac{1}{n}\sum_{i=1}^n \frac{1}{T_{post}+T_{pre}}\int_{-T_{pre}}^{T_{post}} \varepsilon_i(s)\dd s,
	%\text{and}\quad\dot{D}_i&=D_i-\frac{1}{n}\sum_{i=1}^n D_i.
\end{align*}
\endgroup
$\dot{D}_i=D_i-n^{-1}\sum_{i=1}^n D_i$, $\dot{Y}_i(t)=Y_i(t)-\frac{1}{n} \sum_{i=1}^n Y_i(t)$ and $\dot{\varepsilon}_i(t)=\varepsilon_i(t)-\frac{1}{n} \sum_{i=1}^n \varepsilon_i(t)$. Hence, Model \eqref{eq:fct_data_reg_final} holds. As $\gamma(-1)=0-\frac{1}{T_{post}+T_{pre}}\int_{-T_{pre}}^{T_{post}} \beta(s)\dd s$, we then have $\beta(t)=\gamma(t)+\frac{1}{T_{post}+T_{pre}}\int_{-T_{pre}}^{T_{post}} \beta(s)\dd s=\gamma(t)-\gamma(-1)$.

%%%%%%%%%%%%%%%%%%%%%%%%%%%%%%%%%%%%%%%%%%%%%%%%%%%%%%%%%%%%%%%%%%
\subsection{Derivation of Covariates-adjusted Estimators \eqref{eq:BetaHat_Covariate} and \eqref{eq:cov_est_Covariate}} \label{app:cov_est_Covariate}
%%%%%%%%%%%%%%%%%%%%%%%%%%%%%%%%%%%%%%%%%%%%%%%%%%%%%%%%%%%%%%%%%%

\textbf{First}, we start with the functional data model with covariates adjustment, as shown in \eqref{eq:app_fct_on_scalar_CoV_model}:
\begin{equation}\label{eq:app_fct_on_scalar_CoV_model}
	\ddot{Y}_i(t)=\gamma(t)\dot{D}_i +\dot{W}_i^{\top}\breve{\xi}(t) + \ddot{\varepsilon}_i(t).
\end{equation}
By Frisch-Waugh-Lovell (FWL) Theorem \citepappendix{frisch_waugh_1933_app, lovell_1963_app}, the estimate of $\gamma(t)$ in \eqref{eq:app_fct_on_scalar_CoV_model} is equivalent to the estimate in a residual-on-residual regression \eqref{eq:app_CoV_FWL}, as shown below:
\begin{equation}\label{eq:app_CoV_FWL}
    \sum_{j=1}^n L_{ij}\ddot{Y}_j(t)= \gamma(t) \widetilde{D}_i + \sum_{j=1}^n L_{ij}\ddot{\varepsilon}_j(t), 
\end{equation}
where $\widetilde{D}_i=\sum_{j=1}^n L_{ij}\dot{D}_j$ and $L_{ij}$ is the $(i,j)$-th entry of ($n \times n$) matrix $L=I-\dot{W}(\dot{W}^{\top}\dot{W})^{-1}\dot{W}^{\top}$ with $I$ as an ($n \times n$) identity matrix and $\dot{W}=(\dot{W}_1, \dot{W}_2, \dots, \dot{W}_n)^{\top}$ as an ($n \times k$) matrix. 

\textbf{Second}, $\gamma(t)$ in \eqref{eq:app_CoV_FWL} can be estimated by a least squares estimator:
\begin{equation*}
    \widehat{\gamma}^{FWL}_n(t) = \left(\frac{1}{n}\sum_{i=1}^n \widetilde{D}_i^2\right)^{-1}\left(\frac{1}{n}\sum_{i=1}^n \widetilde{D}_i \left(\sum_{j=1}^n L_{ij}\ddot{Y}_j(t)\right)\right),\quad t\in[-T_{pre},-1]\cup[0,T_{post}].
\end{equation*}    
By $\beta(t)=\gamma(t)-\gamma(-1)$, a plug-in estimator for $\beta(t)$ is
\begin{equation}\label{eq:app_BetaHat_Cov_FWL}
    \widehat{\beta}^{FWL}_n(t)= \widehat{\gamma}_n(t)- \widehat{\gamma}_n(-1)=\left(\frac{1}{n}\sum_{i=1}^n \widetilde{D}_i^2\right)^{-1}\left(\frac{1}{n}\sum_{i=1}^n \widetilde{D}_i \left(\sum_{j=1}^n L_{ij} (\ddot{Y}_j(t)-\ddot{Y}_j(-1)) \right)\right)
\end{equation}
Note that, the $\ddot{Y}_j(t)$ and $\ddot{Y}_j(-1)$ expressions in \eqref{eq:app_BetaHat_Cov_FWL} include integration operations, which cancel out due to the subtraction such that
\begin{equation*}
    \widehat{\beta}^{FWL}_n(t)=\left(\frac{1}{n}\sum_{i=1}^n \widetilde{D}_i^2\right)^{-1}\left(\frac{1}{n}\sum_{i=1}^n \widetilde{D}_i (\widetilde{Y}_i(t)-\widetilde{Y}_i(-1) )\right)
\end{equation*}
where $\widetilde{Y}_i(t)=\sum_{j=1}^n L_{ij}\dot{Y}_j(t)$. The covariates-adjusted estimator \eqref{eq:BetaHat_Covariate} for $\beta(t)$ is thus derived.

\textbf{Third}, we proceed to the covariates-adjusted covariance estimator \eqref{eq:cov_est_Covariate}. Analogous to the covariance estimator without covariates in \eqref{eq:cov_est}, we have
\begin{align*}
    \widehat{C}^{FWL}_{\beta,n}(s,t)
    &=\left(\frac{1}{n}\sum_{i=1}^n \widetilde{D}_i^2 \left( \Delta_0 \widetilde{Y}_i(s)\right)\left( \Delta_0 \widetilde{Y}_i(t)\right) \right)\left(\frac{1}{n}\sum_{i=1}^n \widetilde{D}_i^2\right)^{-2},
\end{align*}
where $\Delta_0 \widetilde{Y}_i(t) = (\widetilde{Y}_i(t)-\widetilde{Y}_i(-1))-\widehat{\beta}^{FWL}_n(t)\widetilde{D}_i$ with $\widehat{\beta}^{FWL}_n(t)$ as the covariates-adjusted estimator for $\beta(t)$ defined in \eqref{eq:BetaHat_Covariate}. The covariates-adjusted covariance estimator is therefore derived.

It is important to note that in Model \eqref{eq:app_fct_on_scalar_CoV_model}, we can not estimate the coefficient function $\xi(t)$ directly, but only the demeaned version $\breve{\xi}(t)$. Nonetheless, since
\begin{equation*}
\breve{\xi}(s) - \breve{\xi}(t) = \xi(s) - \xi(t) \quad \text{for all} \quad s, t \in [-T_{pre},-1]\cup[0,T_{post}],
\end{equation*}
we can still consistently estimate the relative shape of $\xi(t)$. Moreover, our primary inference focus lies in the DiD parameter $\beta(t)$, which remains unaffected by the covariates.

%%%%%%%%%%%%%%%%%%%%%%%%%%%%%%%%%%%%%%%%%%%%%%%%%%%%%%%%%%%%%%%%%%
\subsection{Derivation of Aggregated Covariance Function \eqref{eq:cov_stagger}}\label{app:cov_staggered_did}
%%%%%%%%%%%%%%%%%%%%%%%%%%%%%%%%%%%%%%%%%%%%%%%%%%%%%%%%%%%%%%%%%%

We start with the definition of aggregated asymptotic covariance function. Then, we write out the expression of definition. Finally, we obtain the aggregated covariance function \eqref{eq:cov_stagger}.
\begingroup
\allowdisplaybreaks
\begin{align*}
    C_{\beta_A}(e_1,e_2)
    &=\lim_{n\to \infty} n\operatorname{Cov}(\widehat{\beta}_{A}(e_1), \widehat{\beta}_{A}(e_2))\\
    &=\lim_{n\to \infty} n\operatorname{Cov}\left(\sum_{g\in\mathcal{G}} w_g \widehat{\beta}_{g}(e_1), \sum_{g\in\mathcal{G}} w_g \widehat{\beta}_{g}(e_2)\right)\\
    &=\sum_{g\in\mathcal{G}} \lim_{n\to \infty} n\operatorname{Cov}\left(w_g\widehat{\beta}_{g}(e_1), w_g\widehat{\beta}_{g}(e_2)\right)\\
    &=\sum_{g\in\mathcal{G}} w_g^2\lim_{n\to \infty} n\operatorname{Cov}\left(\widehat{\beta}_{g}(e_1), \widehat{\beta}_{g}(e_2)\right)\\
    &=\sum_{g\in\mathcal{G}} w_g^2 C_{\beta_g}(e_1,e_2)
\end{align*}
\endgroup
for $e_1,e_2\in[-T_{pre,A}, T_{post,A}]$. The third line follows from a practice where we use different control units for computing group-specific $\widehat{\beta}_g(e)$; otherwise the resulting expression would be more complicated, taking into account the cross-group covariance.

%%%%%%%%%%%%%%%%%%%%%%%%%%%%%%%%%%%%%%%%%%%%%%%%%%%%%%%%%%%%%%%%%%
\section{Proofs of Theorems}\label{sec:PROOFS_Theorems} 
%%%%%%%%%%%%%%%%%%%%%%%%%%%%%%%%%%%%%%%%%%%%%%%%%%%%%%%%%%%%%%%%%%

%%%%%%%%%%%%%%%%%%%%%%%%%%%%%%%%%%%%%%%%%%%%%%%%%%%%%%%%%%%%%%%%%%
\subsection{Proof of Theorem \ref{thm:beta_equiv}: Equivalence of functional DiD parameter and panel data DiD parameter}
%%%%%%%%%%%%%%%%%%%%%%%%%%%%%%%%%%%%%%%%%%%%%%%%%%%%%%%%%%%%%%%%%%

\begin{proof}
    \textbf{First}, we start with the TWFE panel data model in \eqref{eq:pd_twfe_reg_main} and further simplify it.
    \begingroup
    \allowdisplaybreaks
    \begin{align}\label{eq:pd_twfe_reg_main_simplified}
    \begin{split}
        Y_{it}
        &=\sum_{\substack{s=-T_{pre}\\s\neq -1}}^{T_{post}}\beta^{PD}_s D_{its} + \lambda_i + \phi_t + \epsilon_{it}\\
        &=\sum_{\substack{s=-T_{pre}\\s\neq -1}}^{T_{post}}\beta^{PD}_s D_i \mathbbm{1}_{\{t=s\}} + \lambda_i + \phi_t + \epsilon_{it}\\
        &=\beta_t^{PD} D_i \mathbbm{1}_{\{t \neq -1\}} + \lambda_i + \phi_t + \epsilon_{it},
    \end{split}
    \end{align}
    \endgroup
    where $\beta_t^{PD}$ is the panel data DiD parameter for $t\in\{-T_{pre},\dots,T_{post}\}$ and $\beta^{PD}_{-1}=0$.
    
    \textbf{Second}, we can write the simplified TWFE panel data model in \eqref{eq:pd_twfe_reg_main_simplified} as a function-on-scalar regression model, which directly leads to our functional model in \eqref{eq:fct_data_reg_simple}:
    \begin{equation*}
        Y_{i}(t)=\beta(t) D_i \mathbbm{1}_{\{t \neq -1\}} + \lambda_i + \phi(t) + \epsilon_{i}(t),
    \end{equation*}
    where $\beta(t)$ is our functional DiD parameter and $\beta(-1)=0$.

    \textbf{Third}, by adopting the two-way panel data transformation to functional data, we obtain a new function-on-scalar model, the same as in \eqref{eq:fct_data_reg_final}.
    \begin{equation}\label{eq:fct_data_reg_transf}
	    \ddot{Y}_i(t)=\underbrace{\left( \beta(t)-\frac{1}{T_{post}+T_{pre}}\int_{-T_{pre}}^{T_{post}} \beta(s)\dd s \right)}_{\gamma(t)}\dot{D}_i + \ddot{\varepsilon}_i(t).
    \end{equation}

    \textbf{Finally}, since the regression models \eqref{eq:pd_twfe_reg_main_simplified} and \eqref{eq:fct_data_reg_transf} are pointwise equivalent, we can conclude that the functional DiD parameter $\beta$ is pointwise equivalent to the panel data DiD parameter $\beta^{PD}$.
    \begin{equation*}
        \beta(t)=\beta^{PD}_t \quad\text{for every}\quad t\in\{-T_{pre},\dots,T_{post}\},\quad\text{with}\quad \beta(-1)=\beta^{PD}_{-1}=0. 
    \end{equation*}
    This completes the proof.
\end{proof}

%%%%%%%%%%%%%%%%%%%%%%%%%%%%%%%%%%%%%%%%%%%%%%%%%%%%%%%%%%%%%%%%%%
\subsection{Proof of Theorem \ref{thm:beta_hat_equiv}: Equivalence of functional DiD estimator and panel data DiD estimator}
%%%%%%%%%%%%%%%%%%%%%%%%%%%%%%%%%%%%%%%%%%%%%%%%%%%%%%%%%%%%%%%%%%

\begin{proof}
    \textbf{First}, we start with the panel data DiD estimator $\widehat{\beta}^{PD}_n$ in \eqref{eq:beta_hat_transf_main} and further simplify it.
    \begin{equation*}%\label{eq:beta_hat_transf_main_simplified}
    \widehat{\beta}^{PD}_n = 
    \underbrace{\left(\frac{1}{nT}\sum_{i=1}^n \sum_{t=-T_{pre}}^{T_{post}} \ddot{D}^{PD}_{it} \ddot{D}^{PD^\top}_{it} \right)^{-1}}_{:=P_1^{-1}}\underbrace{\left(\frac{1}{nT}\sum_{i=1}^n \sum_{t=-T_{pre}}^{T_{post}} \ddot{D}_{it}^{PD} \ddot{Y}_{it}^{PD}\right)}_{:=P_2},
\end{equation*}
    where $\ddot{D}_{it}^{PD}=(\ddot{D}^{PD}_{it,-T_{pre}}, \dots ,\ddot{D}^{PD}_{it,-2},\ddot{D}^{PD}_{it,0},\dots,\ddot{D}^{PD}_{it,T_{post}})^{\top}$. Again, we denote $T:=T_{pre}+T_{post}+1$. We simplify $P_1$ and $P_2$ separately. 
    
    \textbf{Second}, to simplify $P_1$, we have 
    \begingroup
    \allowdisplaybreaks
    \begin{align*}
    %\begin{split}
        P_1
        =&\frac{1}{nT}\sum_{i=1}^n \sum_{t=-T_{pre}}^{T_{post}} \ddot{D}^{PD}_{it} \ddot{D}^{PD^\top}_{it}\\
        =&\frac{1}{nT}\sum_{i=1}^n \sum_{t=-T_{pre}}^{T_{post}}
        \begin{pmatrix}
            \ddot{D}^{PD}_{it,-T_{pre}}\ddot{D}^{PD}_{it,-T_{pre}} & \dots & \ddot{D}^{PD}_{it,-T_{pre}}\ddot{D}^{PD}_{it,T_{post}}\\
            \vdots & \ddots & \vdots\\
            \ddot{D}^{PD}_{it,T_{post}}\ddot{D}^{PD}_{it,-T_{pre}} & \dots & \ddot{D}^{PD}_{it,T_{post}} \ddot{D}^{PD}_{it,T_{post}}
        \end{pmatrix}_{(T-1) \times (T-1)}\\
        =&\frac{1}{nT}\sum_{i=1}^n
        \left[
        \begin{pmatrix}
            (1-\frac{1}{T})^2 & \dots & (-\frac{1}{T})(1-\frac{1}{T})\\
            \vdots & \ddots & \vdots\\
            (-\frac{1}{T})(1-\frac{1}{T}) & \dots & (-\frac{1}{T})^2
        \end{pmatrix}\dot{D}_i^2+\dots+
        \begin{pmatrix}
            (-\frac{1}{T})^2 & \dots & (-\frac{1}{T})(1-\frac{1}{T})\\
            \vdots & \ddots & \vdots\\
            (-\frac{1}{T})(1-\frac{1}{T}) & \dots & (1-\frac{1}{T})^2
        \end{pmatrix}\dot{D}_i^2
        \right]\\
        =&\left(\frac{1}{nT}\sum_{i=1}^n \dot{D}_i^2\right)
        \left[
        \begin{pmatrix}
            1 & \dots & 0 \\
            \vdots & \ddots & \vdots \\
            0 & \dots & 1 \\
        \end{pmatrix}_{(T-1)\times(T-1)}
        -
        \begin{pmatrix}
            \frac{1}{T} & \dots & \frac{1}{T} \\
            \vdots & \ddots & \vdots \\
            \frac{1}{T} & \dots & \frac{1}{T} \\
        \end{pmatrix}_{(T-1)\times(T-1)}
        \right]\\
        =&\left(\frac{1}{nT}\sum_{i=1}^n \dot{D}_i^2\right)
        \underbrace{
        \left[
        \underbrace{
        \begin{pmatrix}
            1 & \dots & 0 \\
            \vdots & \ddots & \vdots \\
            0 & \dots & 1 \\
        \end{pmatrix}_{(T-1)\times(T-1)}}_{:=I}
        +
        \underbrace{
        \begin{pmatrix}
            -\frac{1}{\sqrt{T}}\\
            \vdots\\
            -\frac{1}{\sqrt{T}}
        \end{pmatrix}_{(T-1)\times 1}}_{:=U}
        \underbrace{
        \begin{pmatrix}
            \frac{1}{\sqrt{T}} & \dots & \frac{1}{\sqrt{T}}
        \end{pmatrix}_{1 \times (T-1)}}_{:=V^{\top}}
        \right]
        }_{:=Q},
    %\end{split}
    \end{align*}
    \endgroup
    where $I$ is a $(T-1)\times (T-1)$ identity matrix and $Q=I+UV^{\top}$. By Sherman-Morrison formula (\citealpappendix{sherman_morrison_1950}), we have
    \begingroup
    \allowdisplaybreaks
    \begin{align*}
        Q^{-1}
        =&(I+UV^\top)^{-1}\\
        =&I^{-1}+\frac{I^{-1}UV^{\top}I^{-1}}{1+V^{\top}I^{-1}U}\\
        =&I+\frac{UV^{\top}}{1+V^{\top}U}\\
        =&I+\mathbf{1},
    \end{align*}
    \endgroup
    where $\mathbf{1}$ is a $(T-1)\times(T-1)$ matrix with all elements being 1. Hence, we have
    \begin{align*}
    \begin{split}
        P_1^{-1}=\left(\frac{1}{nT}\sum_{i=1}^n \dot{D}_i^2\right)^{-1}(I+\mathbf{1}).
    \end{split}
    \end{align*}

    \textbf{Third}, to simplify $P_2$, we have
    \begingroup
    \allowdisplaybreaks
    \begin{align*}
    %\begin{split}
        P_2
        =&\frac{1}{nT}\sum_{i=1}^n \sum_{t=-T_{pre}}^{T_{post}} \ddot{D}_{it}^{PD} \ddot{Y}_{it}^{PD}\\
        =&\frac{1}{nT}\sum_{i=1}^n
        \left[ 
        \begin{pmatrix}
            \ddot{D}^{PD}_{i,-T_{pre},-T_{pre}}\\
            \vdots\\ 
            \ddot{D}^{PD}_{i,-T_{pre},T_{post}}
        \end{pmatrix}_{(T-1)\times 1}
        \ddot{Y}_{i,-T_{pre}}^{PD}
        +
        \dots
        +
        \begin{pmatrix}
            \ddot{D}^{PD}_{i,T_{post},-T_{pre}}\\
            \vdots\\ 
            \ddot{D}^{PD}_{i,T_{post},T_{post}}
        \end{pmatrix}_{(T-1)\times 1}
        \ddot{Y}_{i,T_{post}}^{PD}
        \right]\\
        =&\frac{1}{nT}\sum_{i=1}^n
        \begin{pmatrix}
            (1-\frac{1}{T})\dot{D}_i\ddot{Y}_{i,-T_{pre}}^{PD} + \dots + (-\frac{1}{T})\dot{D}_i\ddot{Y}_{i,T_{post}}^{PD}\\
            \vdots\\ 
            (-\frac{1}{T})\dot{D}_i\ddot{Y}_{i,-T_{pre}}^{PD} + \dots + (1-\frac{1}{T})\dot{D}_i\ddot{Y}_{i,T_{post}}^{PD}
        \end{pmatrix}_{(T-1) \times 1}\\
        =&\frac{1}{nT}\sum_{i=1}^n \dot{D}_i
        \left[
        \begin{pmatrix}
            \ddot{Y}_{i,-T_{pre}}^{PD}\\
            \vdots\\
            \ddot{Y}_{i,T_{post}}^{PD}
        \end{pmatrix}_{(T-1)\times 1}
        -
        \underbrace{
        \frac{1}{T}\sum_{t=-T_{pre}}^{T_{post}}\ddot{Y}_{it}^{PD}
        }_{=0}
        \right]\\
        =&\frac{1}{nT}\sum_{i=1}^n \dot{D}_i^2
        \begin{pmatrix}
            \beta_{-T_{pre}}-\frac{1}{T}\sum_{s \neq -1}\beta_s\\
            \vdots\\
            \beta_{T_{post}}-\frac{1}{T}\sum_{s \neq -1}\beta_s
        \end{pmatrix}_{(T-1)\times 1}
        +
        \frac{1}{nT}\sum_{i=1}^n \dot{D}_i
        \begin{pmatrix}
            \ddot{\varepsilon}^{PD}_{i,-T_{pre}}\\
            \vdots\\
            \ddot{\varepsilon}^{PD}_{i,T_{post}}
        \end{pmatrix}_{(T-1)\times 1}.
    %\end{split}
    \end{align*}
    \endgroup

    \textbf{Fourth}, with the simplified $P_1$ and $P_2$, we can then simplify $\widehat{\beta}_n^{PD}$,
    \begingroup
    \allowdisplaybreaks
    \begin{align*}
        \widehat{\beta}_n^{PD}
        =&P_1^{-1}P_2\\
        =&
        \begin{pmatrix}
            \beta_{-T_{pre}}-\frac{1}{T}\sum_{s \neq -1}\beta_s\\
            \vdots\\
            \beta_{T_{post}}-\frac{1}{T}\sum_{s \neq -1}\beta_s
        \end{pmatrix}_{(T-1)\times 1}
        +
        \left(\frac{1}{nT}\sum_{i=1}^n \dot{D}_i^2\right)^{-1}
        \begin{pmatrix}
            \frac{1}{nT}\sum_{i=1}^n \ddot{\varepsilon}^{PD}_{i,-T_{pre}}\dot{D}_i\\
            \vdots\\
            \frac{1}{nT}\sum_{i=1}^n \ddot{\varepsilon}^{PD}_{i,T_{post}}\dot{D}_i
        \end{pmatrix}_{(T-1)\times 1}\\
        &+
        \begin{pmatrix}
            \sum_{s\neq -1}\beta_s-\frac{T-1}{T}\sum_{s\neq -1}\beta_s\\
            \vdots\\
            \sum_{s\neq -1}\beta_s-\frac{T-1}{T}\sum_{s\neq -1}\beta_s
        \end{pmatrix}_{(T-1)\times 1}
        +
        \left(\frac{1}{nT}\sum_{i=1}^n \dot{D}_i^2\right)^{-1}
        \begin{pmatrix}
            \frac{1}{nT}\sum_{i=1}^n \sum_{s\neq -1} \ddot{\varepsilon}^{PD}_{is}\dot{D}_i\\
            \vdots\\
            \frac{1}{nT}\sum_{i=1}^n \sum_{s\neq -1} \ddot{\varepsilon}^{PD}_{is}\dot{D}_i
        \end{pmatrix}_{(T-1)\times 1}\\
        =&
        \begin{pmatrix}
            \beta_{-T_{pre}}\\
            \vdots\\
            \beta_{T_{post}}
        \end{pmatrix}_{(T-1)\times 1}
        +
        \left(\frac{1}{nT}\sum_{i=1}^n \dot{D}_i^2\right)^{-1}
        \begin{pmatrix}
            \frac{1}{nT}\sum_{i=1}^n \ddot{\varepsilon}^{PD}_{i,-T_{pre}}\dot{D}_i+\frac{1}{nT}\sum_{i=1}^n \sum_{s\neq -1} \ddot{\varepsilon}^{PD}_{is}\dot{D}_i\\
            \vdots\\
            \frac{1}{nT}\sum_{i=1}^n \ddot{\varepsilon}^{PD}_{i,T_{post}}\dot{D}_i+\frac{1}{nT}\sum_{i=1}^n \sum_{s\neq -1} \ddot{\varepsilon}^{PD}_{is}\dot{D}_i
        \end{pmatrix}_{(T-1)\times 1}\\
        =&
        \begin{pmatrix}
            \beta_{-T_{pre}}\\
            \vdots\\
            \beta_{T_{post}}
        \end{pmatrix}_{(T-1)\times 1}
        +
        \left(\frac{1}{n}\sum_{i=1}^n \dot{D}_i^2\right)^{-1} \left[\frac{1}{n}\sum_{i=1}^n \dot{D}_i
        \begin{pmatrix}
            \ddot{\varepsilon}^{PD}_{i,-T_{pre}}+\sum_{s\neq -1} \ddot{\varepsilon}^{PD}_{is}\\
            \vdots\\
            \ddot{\varepsilon}^{PD}_{i,T_{post}}+\sum_{s\neq -1} \ddot{\varepsilon}^{PD}_{is}
        \end{pmatrix}_{(T-1)\times 1}\right].
    \end{align*}
    \endgroup
    
    \textbf{Fifth}, we continue with the functional DiD estimator $\widehat{\beta}_n(t)$ in \eqref{eq:BetaHat} and simplify it,
    \begingroup
    \allowdisplaybreaks
    \begin{align*}
        \widehat{\beta}_n(t)
        =&\widehat{\gamma}_n(t) -\widehat{\gamma}_n(-1) \\
	=&\left(\frac{1}{n}\sum_{i=1}^n \dot{D}_i^2\right)^{-1}\left(\frac{1}{n}\sum_{i=1}^n \dot{D}_i (\dot{Y}_i(t) - \dot{Y}_i(-1))\right)\\
	=&\left(\frac{1}{n}\sum_{i=1}^n \dot{D}_i^2\right)^{-1}\left(\frac{1}{n}\sum_{i=1}^n \dot{D}_i \left(Y_i(t)-Y_i(-1)-\frac{1}{n}\sum_{i=1}^n Y_i(t)+\frac{1}{n}\sum_{i=1}^nY_i(-1)\right)\right)\\
        =&\left(\frac{1}{n}\sum_{i=1}^n \dot{D}_i^2\right)^{-1}\left(\frac{1}{n}\sum_{i=1}^n \dot{D}_i \left( \beta(t)\dot{D}_i+\varepsilon_i(t)-\varepsilon_i(-1)-\frac{1}{n}\sum_{i=1}^n \varepsilon_i(t)+\frac{1}{n}\sum_{i=1}^n \varepsilon_i(-1) \right) \right)\\
        =&\beta(t)+ \left(\frac{1}{n}\sum_{i=1}^n \dot{D}_i^2\right)^{-1}\left( \frac{1}{n}\sum_{i=1}^n \dot{D}_i \left( \varepsilon_i(t)-\varepsilon_i(-1)-\frac{1}{n}\sum_{i=1}^n \varepsilon_i(t)+\frac{1}{n}\sum_{i=1}^n \varepsilon_i(-1)\right) \right).
    \end{align*}
    \endgroup
    
    \textbf{Sixth}, we have
    \begingroup
    \allowdisplaybreaks
    \begin{align*}
        \ddot{\varepsilon}^{PD}_{it}+\sum_{s\neq -1} \ddot{\varepsilon}^{PD}_{is}
        =&\varepsilon_{it}-\frac{1}{n}\sum_{i=1}^n \varepsilon_{it}-\frac{1}{T}\sum_{t=-T_{pre}}^{T_{post}} \varepsilon_{it}+\frac{1}{nT}\sum_{i=1}^n \sum_{t=-T_{pre}}^{T_{post}} \varepsilon_{it}+\\
        &\sum_{s\neq -1}\left( \varepsilon_{is}-\frac{1}{n}\sum_{i=1}^n \varepsilon_{is}-\frac{1}{T}\sum_{s=-T_{pre}}^{T_{post}} \varepsilon_{is}+\frac{1}{nT}\sum_{i=1}^n \sum_{s=-T_{pre}}^{T_{post}} \varepsilon_{is} \right)\\
        =&\varepsilon_{it}-\frac{1}{n}\sum_{i=1}^n \varepsilon_{it}-\frac{1}{T}\sum_{t=-T_{pre}}^{T_{post}} \varepsilon_{it}+\frac{1}{nT}\sum_{i=1}^n \sum_{t=-T_{pre}}^{T_{post}} \varepsilon_{it}+\\
        &\sum_{s\neq -1}\varepsilon_{is}-\sum_{s\neq -1}\frac{1}{n}\sum_{i=1}^n \varepsilon_{is}-(T-1)\frac{1}{T}\sum_{s=-T_{pre}}^{T_{post}} \varepsilon_{is}+(T-1)\frac{1}{nT}\sum_{i=1}^n \sum_{s=-T_{pre}}^{T_{post}} \varepsilon_{is} \\
        =&\varepsilon_{it}-\frac{1}{n}\sum_{i=1}^n \varepsilon_{it}-\frac{T}{T}\sum_{t=-T_{pre}}^{T_{post}} \varepsilon_{it}+\frac{T}{nT}\sum_{i=1}^n \sum_{t=-T_{pre}}^{T_{post}} \varepsilon_{it}+\sum_{s\neq -1}\varepsilon_{is}-\sum_{s\neq -1}\frac{1}{n}\sum_{i=1}^n \varepsilon_{is}\\
        =&\varepsilon_{it}-\frac{1}{n}\sum_{i=1}^n \varepsilon_{it}-\varepsilon_{i,-1}+\frac{1}{n}\sum_{i=1}^n \varepsilon_{i,-1}.
    \end{align*}
    \endgroup
    This shows the equivalence of final expressions in $\widehat{\beta}_{n,t}^{PD}$ and $\widehat{\beta}_n(t)$ for every $t\in\{-T_{pre},\dots,T_{post}\}$. 
    
    \textbf{Finally}, 
    $
        \widehat\beta^{PD}_{n,t}=\widehat\beta_n(t)\quad\text{for every}\quad t\in\{-T_{pre},\dots,T_{post}\},\quad\text{with}\quad \widehat\beta^{PD}_{n,-1}=\widehat\beta_n(-1)=0. 
    $
    This completes the proof.
\end{proof}

%%%%%%%%%%%%%%%%%%%%%%%%%%%%%%%%%%%%%%%%%%%%%%%%%%%%%%%%%%%%%%%%%%
\subsection{Proof of Theorem \ref{thm:pw_normality}: Pointwise Asymptotic Normality}\label{app:pw_normality}
%%%%%%%%%%%%%%%%%%%%%%%%%%%%%%%%%%%%%%%%%%%%%%%%%%%%%%%%%%%%%%%%%%

\begin{proof}
\textbf{First}, to show the pointwise normality of $\widehat{\beta}_n(t)$, we should show the pointwise asymptotic distribution of 
$
    \widehat{\gamma}_n(t) = \gamma(t) + (n^{-1}\sum_{i=1}^n \dot{D}_i^2)^{-1}(n^{-1}\sum_{i=1}^n \dot{D}_i \ddot{\varepsilon}_i(t)).
$
To derive the asymptotic distribution, we write
\begin{equation*}
    \sqrt{n}(\widehat{\gamma}_n(t)-\gamma(t))=\left(\frac{1}{n}\sum_{i=1}^n \dot{D}_i^2\right)^{-1}\left(\frac{1}{\sqrt{n}}\sum_{i=1}^n \dot{D}_i \ddot{\varepsilon}_i(t)\right).
\end{equation*}

\textbf{Second}, by Law of Large Numbers and Continuous Mapping Theorem, Assumptions \ref{assumption: moments a} implies that
\begin{equation*}
    \left(\frac{1}{n}\sum_{i=1}^n \dot{D}_i^2\right)^{-1} \stackrel{p}{\to}\mathbb{E}[\dot{D}^2]^{-1}.
\end{equation*}
Despite the fact that $\dot{D}_i=D_i-n^{-1}\sum_{i=1}^n D_i$ involves the sample mean, making $\{\dot{D}_i\}_{i=1}^n$ not independent across $i$ and failing the assumption for Law of Large Numbers, the above convergence result still holds, because $\dot{D}_i-\dot{D}_i^\circ=o_p(1)$, where $\dot{D}_i^\circ=D_i-\mathbb{E}[D]$ and $\{\dot{D}_i^\circ\}_{i=1}^n$ are independent across $i$. The same logic applies to the Central Limit Theorem imposed on demeaned variables below.

We also know $\mathbb{E}[\dot{D}\ddot{\varepsilon}(t)]=0$. To apply Central Limit Theorem on the term $\frac{1}{\sqrt{n}}\sum_{i=1}^n \dot{D}_i \ddot{\varepsilon}_i(t)$, we need to ensure the variance of $\dot{D} \ddot{\varepsilon}(t)$ is finite. It is noted that, by Cauchy-Schwarz inequality, we have
$
    Var[\dot{D} \ddot{\varepsilon}(t)]=\mathbb{E}[\dot{D}^2 \ddot{\varepsilon}^2(t)]\le \sqrt{\mathbb{E}[\dot{D}^4] \mathbb{E}[\ddot{\varepsilon}^4(t)]}.
$
It is easy to show that $\mathbb{E}[\dot{D}^4]<\infty$ is implied by $\mathbb{E}[D^4]<\infty$, and $\mathbb{E}[\ddot{\varepsilon}^4(t)]<\infty$ is implied by $\mathbb{E}[D^4]<\infty$ and $\mathbb{E}[Y^4(t)]<\infty$. Hence, by Central Limit Theorem, we can conclude that
\begin{equation*}
    \frac{1}{\sqrt{n}}\sum_{i=1}^n \dot{D}_i \ddot{\varepsilon}_i(t) \stackrel{d}{\to} \mathcal{N}(0, \mathbb{E}[\dot{D}^2\ddot{\varepsilon}^2(t)])
\end{equation*}
for each $t\in[-T_{pre},T_{post}]$. Then, by Slutzky's Theorem, we have
\begin{equation*}
    \sqrt{n}(\widehat{\gamma}_n(t)-\gamma(t)) \stackrel{d}{\to} \mathcal{N}(0,C_{\gamma}(t,t))
\end{equation*}
pointwise for each $t\in[-T_{pre},T_{post}]$, where $C_{\gamma}(t,t)=\mathbb{E}[\dot{D}^2\ddot{\varepsilon}^2(t)]\mathbb{E}[\dot{D}^2]^{-2}$. 

\textbf{Third}, based on the asymptotic distribution of $\widehat{\gamma}_n(t)$, we could derive the asymptotic distribution of $\widehat{\beta}_n(t)=\widehat{\gamma}_n(t)-\widehat{\gamma}_n(-1)$ by Continuous Mapping Theorem. We have
\begin{equation*}
    \sqrt{n}(\widehat{\beta}_n(t)-\beta(t)) \stackrel{d}{\to} \mathcal{N}(0, C_{\beta}(t,t))
\end{equation*}
pointwise for each $t\in[-T_{pre},T_{post}]$, where 
$
    C_{\beta}(t,t)
    =C_{\gamma}(t,t)+C_{\gamma}(-1,-1)-2C_{\gamma}(t,-1)
    =\mathbb{E}[\dot{D}^2(\ddot{\varepsilon}(t)-\ddot{\varepsilon}(-1))^2]\mathbb{E}[\dot{D}^2]^{-2}
    =\mathbb{E}[\dot{D}^2(\dot{\varepsilon}(t)-\dot{\varepsilon}(-1))^2]\mathbb{E}[\dot{D}^2]^{-2}
$
with $\dot{\varepsilon}(t)=\varepsilon(t)-\mathbb{E}[\varepsilon(t)]$. This completes the proof.
\end{proof}

%%%%%%%%%%%%%%%%%%%%%%%%%%%%%%%%%%%%%%%%%%%%%%%%%%%%%%%%%%%%%%%%%%
\subsection{Proof of Theorem \ref{thm:gauss}: Uniform Asymptotic Normality of the Oracle Estimator \eqref{eq:BetaHat}}\label{app:gaussian_process}
%%%%%%%%%%%%%%%%%%%%%%%%%%%%%%%%%%%%%%%%%%%%%%%%%%%%%%%%%%%%%%%%%%
We first show the proof of Theorem \ref{thm:gauss_general} which has a more generic notation, and then we argue that Theorem \ref{thm:gauss_general} implies Theorem \ref{thm:gauss}.

\spacingset{1.2}
\begin{theorem}\label{thm:gauss_general}
    Let $\{X_i\}_{i=1}^n \stackrel{\text{i.i.d.}}{\sim} X$ with $X \in C^2[-1,1]$ almost surely. $X$ is $p \; (\ge1)$ dimensional vector of functions. Assume $0<\mathbb{E}[\sup_{t \in (-1,1)} X^{\prime}(t)^2]< \infty$. Let $Z_n=\frac{1}{\sqrt{n}}\sum_{i=1}^n \widetilde{X}_i$, where $\widetilde{X}_i=X_i-\mathbb{E}[X]$. If for each $t\in [-1,1]$, $Z_n(t) \stackrel{d}{\to} \mathcal{N}_p(0, \mathbb{E}[\widetilde{X}(t)^2])$. Then, we have 
    \begin{equation*}
    Z_n \stackrel{d}{\to} \mathcal{GP}(0, C_Z),
    \end{equation*}
    where the covariance function $C_Z(s,t)=\mathbb{E}[\widetilde{X}(s) \widetilde{X}(t)]$ for $s,t \in[-1,1]$.
\end{theorem}
\spacingset{1.5}

\begin{proof}
    \textbf{First}, by assumption, we also have $\widetilde{X}\in C^2[-1,1]$ almost surely. By Mean Value Theorem, there exists $\xi \in (s,t)$ with $s,t\in[-1,1]$ such that
    \begin{equation}\label{eq:mvt}
    \widetilde{X}(t)-\widetilde{X}(s)=\widetilde{X}^{\prime}(\xi)\cdot(t-s),
    \end{equation}
    where $\widetilde{X}^\prime$ is the derivative of $\widetilde{X}$. By taking squares on both sides of \eqref{eq:mvt}, we have
    \begin{equation}\label{eq:squared_mvt}
    (\widetilde{X}(t)-\widetilde{X}(s))^2=\widetilde{X}^{\prime}(\xi)^2 \cdot (t-s)^2\le \sup_{\xi \in (s,t)}\widetilde{X}^{\prime}(\xi)^2 \cdot (t-s)^2.
    \end{equation}
    By taking expectation on both sides of \eqref{eq:squared_mvt} w.r.t $\widetilde{X}$, we have
    \begin{align*}
    \begin{split}
    \mathbb{E}\left[(\widetilde{X}(t)-\widetilde{X}(s))^2\right] 
    &\le \mathbb{E}\left[ \sup_{\xi \in (s,t)}\widetilde{X}^{\prime}(\xi)^2 \cdot (t-s)^2 \right]\\
    &= \mathbb{E}\left[ \sup_{\xi \in (s,t)}\widetilde{X}^{\prime}(\xi)^2\right]\cdot(t-s)^2\\
    &:= K \cdot (t-s)^2\\
    &:= f(|t-s|),
    \end{split}
    \end{align*}
    where $K$ is a finite constant. We know that $\mathbb{E}[ \sup_{\xi \in (s,t)}\widetilde{X}^{\prime}(\xi)^2]$ is finite from the assumption that $\mathbb{E}[\sup_{t \in (-1,1)} X^{\prime}(t)^2]$ is finite. Denote $y=|t-s|$, and then we have $f(y)=Ky^2$. Since $y=|t-s| \le 2$ for $t,s \in [-1,1]$, let us define $f(y)=0$ for $y>2$. We find that $f$ is a non-decreasing function around zero. Also, we have
    \begin{align*}
        \int_0^{\infty} y^{-3/2} \sqrt{f(y)} \, dy = \sqrt{K}\int_0^2 y^{-3/2}y \, dy=2\sqrt{2K} < \infty.
    \end{align*}
    Then, by Theorem 2.3 in \citetappendix{Hahn_1977}, $\widetilde{X}$ is mean-square continuous in a sense that
    \begin{align}\label{eq:ms_cont}
        \left|\widetilde{X}(t)-\widetilde{X}(s)\right| \le A \cdot \phi(|t-s|),
    \end{align}
    where $\phi$ is a non-decreasing, continuous function that only depends on $f$ and $\phi(0)=0$. $A$ is a random variable with bounded variance $Var(A)=\sigma^2_A(\phi, f)<\infty$. 
    
    \textbf{Second}, we have
    \begingroup
    \allowdisplaybreaks
    \begin{align*}
        \left|Z_n(t)-Z_n(s)\right|
        & \le \frac{1}{\sqrt{n}}\sum_{i=1}^n \left|\widetilde{X}_i(t)-\widetilde{X}_i(s)\right| \\
        & \le \frac{1}{\sqrt{n}} \sum_{i=1}^n A_i \cdot \phi(|t-s|)\\
        & := A_n\cdot \phi(|t-s|),
    \end{align*}
    \endgroup
    where $A_n$ is a random variable with bounded variance $Var(A_n)=\frac{1}{n}\cdot nVar(A)=\sigma_A^2(\phi, f)<\infty$. The mean of $A_n$ should also be bounded. Hence, it is easy to show $\mathbb{E}[A_n^2]=Var(A_n)+\mathbb{E}[A_n]^2<\infty$ is bounded. By generalized Markov inequality, we have, for any $M>0$,
    \begin{equation*}
        P(|A_n|\ge M)\le \frac{\mathbb{E}[A_n^2]}{M^2} < \infty.
    \end{equation*}
    By replacing $\mathbb{E}[A_n^2]/M^2$ with $\epsilon$, we can have, for any $\epsilon>0$, there exists an $M_{\epsilon}=\sqrt{\mathbb{E}[A_n^2]/\epsilon}<\infty$ such that
    \begin{equation*}
        P(|A_n|\ge M_{\epsilon})\le \epsilon.
    \end{equation*}
    This implies $A_n=O_p(1)$. Therefore, $Z_n$ is equicontinuous. By assumption, $Z_n(t)=O_p(1)$ for any $t\in[-1,1]$. Then, by Theorem 7.2 in \citetappendix{billingsley_1999}, we know that $Z_n$ is tight. 
    
    \textbf{Third}, by Theorem 7.1 in \citetappendix{billingsley_1999}, the tightness of $Z_n$ and its pointwise normality imply that it is asymptotically a mean-zero Gaussian process with covariance function $C_Z(s,t)=\mathbb{E}[\widetilde{X}(s), \widetilde{X}(t)]$. This completes the proof.
\end{proof}

\noindent \centerline{ \textbf{\MakeUppercase{Proof of Theorem \ref{thm:gauss}}}}
\begin{proof}
    \textbf{First}, in line with the notation of Theorem \ref{thm:gauss_general}, let us denote $X_i(t):=\dot{D}_i(\ddot{\varepsilon}_i(t)-\ddot{\varepsilon}_i(-1))=\dot{D}_i(\dot{\varepsilon}_i(t)-\dot{\varepsilon}_i(-1))$ for $t\in[-T_{pre},-1]\cup[0,T_{post}]$. We can then denote
    $$
    Z_n(t)
    = \frac{1}{\sqrt{n}}\sum_{i=1}^n \left(X_i(t)-\mathbb{E}[X(t)] \right)
    = \frac{1}{\sqrt{n}}\sum_{i=1}^n \dot{D}_i(\dot{\varepsilon}_i(t)-\dot{\varepsilon}_i(-1)).
    $$
    It is easy to show that for each $t\in[-T_{pre},-1]\cup[0,T_{post}]$,
    $
        Z_n(t) \stackrel{d}{\to} \mathcal{N}(0, \mathbb{E}[\dot{D}^2(\dot{\varepsilon}(t)-\dot{\varepsilon}(-1) )^2])
    $ with $\dot{\varepsilon}(t)=\varepsilon(t)-\mathbb{E}[\varepsilon(t)]$. Despite the fact that $\dot{D}_i$ and $\dot{\varepsilon}_i(t)$ involve the sample mean, making $\{\dot{D}_i(\dot{\varepsilon}_i(t)-\dot{\varepsilon}_i(-1))\}_{i=1}^n$ not independent across $i$ and failing the assumption for Central Limit Theorem, the convergence result still holds, because $\dot{D}_i-\dot{D}_i^\circ=o_p(1)$ and $\sup_{t\in[-T_{pre},-1]\cup[0,T_{post}]}|\dot{\varepsilon}_i(t)-\dot{\varepsilon}_i^\circ(t)|=o_p(1)$, where $\dot{D}_i^\circ=D_i-\mathbb{E}[D]$, $\dot{\varepsilon}_i^\circ(t)=\varepsilon_i(t)-\mathbb{E}[\varepsilon(t)]$ and $\{\dot{D}_i^\circ(\dot{\varepsilon}_i^\circ(t)-\dot{\varepsilon}_i^\circ(-1))\}_{i=1}^n$ are independent across $i$.
    
    \textbf{Second}, Assumptions \ref{assumption: moments b}, \ref{assumption: smooth a} and \ref{assumption: smooth b} would imply that $\mathbb{E}[\sup_{t \in (-T_{pre},-1)\cup(0,T_{post})} X'(t)^2] < \infty$ and $X\in C^2([-T_{pre},-1]\cup[0,T_{post}])$. Hence, by Theorem \ref{thm:gauss_general}, we can show
    $
        Z_n \stackrel{d}{\to} \mathcal{GP}(0,C_Z),
    $
    where $C_Z=\{C_Z(s,t): s,t \in [-T_{pre},-1]\cup[0,T_{post}]\}$ and $C_Z(s,t)=\mathbb{E}[\dot{D}^2 (\dot{\varepsilon}(s)-\dot{\varepsilon}(-1))(\dot{\varepsilon}(t)-\dot{\varepsilon}(-1))]$.

    \textbf{Third}, to derive the asymptotic stochastic process of $\widehat{\beta}_n(t)$, we write 
    \begin{equation*}
        \sqrt{n}(\widehat{\beta}_n(t)-\beta(t))=\left(\frac{1}{n}\sum_{i=1}^n \dot{D}_i^2\right)^{-1}\left(\frac{1}{\sqrt{n}}\sum_{i=1}^n \dot{D}_i (\dot{\varepsilon}_i(t)-\dot{\varepsilon}_i(-1))\right)=\left(\frac{1}{n}\sum_{i=1}^n \dot{D}_i^2\right)^{-1} Z_n(t)
    \end{equation*}
    for all $t\in[-T_{pre},-1]\cup[0,T_{post}]$. In the proof of Theorem \ref{thm:pw_normality}, we have shown that
    $
        \left(\frac{1}{n}\sum_{i=1}^n \dot{D}_i^2\right)^{-1} \stackrel{p}{\to}\mathbb{E}[\dot{D}^2]^{-1}.
    $
    By Slutzky's Theorem, we have
    \begin{equation*}
        \sqrt{n}(\widehat{\beta}_n-\beta) \stackrel{d}{\to} \mathcal{GP}(0,C_{\beta}),
    \end{equation*}
    where $C_{\beta}=\{ C_{\beta}(s,t): s,t \in [-T_{pre},-1]\cup[0,T_{post}] \}$ and $C_{\beta}(s, t) =\mathbb{E}[\dot{D}^2(\dot{\varepsilon}(s)-\dot{\varepsilon}(-1))(\dot{\varepsilon}(t)-\dot{\varepsilon}(-1))]\mathbb{E}[\dot{D}^2]^{-2}$. This completes the proof. 
\end{proof}

%%%%%%%%%%%%%%%%%%%%%%%%%%%%%%%%%%%%%%%%%%%%%%%%%%%%%%%%%%%%%%%%%%
\subsection{Proof of Theorem \ref{thm:consistency_cov}: Uniform Consistency of Empirical Covariance}\label{app:consistency_cov}
%%%%%%%%%%%%%%%%%%%%%%%%%%%%%%%%%%%%%%%%%%%%%%%%%%%%%%%%%%%%%%%%%%

We first show the proof of Theorem \ref{thm:UniformConvergenceKernel} which has a more generic notation, and then we argue that Theorem \ref{thm:UniformConvergenceKernel} implies Theorem \ref{thm:consistency_cov}.

\spacingset{1.2}
\begin{theorem}\label{thm:UniformConvergenceKernel}
    Let $\{X_i\}_{i=1}^n \stackrel{\text{i.i.d.}}{\sim} X$ with $X \in C^2[-1,1]$ almost surely. $X$ is $p$ ($\ge 1$) dimensional vector of functions. Let $Z_n=\frac{1}{\sqrt{n}}\sum_{i=1}^n \widetilde{X}_i$, where $\widetilde{X}_i=X_i-\mathbb{E}[X]$. Assume $Z_n$ is asymptotically a mean-zero Gaussian process with covariance function $C_Z(s,t)=\mathbb{E}[\widetilde{X}(s)\widetilde{X}(t)]$ for $s,t \in [-1,1]$. Also, assume $\mathbb{E}[\sup_{t \in [-1, 1]} X(t)^2] < \infty$ and $\mathbb{E}[\sup_{t \in (-1, 1)} X'(t)^2] < \infty$. Let $\widehat{C}_{Z,n}(s, t) = n^{-1}\sum_{i = 1}^n \widetilde{X}_i(s) \widetilde{X}_i(t)$ denote the empirical covariance function. Then, we have 
    \begin{align*}
        \sup_{s,t\in[-1,1]}\left| \widehat{C}_{Z,n}(s,t) - C_Z(s,t) \right| &\overset{a.s.}{\to} 0 \, .
    \end{align*}
\end{theorem}
\spacingset{1.5}

\begin{proof}
    \textbf{First}, we need to show that the empirical covariance $\widehat{C}_{Z,n}$ is strongly stochastically equicontinuous, i.e.
    \begin{equation*}
        \left|\widehat{C}_{Z,n}(s, t) - \widehat{C}_{Z,n}(u, v)\right| \leq B_n h\left(\sqrt{(s - u)^2 + (t -
        v)^2}\right)
    \end{equation*}
    almost surely for all $s, t, u, v \in [-1, 1]$ with $h(x) \downarrow 0$ as $x \downarrow 0$, and $B_n$ is a positive stochastic sequence independent of $s, t, u, v$. To show it, pick $s, t, u, v \in [-1, 1]$ arbitrary. We have
    \begin{align*}
    \begin{split}
        &\left|\frac{1}{n} \sum_{i = 1}^n \widetilde{X}_i(s) \widetilde{X}_i(t) - \frac{1}{n} \sum_{i = 1}^n \widetilde{X}_i(u) \widetilde{X}_i(v)\right| \\
        =& \left|\frac{1}{n} \sum_{i = 1}^n (\widetilde{X}_i(s) - \widetilde{X}_i(u)) \widetilde{X}_i(t) +
        \frac{1}{n} \sum_{i = 1}^n (\widetilde{X}_i(t) - \widetilde{X}_i(v)) \widetilde{X}_i(u)\right| \\
        \leq& \underbrace{\frac{1}{n} \sum_{i = 1}^n \left|\widetilde{X}_i(s) - \widetilde{X}_i(u)\right| \left|\widetilde{X}_i(t)\right|}_{:=I_1} +
        \underbrace{\frac{1}{n} \sum_{i = 1}^n \left|\widetilde{X}_i(t) - \widetilde{X}_i(v)\right| \left|\widetilde{X}_i(u)\right|}_{:=I_2} \,.
    \end{split}
    \end{align*}
    The third line is by triangle inequality. Considering $I_1$, we have
    \begingroup
    \allowdisplaybreaks
    \begin{align*}
        I_1 &\leq \left(\frac{1}{n} \sum_{i = 1}^n \left|\widetilde{X}_i(s) - \widetilde{X}_i(u)\right|^2\right)^{1/2}
        \left(\frac{1}{n} \sum_{i = 1}^n \left|\widetilde{X}_i(t)\right|^2\right)^{1/2} \\
        &\leq \left(\frac{1}{n} \sum_{i = 1}^n \left|\widetilde{X}_i(s) - \widetilde{X}_i(u)\right|^2\right)^{1/2} \left(\frac{1}{n}
        \sum_{i = 1}^n \sup_{t \in [-1, 1]} \left|\widetilde{X}_i(t)\right|^2\right)^{1/2} \\
        &= \left(\frac{1}{n} \sum_{i = 1}^n \sup_{t \in [-1, 1]} \widetilde{X}_i(t)^2\right)^{1/2}
        \left(\frac{1}{n}\sum_{i = 1}^n \left|\widetilde{X}_i(s) - \widetilde{X}_i(u)\right|^2\right)^{1/2} \\
        &\leq \left(\frac{1}{n} \sum_{i = 1}^n \sup_{t \in [-1, 1]} \widetilde{X}_i(t)^2\right)^{1/2}
        \left(\frac{1}{n}\sum_{i = 1}^n A_i^2 \phi^2(|s - u|)\right)^{1/2} \\
        &= \underbrace{\left(\frac{1}{n} \sum_{i = 1}^n \sup_{t \in [-1, 1]} \widetilde{X}_i(t)^2\right)^{1/2}
        \left(\frac{1}{n}\sum_{i = 1}^n A_i^2 \right)^{1/2}}_{:=B_n} \phi(|s - u|) \,.
    \end{align*}
    \endgroup
    The first line is by Cauchy-Schwarz inequality, and the fourth line is by \eqref{eq:ms_cont} from the proof of Theorem \ref{thm:gauss_general}. Considering $I_2$, we use the same calculation and have
    \begin{equation*}
        I_2 \le \underbrace{\left(\frac{1}{n} \sum_{i = 1}^n \sup_{s \in [-1, 1]} \widetilde{X}_i(s)^2\right)^{1/2}
        \left(\frac{1}{n}\sum_{i = 1}^n A_i^2 \right)^{1/2}}_{:=B_n} \phi(|t - v|).
    \end{equation*}
    Since $\phi$ is a non-decreasing function, we have $\phi(|s - u|) + \phi(|t - v|) \le 2\phi(|s - u| + |t - v|)$. Then, we have
    \begin{align*}
    \begin{split}
        I_1 + I_2 &\leq B_n \left(\phi(|s - u|) + \phi(|t - v|)\right) \\
        &\leq B_n 2 \phi(|s - u| + |t - v|) \\
        &\leq B_n 2 \phi\left(\sqrt{2\left(|s - u|^2 + |t - v|^2\right)}\right).
    \end{split}
    \end{align*}
    The third line is by $\ell_1-\ell_2$-norm inequality. Let $h(x) := 2\phi(\sqrt{2}x)$. Since $\phi$ is continuous, non-decreasing and $\phi(0) = 0$, we have $h(x) \downarrow 0$ as $x \downarrow 0$. Thus, we can write
    \begin{equation*}
        \left|\widehat{C}_{Z,n}(s, t) - \widehat{C}_{Z,n}(u, v)\right| \leq B_n h\left(\sqrt{|s - u|^2 + |t - v|^2}\right).
    \end{equation*}
    By Theorem 22.10 in \citetappendix{davidson_2021}, to show the strongly stochastically equicontinuity of $\widehat{C}_{Z,n}$, we also need $\limsup_n B_n < \infty$ almost surely. We show it by proofing that $B_n$ converges almost surely to a finite value. Recall that $\{A_i\}_{i=1}^n$ are i.i.d.~and have finite second moments, as shown in Theorem \ref{thm:gauss_general}. Thus by the Strong Law of Large Numbers and Continuous Mapping Theorem, we know 
    \begin{equation*}
        \left(\frac{1}{n} \sum_{i=1}^n A_i^2\right)^{1/2} \overset{a.s.}{\to}
        \mathbb{E}[A^2]^{1/2} < \infty \,.
    \end{equation*}
    Similarly, as $\{X_i\}_{i=1}^n$ are i.i.d.~and $\mathbb{E}[\sup_{t \in [-1, 1]}
    X(t)^2] < \infty$, we know
    \begin{equation*}
        \left(\frac{1}{n} \sum_{i=1}^n \sup_{t \in [-1, 1]} \widetilde{X}_i(t)^2\right)^{1/2} \overset{a.s.}{\to}
        \mathbb{E}[\sup_{t \in [-1, 1]} \widetilde{X}(t)^2]^{1/2} < \infty \,.
    \end{equation*}
    Thus, $B_n$ converges to a finite value almost surely, as desired.
    
    \textbf{Second}, with the strongly stochastically equicontinuity of empirical covariance $\widehat{C}_{Z,n}$, Theorem 22.8 in \citetappendix{davidson_2021} suggests 
    \begin{align*}
        \sup_{s,t\in[-1,1]}\left| \widehat{C}_{Z,n}(s,t) - C_{Z}(s,t) \right| &\overset{a.s.}{\to} 0 .
    \end{align*}
    This completes the proof.
\end{proof}

\noindent \centerline{ \textbf{\MakeUppercase{Proof of Theorem \ref{thm:consistency_cov}}}} 
\begin{proof}
    \textbf{First}, in line with the notation of Theorem \ref{thm:gauss_general} again, let us denote 
    $
    X_i(t):=\dot{D}_i(\dot{\varepsilon}_i(t)-\dot{\varepsilon}_i(-1))
    $ 
    and
    $
    Z_n(t) =\frac{1}{\sqrt{n}}\sum_{i=1}^n \left(X_i(t)-\mathbb{E}[X(t)] \right)= \frac{1}{\sqrt{n}}\sum_{i=1}^n \dot{D}_i(\dot{\varepsilon}_i(t)-\dot{\varepsilon}_i(-1))
    $
    for $t\in[-T_{pre},-1]\cup[0,T_{post}]$. We have shown 
    $
    Z_n \stackrel{d}{\to} \mathcal{GP}(0,C_Z),
    $
    where $C_Z=\{C_Z(s,t): s,t \in [-T_{pre},-1]\cup[0,T_{post}]\}$ and $C_Z(s,t)=\mathbb{E}[\dot{D}^2(\dot{\varepsilon}(s)-\dot{\varepsilon}(-1))(\dot{\varepsilon}(t)-\dot{\varepsilon}(-1))]$.

    \textbf{Second}, Assumptions \ref{assumption: moments c}, \ref{assumption: smooth a} and \ref{assumption: smooth b} would imply that $\mathbb{E}[\sup_{t \in [-T_{pre},-1]\cup[0,T_{post}]} X(t)^2] < \infty$ and $X\in C^2([-T_{pre},-1]\cup[0,T_{post}])$. Hence, by Theorem \ref{thm:UniformConvergenceKernel}, we can show 
    \begin{align*}
        \sup_{s,t\in[-T_{pre},-1]\cup[0,T_{post}]}\left| \widehat{C}_{Z,n}(s,t) - C_Z(s,t) \right| &\overset{a.s.}{\to} 0 \,,
    \end{align*}
    where $\widehat{C}_{Z,n}(s,t)= \frac{1}{n}\sum_{i=1}^n\dot{D}_i^2(\Delta_0 \dot{Y}_i(s))(\Delta_0 \dot{Y}_i(t))$ and $\Delta_0 \dot{Y}_i(t)=(\dot{Y}_i(t)-\dot{Y}_i(-1))-\widehat{\beta}_n(t)\dot{D}_i$.
     Despite the fact that $\dot{D}_i$ and $\dot{\varepsilon}_i(t)$ involve the sample mean, making $\{\dot{D}_i(\dot{\varepsilon}_i(t)-\dot{\varepsilon}_i(-1))\}_{i=1}^n$ not independent across $i$ and failing the assumption for Theorem \ref{thm:UniformConvergenceKernel}, the convergence result still holds, because $\dot{D}_i-\dot{D}_i^\circ=o_p(1)$ and $\sup_{t\in[-T_{pre},-1]\cup[0,T_{post}]}|\dot{\varepsilon}_i(t)-\dot{\varepsilon}_i^\circ(t)|=o_p(1)$, where $\dot{D}_i^\circ=D_i-\mathbb{E}[D]$, $\dot{\varepsilon}_i^\circ(t)=\varepsilon_i(t)-\mathbb{E}[\varepsilon(t)]$ and $\{\dot{D}_i^\circ(\dot{\varepsilon}_i^\circ(t)-\dot{\varepsilon}_i^\circ(-1))\}_{i=1}^n$ are independent across $i$.

    \textbf{Third}, in the proof of Theorem \ref{thm:pw_normality}, we have shown that
        $(n^{-1}\sum_{i=1}^n \dot{D}_i^2)^{-1} \stackrel{p}{\to}\mathbb{E}[\dot{D}^2]^{-1}.$
    By Continuous Mapping Theorem, we have 
    \begin{equation*}
        \sup_{s,t\in[-T_{pre},-1]\cup[0,T_{post}]}\left| \widehat{C}_{\beta,n}(s,t) - C_{\beta}(s,t) \right| \stackrel{a.s.}{\to} 0 ,
    \end{equation*}
    where $C_{\beta}(s,t)=\mathbb{E}[\dot{D}^2(\dot{\varepsilon}(s)-\dot{\varepsilon}(-1))(\dot{\varepsilon}(t)-\dot{\varepsilon}(-1))]\mathbb{E}[\dot{D}^2]^{-2}$ and $\widehat{C}_{\beta,n}$ is the sample analogue. This completes the proof.
\end{proof}

%%%%%%%%%%%%%%%%%%%%%%%%%%%%%%%%%%%%%%%%%%%%%%%%%%%%%%%%%%%%%%%%%%
\subsection{Proof of Theorem \ref{thm:SCBs}: Non-Coverage Probabilities of SCBs}\label{app:SCBs} 
%%%%%%%%%%%%%%%%%%%%%%%%%%%%%%%%%%%%%%%%%%%%%%%%%%%%%%%%%%%%%%%%%%

\begin{proof}

\textbf{First}, we need to prove the correct simultaneous non-coverage probability of $\widehat{\operatorname{SCB}}^{\sup}_{1-\alpha}$. We consider the following one-sided formulation at first.
\begingroup
\allowdisplaybreaks
\begin{align*}
\lim_{n\to\infty}\quad & P\left(u_{1-\alpha/2}^{\sup} < \sup_{t \in [0,T_{post}]} \frac{\widehat{\beta}_n(t)-\beta(t)}{\sqrt{\widehat{C}_{\beta,n}(t, t) /n}} \right) = \alpha/2\\
\Leftrightarrow\lim_{n\to\infty}\quad & P\left(u_{1-\alpha/2}^{\sup}<\frac{\widehat{\beta}_n(t)-\beta(t)}{\sqrt{\widehat{C}_{\beta,n}(t, t) /n}}\quad\text{for at least one}\quad t \in [0,T_{post}] \right) = \alpha/2\\
\Leftrightarrow\lim_{n\to\infty}\quad & P\left(\beta(t) < \widehat{\beta}_n(t) - u_{1-\alpha/2}^{\sup}\sqrt{\widehat{C}_{\beta,n}(t, t) /n} \quad\text{for at least one}\quad t \in [0,T_{post}] \right) = \alpha/2.
\end{align*}
\endgroup
Following parallel arguments as above, we also have
\begin{align*}
\lim_{n\to\infty}\quad & P\left(\widehat{\beta}_n(t) + u_{1-\alpha/2}^{\sup}\sqrt{\widehat{C}_{\beta,n}(t, t) /n}< \beta(t) \quad\text{for at least one}\quad t \in [0,T_{post}] \right) = \alpha/2,
\end{align*}
which implies the correct simultaneous non-coverage probability of $\widehat{\operatorname{SCB}}^{\sup}_{1-\alpha}$.

\textbf{Second}, we need to prove the correct simultaneous non-coverage probability of $\widehat{\operatorname{SCB}}^{\inf,+}_{1-\alpha}$:
\begingroup
\allowdisplaybreaks
\begin{align*}
\lim_{n\to\infty}\quad & P\left(\sup_{t \in [-T_{pre},-1]} \frac{\widehat{\beta}_n(t)-\beta(t)}{\sqrt{\widehat{C}_{\beta,n}(t, t) /n}} < u_{\alpha}^{\,\sup}\right) = \alpha\\
\Leftrightarrow\lim_{n\to\infty}\quad & P\left(\frac{\widehat{\beta}_n(t)-\beta(t)}{\sqrt{\widehat{C}_{\beta,n}(t, t) /n}}< u_{\alpha}^{\sup}\quad\text{for all}\quad t \in [-T_{pre},-1] \right) = \alpha\\
\Leftrightarrow\lim_{n\to\infty}\quad & P\left(\widehat{\beta}_n(t) - u_{\alpha}^{\sup}\sqrt{\widehat{C}_{\beta,n}(t, t) /n} < \beta(t) \quad\text{for all}\quad t \in [-T_{pre},-1] \right) = \alpha.\\
\intertext{Using that $u_{\alpha}^{\sup} = -u_{1-\alpha}^{\inf}$ yields:}
\Leftrightarrow\lim_{n\to\infty}\quad & P\left(\widehat{\beta}_n(t) + u_{1-\alpha}^{\inf}\sqrt{\widehat{C}_{\beta,n}(t, t) /n} < \beta(t) \quad\text{for all}\quad t \in [-T_{pre},-1] \right) = \alpha\\
\Leftrightarrow\lim_{n\to\infty}\quad & P\left(\beta(t)\not\in\widehat{\operatorname{SCB}}^{\inf,+}_{1-\alpha}(t)\;\text{for all}\; t \in [-T_{pre},-1] \right) = \alpha.
\end{align*}
\endgroup

\textbf{Third}, by symmetry, we also have
\begingroup
\allowdisplaybreaks
\begin{align*}
\lim_{n\to\infty}\quad & P\left( u_{1-\alpha}^{\inf} < \inf_{t \in [-T_{pre},-1]} \frac{\widehat{\beta}_n(t)-\beta(t)}{\sqrt{\widehat{C}_{\beta,n}(t, t) /n}} \right) = \alpha\\
\Leftrightarrow\lim_{n\to\infty}\quad & P\left(u_{1-\alpha}^{\inf} < \frac{\widehat{\beta}_n(t)-\beta(t)}{\sqrt{\widehat{C}_{\beta,n}(t, t) /n}} \quad\text{for all}\quad t \in [-T_{pre},-1] \right) = \alpha\\
\Leftrightarrow\lim_{n\to\infty}\quad & P\left(\beta(t) < \widehat{\beta}_n(t) - u_{1-\alpha}^{\inf}\sqrt{\widehat{C}_{\beta,n}(t, t) /n}  \quad\text{for all}\quad t \in [-T_{pre},-1] \right) = \alpha\\
\Leftrightarrow\lim_{n\to\infty}\quad & P\left(\beta(t)\not\in\widehat{\operatorname{SCB}}^{\inf,-}_{1-\alpha}(t)\;\text{for all}\; t \in [-T_{pre},-1] \right) = \alpha,
\end{align*}
\endgroup
which proves the correct simultaneous non-coverage probability of $\widehat{\operatorname{SCB}}^{\inf,-}_{1-\alpha}$. By Theorem \ref{thm:consistency_cov}, $\widehat{C}_{\beta,n}$ is uniformly consistent for $C_{\beta}$ such that we can replace $\widehat{C}_{\beta,n}$ by $C_{\beta}$ in the above arguments (Slutsky's Theorem). This completes the proof.
\end{proof}

%%%%%%%%%%%%%%%%%%%%%%%%%%%%%%%%%%%%%%%%%%%%%%%%%%%%%%%%%%%%%%%%%%
\subsection{Proof of Theorem \ref{thm:intrpl_consistency}: Uniform Consistency of Interpolation Estimator (\ref{eq:BetaHatHat})}\label{app:intrpl_error} 
%%%%%%%%%%%%%%%%%%%%%%%%%%%%%%%%%%%%%%%%%%%%%%%%%%%%%%%%%%%%%%%%%%

To show Theorem \ref{thm:intrpl_consistency}, we need Theorem \ref{thm:hormann}, which is a revised version of Lemma 1 in \citetappendix{Hormann_2022}.

\spacingset{1.2}
\begin{theorem}\label{thm:hormann}
    Define $\widehat{\beta}^*_n(t^*_j)$ as the estimates of unobservable parameters $\beta^*(t^*_j)$ at discrete time points $t^*_j \in [-1,1]$ for $j=1, \ldots, T$, and $\{\doublewidehat{\beta}^*_n(t^*), t^* \in [-1,1]\}$ as an interpolation on estimates $\widehat{\beta}^*_n(t^*_j)$. Let $\omega^{\beta^*}(\delta_T)=\sup_{s^*,t^*\in[-1,1]:~|s^*-t^*|\le\delta_T}|\beta^*(s^*)-\beta^*(t^*)|$ be the modulus of continuity of function $\beta^*$: $[-1,1] \to \mathbb{R}$, and $\delta_T=\max_{j\in\{1,\dots,T-1\}}|t^*_{j+1}-t^*_j|$, we have
    \begin{equation*}
        \sup_{t^* \in [-1,1]}|\doublewidehat{\beta}^*_n(t^*)-\beta^*(t^*)| \le \max_{j\in\{1,\dots,T\}} c_1~ |\beta^*(t^*_j)-\hat{\beta}^*_n(t^*_j)|+c_2 ~\omega^{\beta^*}(\delta_T),
    \end{equation*}
    where $c_1$ and $c_2$ are constants relative to interpolation. 
\end{theorem}
\spacingset{1.5}

\begin{proof}
    \textbf{First}, let $\{\widetilde{\beta}^*(t^*), t^* \in [-1,1]\}$ be the interpolation of the unobservable parameters $\beta^*(t^*_j)$ for $j=1,\ldots,T$. Then, by triangle inequality, we have
    \begin{equation*}
        \sup_{t^* \in [-1,1]}|\doublewidehat{\beta}^*_n(t^*)-\beta^*(t^*)| \le \underbrace{\sup_{t^* \in [-1,1]}|\doublewidehat{\beta}^*_n(t^*)-\widetilde{\beta}^*(t^*)|}_{:=E_1} +\underbrace{\sup_{t^* \in [-1,1]}|\widetilde{\beta}^*(t^*)-\beta^*(t^*)|}_{:=E_2}.
    \end{equation*}
    
    \textbf{Second}, regarding $E_1$, since the interval $[t^*_j, t^*_{j+1}]$ is bounded, the range of any interpolation on this interval should also be bounded. It means, $\exists$ a constant $c_1$, such that 
    $$
        \sup_{t^*\in[t^*_j, t^*_{j+1}]} |\doublewidehat{\beta}_n^*(t^*)-\widetilde{\beta}^*(t^*)| \le c_1  \max\left\{~|\doublewidehat{\beta}_n^*(t^*_j)-\widetilde{\beta}^*(t^*_j)|,~ |\doublewidehat{\beta}_n^*(t^*_{j+1})-\widetilde{\beta}^*(t^*_{j+1})|~\right\}
    $$
    for $j\in\{1,\dots,T-1\}$. Hence, we have
    \begin{align*}
    \begin{split}
        E_1 &=\sup_{t^*\in[-1,1]} |\doublewidehat{\beta}_n^*(t^*)-\widetilde{\beta}^*(t^*)|\\
        &\le c_1 \max_{j\in\{1,\dots,T\}}|\doublewidehat{\beta}_n^*(t^*_j)-\widetilde{\beta}^*(t^*_j)| \\
        &= c_1 \max_{j\in\{1,\dots,T\}}|\widehat{\beta}_n^*(t^*_j)-\beta^*(t^*_j)|,
    \end{split}
    \end{align*}
    where we use $\widetilde{\beta}^*(t^*_j)=\beta^*(t^*_j)$ and $\doublewidehat{\beta}_n^*(t^*_j)=\widehat{\beta}_n^*(t^*_j)$ for all $j\in\{1,\dots,T\}$.

    \textbf{Third}, regarding $E_2$, let us denote $h(t^*)=\frac{\widetilde{\beta}^*(t^*)-\beta^*(t^*_j)}{\beta^*(t^*_{j+1})-\beta^*(t^*_j)}<\infty$ for $t^* \in [t^*_j, t^*_{j+1}]$, then 
    \begingroup
    \allowdisplaybreaks
    \begin{align*}
        E_2
        =&\sup_{t^* \in [-1,1]}|\widetilde{\beta}^*(t^*)-\beta^*(t^*)|\\
        =&\max_{j\in\{1,\dots,T-1\}} \sup_{t^*\in[t^*_j,t^*_{j+1}]}|\widetilde{\beta}^*(t^*)-\beta^*(t^*)|\\
        =&\max_{j\in\{1,\dots,T-1\}} \sup_{t^*\in[t^*_j,t^*_{j+1}]}\left|\beta^*(t^*)-\left[\beta^*(t^*_j)+h(t^*)\left(\beta^*(t^*_{j+1})-\beta^*(t^*_j)\right)\right]\right|\\
        \le&\max_{j\in\{1,\dots,T-1\}} \left\{ \sup_{t^*\in[t^*_j,t^*_{j+1}]}|\beta^*(t^*)-\beta^*(t^*_j) | + \sup_{t^*\in[t^*_j,t^*_{j+1}]} \left|h(t^*)\left(\beta^*(t^*_{j+1})-\beta^*(t^*_j)\right)\right| \right\} \\
        =&\max_{j\in\{1,\dots,T-1\}} \left\{ \sup_{t^*\in[t^*_j,t^*_{j+1}]}|\beta^*(t^*)-\beta^*(t^*_j) | + \sup_{t^*\in[t^*_j,t^*_{j+1}]} \left|h(t^*)\right| \left|\left(\beta^*(t^*_{j+1})-\beta^*(t^*_j)\right)\right| \right\} \\    
        \le &\max_{j\in\{1,\dots,T-1\}} \left\{ \sup_{t^*\in[t^*_j,t^*_{j+1}]}|\beta^*(t^*)-\beta^*(t^*_j) | + \sup_{t^*\in[t^*_j,t^*_{j+1}]} \left|h(t^*)\right| \sup_{t^*\in[t^*_j,t^*_{j+1}]} \left|\beta^*(t^*)-\beta^*(t^*_j)\right| \right\} \\
        =& \left(1+ \sup_{t^*\in[-1,1]}|h(t^*)|\right) \max_{j\in\{1,\dots,T-1\}} \sup_{t^*\in[t^*_j,t^*_{j+1}]} |\beta^*(t^*)-\beta^*(t^*_j)|\\
        :=& c_2~\omega^{\beta^*}(\delta_T),
    \end{align*}
    \endgroup
    where $c_2$ is a finite constant.
   
    \textbf{Finally}, for any interpolations, we establish
    \begin{equation*}
        \sup_{t^* \in [-1,1]}|\doublewidehat{\beta}_n^*(t^*)-\beta^*(t^*)| \le c_1 \max_{j\in\{1,\dots,T\}}|\beta^*(t^*_j)-\widehat{\beta}_n^*(t^*_j)|+c_2~\omega^{\beta^*}(\delta_T).
    \end{equation*}
    This completes the proof.
\end{proof}

\noindent \centerline{ \textbf{\MakeUppercase{Proof of Theorem \ref{thm:intrpl_consistency}}}} 

\begin{proof}
\textbf{First}, by Assumption \ref{assumption: smooth d}, we have
    \begin{align*}
	\sup_{t \in [-T_{pre}, -1]\cup[0, T_{post}]} \left|\doublewidehat{\beta}_n(t)-\beta(t)\right| = \sup_{t^* \in [-1,-1/T_{pre}] \cup [0,1]} \left|\doublewidehat{\beta}_n^*(t^*)-\beta^*(t^*)\right|.
    \end{align*}

\textbf{Second}, as implied by Theorem \ref{thm:hormann}, we have
    \begin{align}\label{eq:thm_2.4_2}
        \sup_{t^* \in [-1,-1/T_{pre}] \cup [0,1]}|\doublewidehat{\beta}_n^*(t^*)-\beta^*(t^*)| \le c_1 \max_{j\in\{1,\dots,T\}}|\beta^*(t^*_j)-\widehat{\beta}_n^*(t^*_j)|+c_2 ~\omega^{\beta^*}(\delta_T)
    \end{align}

\textbf{Third}, by Assumption \ref{assumption: smooth d}, the first term of right-hand side of \eqref{eq:thm_2.4_2} is
    \begin{align*}
        c_1 \max_{j\in\{1,\dots,T\}}|\beta^*(t^*_j)-\widehat{\beta}_n^*(t^*_j)|=c_1 \max_{t\in\{-T_{pre},\cdots,T_{post}\}} \left|\beta(t)-\widehat{\beta}_n(t)\right|.
    \end{align*}
    
\textbf{Fourth}, regarding the second term of right-hand side of \eqref{eq:thm_2.4_2}, by Mean Value Theorem, $\exists$ $\xi \in (t^*,v^*)$ such that $\beta^*(t^*)-\beta^*(v^*) = \beta^{*\prime}(\xi)  (t^*-v^*)$, where $\beta^{*\prime}$ is the first derivative. By Assumption \ref{assumption: smooth d} that $\beta^*\in C^2([-1,-1/T_{pre}] \cup [0,1])$, we have $\sup_{t^* \in (-1,-1/T_{pre}) \cup (0,1)}|\beta^{*\prime}(t^*)|=K_{\beta^*}<\infty$. Hence,
    \begingroup
    \allowdisplaybreaks
    \begin{align*}
        \omega^{\beta^*}(\delta_T)
        &=\sup_{\substack{t^*,v^* \in [-1,-\frac{1}{T_{pre}}] \cup [0,1]:\\|t^*-v^*| \le \delta_T}}|\beta^*(t^*)-\beta^*(v^*)| \\
        &= |\beta^{*\prime}(\xi)| \sup_{\substack{t^*,v^* \in [-1,-\frac{1}{T_{pre}}] \cup [0,1]:\\|t^*-v^*|\le \delta_T}}|t^*-v^*|\\
        &\le  K_{\beta^*} \sup_{\substack{t^*,v^* \in [-1,-\frac{1}{T_{pre}}] \cup [0,1]:\\|t^*-v^*|\le \delta_T}} |t^*-v^*|\\
        &= K_{\beta^*} O(\frac{1}{T}),
    \end{align*}
    \endgroup
    where the last line derives from Assumption \ref{assumption: data_str b}. 
    
    \textbf{Finally}, we can write
    \begin{align*}
        \sup_{t \in [-T_{pre}, -1]\cup [0, T_{post}]} \left|\doublewidehat{\beta}_n(t)-\beta(t)\right| \le c_1 \max_{t\in\{-T_{pre},\cdots,T_{post}\}} \left|\beta(t)-\widehat{\beta}_n(t)\right|+\frac{c_2K_{\beta^*}}{T}.
    \end{align*}
    This completes the proof.
\end{proof}

%%%%%%%%%%%%%%%%%%%%%%%%%%%%%%%%%%%%%%%%%%%%%%%%%%%%%%%%%%%%%%%%%%
\subsection{Proof of Theorem \ref{thm:intrpl_dist}: Uniform Asymptotic Normality of the Interpolation Estimator (\ref{eq:BetaHatHat})}\label{app:intrpl_dist}
%%%%%%%%%%%%%%%%%%%%%%%%%%%%%%%%%%%%%%%%%%%%%%%%%%%%%%%%%%%%%%%%%%

\begin{proof}
\textbf{First}, to derive the asymptotic distribution, we write
\begin{equation}\label{eq:intrpl_dist}
    \sqrt{n}(\doublewidehat{\beta}_n-\beta)=\sqrt{n}(\doublewidehat{\beta}_n-\widehat{\beta}_n)+\sqrt{n}(\widehat{\beta}_n-\beta).
\end{equation}

\textbf{Second}, by Thoerem \ref{thm:gauss}, we already know that $\sqrt{n}(\widehat{\beta}_n-\beta) \stackrel{d}{\to} \mathcal{GP}(0,C_{\beta})$. To achieve the asymptotic Gaussian process of $\sqrt{n}(\doublewidehat{\beta}_n-\beta)$, we require the first term of right-hand side of \eqref{eq:intrpl_dist} to be uniformly negligible. 

\textbf{Third}, by replacing $\beta(t)$ in Thoerem \ref{thm:intrpl_consistency} with $\widehat{\beta}_n(t)$, we have $\sup_{t \in [-T_{pre},-1]\cup [0,T_{post}]} |\doublewidehat{\beta}_n(t)-\widehat{\beta}_n(t)| \le \frac{c_2K_{\beta^*}}{T}$. Thus, we have
\begin{equation}\label{eq:intrpl_uniform_neg}
    \sup_{t \in [-T_{pre},-1]\cup [0,T_{post}]} |\sqrt{n}(\doublewidehat{\beta}_n(t)-\widehat{\beta}_n(t))| \le \frac{c_2K_{\beta^*} \sqrt{n} }{T}.
\end{equation}
To make the left-hand side of \eqref{eq:intrpl_uniform_neg} uniformly negligible, we need that $\sqrt{n}/T$ goes to $0$ asymptotically. This completes the proof.
\end{proof}

%%%%%%%%%%%%%%%%%%%%%%%%%%%%%%%%%%%%%%%%%%%%%%%%%%%%%%%%%%%%%%%%%%
\subsection{Proof of Theorem \ref{thm:cov_intrpl_consistency}: Uniform Consistency of Interpolation Estimator (\ref{eq:CovHatHat})}\label{app:cov_intrpl_error} 
%%%%%%%%%%%%%%%%%%%%%%%%%%%%%%%%%%%%%%%%%%%%%%%%%%%%%%%%%%%%%%%%%%

To show Theorem \ref{thm:cov_intrpl_consistency}, we need Theorem \ref{thm:cov_hormann}, which is also a revised version of Lemma 1 in \citetappendix{Hormann_2022}.

\spacingset{1.2}
\begin{theorem}\label{thm:cov_hormann}
    Define $\widehat{C}^*_{\beta,n}(s^*_i, t^*_j)$ as the estimates of unobservable parameters $C^*_{\beta}(s^*_i, t^*_j)$ at discrete time points $s^*_i, t^*_j \in [-1,1]$ for $i, j=1, \ldots, T$, and $\{\doublewidehatCL{C\,}^*_{\beta,n}(s^*, t^*), s^*, t^* \in [-1,1]\}$ as an interpolation on estimates $\widehat{C}^*_{\beta,n}(s^*_i, t^*_j)$. Let 
    $$
    \omega^{C_{\beta}^*}(\delta_S,\delta_T)=\sup_{\substack{s^*,u^*\in[-1,1]:\\|s^*-u^*|\le\delta_S}}|C_{\beta}^*(s^*, \mathcal{T}^*)-C_{\beta}^*(u^*, \mathcal{T}^*)|+\sup_{\substack{t^*,v^*\in[-1,1]:\\|t^*-v^*|\le\delta_T}}|C_{\beta}^*(\mathcal{S}^*, t^*)-C_{\beta}^*(\mathcal{S}^*, v^*)|
    $$
    be the modulus of continuity of function $C_{\beta}^*$: $[-1,1]^2 \to \mathbb{R}$, where
    \begin{align*}
        \delta_S & =\max_{i\in\{1,\dots,T-1\}}|s^*_{i+1}-s^*_i|, \\
        \delta_T & =\max_{j\in\{1,\dots,T-1\}}|t^*_{j+1}-t^*_j|, \\
        (\mathcal{S}^*, \mathcal{T}^*) & = \arg \sup_{ s^*, t^* \in [-1,1]} |\widetilde{C}_{\beta}^*(s^*, t^*)-C_{\beta}^*(s^*, t^*)| 
    \end{align*}
    and $\{\widetilde{C}_{\beta}^*(s^*, t^*), s^*, t^* \in [-1,1]\}$ be the interpolation of the unobservable parameters $C_{\beta}^*(s^*_i, t^*_j)$ for $i,j\in\{1,\dots,T\}$. Then, we have
    \begin{equation*}
        \sup_{ s^*, t^* \in [-1,1]}|\doublewidehatCL{C\,}_{\beta,n}^*(s^*, t^*)-C_{\beta}^*(s^*, t^*)| \le \max_{i,j\in\{1,\dots,T\}} c_3~ |C_{\beta}^*(s^*_i, t^*_j)-\widehat{C}_{\beta, n}^*(s^*_i, t^*_j)|+c_4 ~\omega^{C_{\beta}^*}(\delta_S, \delta_T),
    \end{equation*}
    where $c_3$ and $c_4$ are constants relative to interpolation.
\end{theorem}
\spacingset{1.5}

\begin{proof}
    \textbf{First}, by triangle inequality, we have
    \begin{align*}
        \sup_{s^*, t^* \in [-1,1]}|\doublewidehatCL{C\,}_{\beta,n}^*(s^*,t^*)-C_{\beta}^*(s^*, t^*)| 
        \le& \underbrace{\sup_{s^*, t^* \in [-1,1]}|\doublewidehatCL{C\,}_{\beta,n}^*(s^*, t^*)-\widetilde{C}_{\beta}^*(s^*, t^*)|}_{:=E_1} +\\
        &\underbrace{\sup_{s^*, t^* \in [-1,1]}|\widetilde{C}_{\beta}^*(s^*, t^*)-C_{\beta}^*(s^*, t^*)|}_{:=E_2}.
    \end{align*}

    \textbf{Second}, regarding $E_1$, since the area $[s^*_i, s^*_{i+1}] \times [t^*_j, t^*_{j+1}]$ is bounded, the range of any interpolation on this area should also be bounded. Hence, $\exists$ a constant $c_3$, such that
    \begin{align*}
    \begin{split}
        E_1 &=\sup_{s^*, t^* \in [-1,1]}|\doublewidehatCL{C\,}_{\beta,n}^*(s^*, t^*)-\widetilde{C}_{\beta}^*(s^*, t^*)|\\
        &\le c_3  \max_{i,j\in\{1,\dots,T\}}|\doublewidehatCL{C\,}_{\beta,n}^*(s^*_i, t^*_j)-\widetilde{C}_{\beta}^*(s^*_i,t^*_j)| \\
        &= c_3  \max_{i,j\in\{1,\dots,T\}}|\widehat{C}_{\beta,n}^*(s^*_i, t^*_j)-C_{\beta}^*(s^*_i, t^*_j)|,
    \end{split}
    \end{align*}
    where we use $\widetilde{C}_{\beta}^*(s^*_i,t^*_j)=C_{\beta}^*(s^*_i,t^*_j)$ and $\doublewidehatCL{C\,}_{\beta,n}^*(s^*_i,t^*_j)=\widehat{C}_{\beta,n}^*(s^*_i,t^*_j)$ for all $i,j\in\{1,\dots,T\}$. 
    
    \textbf{Third}, regarding $E_2$, let us denote 
    \begin{align*}
        g(s^*) & =\frac{\widetilde{C}_{\beta}^*(s^*, \mathcal{T}^*)-C_{\beta}^*(s^*_i, \mathcal{T}^*)}{C_{\beta}^*(s^*_{i+1}, \mathcal{T}^*)-C_{\beta}^*(s^*_i, \mathcal{T}^*)} < \infty \quad \text{for} \quad s^*\in[s^*_i, s^*_{i+1}]\\
        h(t^*) & =\frac{\widetilde{C}_{\beta}^*(\mathcal{S}^*, t^*)-C_{\beta}^*(\mathcal{S}^*, t^*_j)}{C_{\beta}^*(\mathcal{S}^*, t^*_{j+1})-C_{\beta}^*(\mathcal{S}^*, t^*_j)} < \infty \quad \text{for} \quad t^*\in[t^*_j, t^*_{j+1}],
    \end{align*}
    where $(\mathcal{S}^*, \mathcal{T}^*) = \arg \sup_{ s^*, t^* \in [-1,1]} |\widetilde{C}_{\beta}^*(s^*, t^*)-C_{\beta}^*(s^*, t^*)|$. $g(s^*)$ and $h(t^*)$ are set to be zero if the denominators vanish. Then, we have
    \begingroup
    \allowdisplaybreaks
    \begin{align*}
        E_2
        = & \sup_{s^*, t^* \in [-1,1]}|\widetilde{C}_{\beta}^*(s^*, t^*)-C_{\beta}^*(s^*, t^*)|\\
        = & \max_{i,j\in\{1,\dots,T-1\}} \sup_{\substack{s^*\in[s^*_i,s^*_{i+1}] \\ t^*\in[t^*_j,t^*_{j+1}]}}|\widetilde{C}_{\beta}^*(s^*, t^*)-C_{\beta}^*(s^*, t^*)| \\
        \le & \max_{i\in\{1,\dots,T-1\}} \sup_{s^*\in[s^*_i,s^*_{i+1}]}|\widetilde{C}_{\beta}^*(s^*, \mathcal{T}^*)-C_{\beta}^*(s^*, \mathcal{T}^*)| + \max_{j\in\{1,\dots,T-1\}} \sup_{t^*\in[t^*_j, t^*_{j+1}]}|\widetilde{C}_{\beta}^*(\mathcal{S}^*, t^*)-C_{\beta}^*(\mathcal{S}^*, t^*)| \\
        = & \max_{i\in\{1,\dots,T-1\}} \sup_{s^*\in[s^*_i,s^*_{i+1}]}|C_{\beta}^*(s^*, \mathcal{T}^*)-C_{\beta}^*(s^*_i, \mathcal{T}^*)-g(s^*)( C_{\beta}^*(s^*_{i+1}, \mathcal{T}^*)-C_{\beta}^*(s^*_i, \mathcal{T}^*))| +\\
        & \max_{j\in\{1,\dots,T-1\}} \sup_{t^*\in[t^*_j,t^*_{j+1}]}|C_{\beta}^*(\mathcal{S}^*, t^*)-C_{\beta}^*(\mathcal{S}^*, t^*_j)-h(t^*)( C_{\beta}^*(\mathcal{S}^*, t^*_{j+1})-C_{\beta}^*(\mathcal{S}^*, t^*_j))| \\
        \le & \max_{i\in\{1,\dots,T-1\}} \left\{ \sup_{s^*\in[s^*_i,s^*_{i+1}]}|C_{\beta}^*(s^*, \mathcal{T}^*)-C_{\beta}^*(s^*_i, \mathcal{T}^*)| + \sup_{s^*\in[s^*_i,s^*_{i+1}]} \left|g(s^*)\right| \sup_{s^*\in[s^*_i,s^*_{i+1}]}\left|C_{\beta}^*(s^*, \mathcal{T}^*)-C_{\beta}^*(s^*_i, \mathcal{T}^*)\right| \right\}+ \\
        & \max_{j\in\{1,\dots,T-1\}} \left\{ \sup_{t^*\in[t^*_j,t^*_{j+1}]}|C_{\beta}^*(\mathcal{S}^*, t^*)-C_{\beta}^*(\mathcal{S}^*, t^*_j)| + \sup_{t^*\in[t^*_j,t^*_{j+1}]} \left|h(t^*)\right| \sup_{t^*\in[t^*_j,t^*_{j+1}]} \left|C_{\beta}^*(\mathcal{S}^*, t^*)-C_{\beta}^*(\mathcal{S}^*, t^*_j)\right| \right\}\\
        = & \left(1+ \sup_{s^*\in[-1,1]}|g(s^*)|\right) \sup_{s^*,u^*\in[-1,1]:~|s^*-u^*|\le\delta_S}|C_{\beta}^*(s^*, \mathcal{T}^*)-C_{\beta}^*(u^*, \mathcal{T}^*)|+\\
        & \left(1+ \sup_{t^*\in[-1,1]}|h(t^*)|\right) \sup_{t^*,v^*\in[-1,1]:~|t^*-v^*|\le\delta_T}|C_{\beta}^*(\mathcal{S}^*, t^*)-C_{\beta}^*(\mathcal{S}^*, v^*)|\\
        \le & \left(1+ \sup_{s^*\in[-1,1]}|g(s^*)|+ \sup_{t^*\in[-1,1]}|h(t^*)|\right) \omega^{C_{\beta}^*}(\delta_S, \delta_T)\\
        := & c_4 ~\omega^{C_{\beta}^*}(\delta_S, \delta_T),
    \end{align*}
    \endgroup
    where $c_4$ is a finite constant.
   
    \textbf{Finally}, for any interpolation, we establish
    \begin{equation*}
        \sup_{ s^*, t^* \in [-1,1]}|\doublewidehatCL{C\,}_{\beta,n}^*(s^*, t^*)-C_{\beta}^*(s^*, t^*)| \le \max_{i,j\in\{1,\dots,T\}} c_3~ |C_{\beta}^*(s^*_i, t^*_j)-\widehat{C}_{\beta, n}^*(s^*_i, t^*_j)|+c_4 ~\omega^{C_{\beta}^*}(\delta_S, \delta_T).
    \end{equation*}
    This completes the proof.
\end{proof}

\noindent \centerline{ \textbf{\MakeUppercase{Proof of Theorem \ref{thm:cov_intrpl_consistency}}}} 

\begin{proof}
\textbf{First}, by Assumption \ref{assumption: smooth e}, we have
    \begin{align*}
	\sup_{s,t \in [-T_{pre}, -1]\cup [0, T_{post}]} \left|\doublewidehatCL{C\,}_{\beta,n}(s,t)-C_{\beta}(s,t)\right| = \sup_{s^*,t^* \in [-1,-\frac{1}{T_{pre}}] \cup [0,1]} \left|\doublewidehatCL{C\,}_{\beta,n}^*(s^*,t^*)-C_{\beta}^*(s^*,t^*)\right|.
    \end{align*}

\textbf{Second}, as implied by Theorem \ref{thm:cov_hormann}, we have
    \begin{align}\label{eq:thm_2.4_2_cov}
        \sup_{s^*,t^* \in [-1,-\frac{1}{T_{pre}}] \cup [0,1]}\left|\doublewidehatCL{C\,}_{\beta,n}^*(s^*,t^*)-C_{\beta}^*(s^*,t^*)\right| \le c_3  \max_{i,j\in\{1,\dots,T\}}\left|C_{\beta}^*(s^*_i,t^*_j)-\widehat{C}_{\beta,n}^*(s^*_i,t^*_j)\right|+c_4 ~\omega^{C_{\beta}^*}(\delta_S, \delta_T)
    \end{align}

\textbf{Third}, by Assumption \ref{assumption: smooth e}, the first term of right-hand side of \eqref{eq:thm_2.4_2_cov} is
    \begin{align*}
        c_3 \max_{i,j\in\{1,\dots,T\}}\left|C_{\beta}^*(s^*_i,t^*_j)-\widehat{C}_{\beta,n}^*(s^*_i,t^*_j)\right|=c_3 \max_{s,t\in\{-T_{pre},\dots,T_{post}\}} \left|C_{\beta}(s,t)-\widehat{C}_{\beta,n}(s,t)\right|.
    \end{align*}
    
\textbf{Fourth}, regarding the second term of right-hand side of \eqref{eq:thm_2.4_2_cov}, by Mean Value Theorem, there exists $\zeta \in (s^*, u^*)$ such that $C_{\beta}^*(s^*, \mathcal{T}^*)-C_{\beta}^*(u^*, \mathcal{T}^*)=\frac{\partial C_{\beta}^*(\zeta, \mathcal{T}^*)}{\partial s^*} (s^*-u^*)$, and also there exists $\xi \in (t^*, v^*)$ such that $C_{\beta}^*(\mathcal{S}^*, t^*)-C_{\beta}^*(\mathcal{S}^*, v^*)=\frac{\partial C_{\beta}^*(\xi, \mathcal{S}^*)}{\partial t^*} (t^*-v^*)$, where $(\mathcal{S}^*, \mathcal{T}^*)$ is as defined in Theorem \ref{thm:cov_hormann}. By Assumption \ref{assumption: smooth e} that $C_{\beta}^*\in C^2([-1,-\frac{1}{T_{pre}}] \cup [0,1])^2$, we have $\sup_{s^* \in (-1,-\frac{1}{T_{pre}}) \cup (0,1)}|\frac{\partial C_{\beta}^*(\zeta, \mathcal{T}^*)}{\partial s^*}| + \sup_{t^* \in (-1,-\frac{1}{T_{pre}}) \cup (0,1)}| \frac{\partial C_{\beta}^*(\mathcal{S}^*, \xi)}{\partial t^*}|=K_{C_{\beta}^*}<\infty$. Hence,
    \begingroup
    \allowdisplaybreaks
    \begin{align*}
        \omega^{C_{\beta}^*}(\delta_S,\delta_T)
        =&\sup_{\substack{s^*,u^*\in[-1,-\frac{1}{T_{pre}}] \cup [0,1]:\\|s^*-u^*|\le\delta_S}}\left|C_{\beta}^*(s^*, \mathcal{T}^*)-C_{\beta}^*(u^*, \mathcal{T}^*)\right|+\sup_{\substack{t^*,v^*\in[-1,-\frac{1}{T_{pre}}] \cup [0,1]:\\|t^*-v^*|\le\delta_T}}\left|C_{\beta}^*(\mathcal{S}^*, t^*)-C_{\beta}^*(\mathcal{S}^*, v^*)\right| \\
        \le& \sup_{s^*\in(-1,-\frac{1}{T_{pre}}) \cup (0,1)}\left|\frac{\partial C_{\beta}^*(s^*, \mathcal{T}^*)}{\partial s^*}\right|  \sup_{\substack{s^*,u^*\in[-1,-\frac{1}{T_{pre}}] \cup [0,1]:\\|s^*-u^*|\le\delta_S}}|s^*-u^*|\\
        &+\sup_{t^*\in(-1,-\frac{1}{T_{pre}}) \cup (0,1)}\left|\frac{\partial C_{\beta}^*(\mathcal{S}^*, t^*)}{\partial t^*}\right|  \sup_{\substack{t^*,v^*\in[-1,-\frac{1}{T_{pre}}] \cup [0,1]:\\|t^*-v^*|\le\delta_T}}| t^*-v^*| \\
        =& K_{C_{\beta}^*} O(\frac{1}{T}),
    \end{align*}
    \endgroup
    where the last line derives from Assumption \ref{assumption: data_str b}. 
    
    \textbf{Finally}, we can write 
    \begin{align*}
    \begin{split}
        \sup_{s,t \in [-T_{pre}, -1]\cup [0, T_{post}]} \left|\doublewidehatCL{C\,}_{\beta,n}(s,t)-C_{\beta}(s,t)\right| \le c_3 \max_{s,t\in\{-T_{pre},\dots,T_{post}\}} \left|C_{\beta}(s,t)-\widehat{C}_{\beta,n}(s,t)\right|+\frac{c_4K_{C_{\beta}^*}}{T}.
    \end{split}
    \end{align*}
    This completes the proof.
\end{proof}

%%%%%%%%%%%%%%%%%%%%%%%%%%%%%%%%%%%%%%%%%%%%%%%%%%%%%%%%%%%%%%%%%%
\section{Proofs of Corollaries}\label{sec:PROOFS_Corollaries} 
%%%%%%%%%%%%%%%%%%%%%%%%%%%%%%%%%%%%%%%%%%%%%%%%%%%%%%%%%%%%%%%%%%

%%%%%%%%%%%%%%%%%%%%%%%%%%%%%%%%%%%%%%%%%%%%%%%%%%%%%%%%%%%%%%%%%%
\subsection{Proof of Corollary \ref{cor:SCBs_Interpolation}: Non-Coverage Probabilities of Interpolation SCBs}\label{app:cor_SCBs_Interpolation} 
%%%%%%%%%%%%%%%%%%%%%%%%%%%%%%%%%%%%%%%%%%%%%%%%%%%%%%%%%%%%%%%%%%

\begin{proof}
The proof of Corollary \ref{cor:SCBs_Interpolation} is identical to that of Theorem \ref{thm:SCBs}, but we use interpolation-based SCBs instead. The uniform consistency of interpolation estimator $\doublewidehatCL{C}_{\beta,n}$ to $C_{\beta}$ is guaranteed by Theorem \ref{thm:cov_intrpl_consistency}. This completes the proof.
\end{proof}

%%%%%%%%%%%%%%%%%%%%%%%%%%%%%%%%%%%%%%%%%%%%%%%%%%%%%%%%%%%%%%%%%%
\subsection{Proof of Corollary \ref{cor:SCBs_relevance}: Size Control in Relevance Testing}\label{app:SCBs_relevance} 
%%%%%%%%%%%%%%%%%%%%%%%%%%%%%%%%%%%%%%%%%%%%%%%%%%%%%%%%%%%%%%%%%%

\begin{proof}
\textbf{First}, if $\Delta_\ell(t)=\Delta_u(t)=\Delta(t)$, i.e., the reference band reduces to a reference line at true bias value $\Delta(t)=\beta(t)$, by Corollary \ref{cor:SCBs_Interpolation}, we have 
\begin{align*}
\lim_{n\to\infty}\quad & P\left( \Delta(t) \not\in \doublewidehatCL{\operatorname{SCB}}{}^{\sup}_{1-\alpha}(t)\;\text{for at least one}\; t \in [0,T_{post}] \right) = \alpha.
\end{align*}

\textbf{Second}, if $\Delta_\ell(t) \neq \Delta_u(t)$, i.e., the reference band contains the true bias value $\Delta(t)=\beta(t)$ in-between, we have
\begin{align*}
\lim_{n\to\infty}\quad & \sup_{H_0} P\left([\Delta_\ell(t), \Delta_u(t)]\cap \doublewidehatCL{\operatorname{SCB}}{}^{\sup}_{1-\alpha}(t)=\emptyset\;\text{for at least one}\; t \in [0,T_{post}] \right) \le \alpha.
\end{align*}
These two scenarios are both under the null hypothesis $H_0$: $\beta(t)=\Delta(t)\in[\Delta_{\ell}(t), \Delta_u(t)]$ for all $t\in[0,T_{post}]$. This completes the proof.
\end{proof}

%%%%%%%%%%%%%%%%%%%%%%%%%%%%%%%%%%%%%%%%%%%%%%%%%%%%%%%%%%%%%%%%%%
\subsection{Proof of Corollary \ref{cor:SCBs_relevance_practical}: Size Control in Relevance Testing with \eqref{eq:ta_bounds} and \eqref{eq:rmtrb_bounds}}\label{app:SCBs_relevance_practical} 
%%%%%%%%%%%%%%%%%%%%%%%%%%%%%%%%%%%%%%%%%%%%%%%%%%%%%%%%%%%%%%%%%%

\begin{proof}
    We show the size control with \eqref{eq:ta_bounds} and \eqref{eq:rmtrb_bounds} in scenario of $\Delta_\ell(t)=\Delta_u(t):=\Delta(t)$, i.e., the reference band reduces to a reference line. The size control in scenario of $\Delta_\ell(t)\neq \Delta_u(t)$ naturally follows. 
    
    \textbf{First}, in the scenario of $\Delta_\ell(t)=\Delta_u(t):=\Delta(t)$, we basically test the following null hypothesis
    $$
    H_0: \beta(t)=\Delta(t), \quad \forall t\in[0, T_{post}]
    $$
    using the supremum-based $\doublewidehatCL{\operatorname{SCB}}{}^{\sup}_{1-\alpha}(t)$ for $\beta(t)$. This is equivalent to testing the following null hypothesis 
    $$
    H_0: \delta(t):=\beta(t)-\Delta(t)=0, \quad \forall t\in[0, T_{post}]
    $$
    using the supremum-based $\doublewidehatCL{\operatorname{SCB}}{}^{\sup, \delta}_{1-\alpha}(t)$ for $\delta(t)$, constructed in the same way as $\doublewidehatCL{\operatorname{SCB}}{}^{\sup}_{1-\alpha}(t)$.

    \textbf{Second}, the estimator $\widehat{\beta}_n(t)$ in \eqref{eq:BetaHat} for $\beta(t)$ can also be seen as the solution to the following first-order condition:
    \begin{equation}\label{eq:FOC_BetaHat}
    \frac{1}{n}\sum_{i=1}^n \dot{D}_i (\dot{Y}_i(t)-\dot{Y}_i(-1)-\beta(t)\dot{D}_i)=0.
    \end{equation}
    In the condition \eqref{eq:FOC_BetaHat}, we replace $\beta(t)$ with $\delta(t)+\widehat\Delta(t)$, where $\widehat\Delta(t):=\widehat\Delta_\ell(t)=\widehat\Delta_u(t)$ is the estimated $\Delta(t)$ in practice, and we obtain
    \begin{equation}\label{eq:FOC_delta}
    \frac{1}{n}\sum_{i=1}^n \dot{D}_i (\dot{Y}_i(t)-\dot{Y}_i(-1)-(\delta(t)+\widehat\Delta(t))\dot{D}_i)=0.
    \end{equation}
    
    \textbf{Third}, we further rewrite the condition \eqref{eq:FOC_delta} as
    \begin{equation}\label{eq:FOC_delta2}
    \frac{1}{n}\sum_{i=1}^n \dot{D}_i (\dot{Y}_i(t)-\dot{Y}_i(-1)-\beta(t)\dot{D}_i) + \frac{1}{n}\sum_{i=1}^n \dot{D}_i (\dot{Y}_i(t)-\dot{Y}_i(-1)-(\widehat\Delta(t)-\Delta(t))\dot{D}_i)=0.
    \end{equation}
    In the condition \eqref{eq:FOC_delta2}, we can see that the first part is exactly from the condition \eqref{eq:FOC_BetaHat}, and $\widehat\Delta(t)$ impacts the condition completely via the second part. From Theorem \ref{thm:gauss}, we know that $\widehat\beta_n(t)$ converges at a rate of $1/\sqrt{n}$. Hence, as long as $\widehat\Delta(t)$ converges at a higher rate, the conditions \eqref{eq:FOC_delta2} and \eqref{eq:FOC_BetaHat} are asymptotically equivalent.

    \textbf{Fourth}, we have
    $$
    \widehat\Delta(t)= \frac{1}{T_A}\sum_{s=-T_{pre}}^{t_A}\widehat{\beta}_{n}(s) 
    $$
    in case of \eqref{eq:ta_bounds}, and
    $$
    \widehat\Delta(t)=  \left(\frac{1}{T_{pre}-1} \sum_{s=-T_{pre}}^{-2} \frac{\widehat{\beta}_n(s)}{s+1} \right) (t+1)
    $$
    in case of \eqref{eq:rmtrb_bounds}. Both of them converge at a rate of $1/\sqrt{nT_{pre}}$. Assumption \ref{assumption: data_str b} implies $T_{pre} \to \infty$, and hence $\widehat\Delta(t)$ has a higher convergence rate than $\widehat\beta_n(t)$. The use of $\widehat\Delta(t)$ in practice in place of $\Delta(t)$ does not affect the size control. Therefore, we have
    \begin{align*}
        \lim_{n\to\infty}\quad & \sup_{H_0} P\left([\widehat\Delta_\ell(t), \widehat\Delta_u(t)]\cap \doublewidehatCL{\operatorname{SCB}}{}^{\sup}_{1-\alpha}(t)=\emptyset\;\text{for at least one}\; t \in [0,T_{post}] \right) \leq \alpha,
    \end{align*}
    where $[\widehat\Delta_\ell(t), \widehat\Delta_u(t)]$ is derived from \eqref{eq:ta_bounds} or \eqref{eq:rmtrb_bounds}. This completes the proof.
\end{proof}

%%%%%%%%%%%%%%%%%%%%%%%%%%%%%%%%%%%%%%%%%%%%%%%%%%%%%%%%%%%%%%%%%%
\subsection{Proof of Corollary \ref{cor:SCBs_equivalence}: Size Control in Equivalence Testing}\label{app:SCBs_equivalence} 
%%%%%%%%%%%%%%%%%%%%%%%%%%%%%%%%%%%%%%%%%%%%%%%%%%%%%%%%%%%%%%%%%%

\begin{proof}
\textbf{First}, the two one-sided infimum-based $(1-\alpha)\times 100\%$ SCBs, $\doublewidehatCL{\operatorname{SCB}}{}^{\inf,-}_{1-\alpha}(t)$ and $\doublewidehatCL{\operatorname{SCB}}{}^{\inf,+}_{1-\alpha}(t)$, can be used to test the two one-sided null hypothesis $H_0^{-}$ and $H_0^{+}$ (Section \ref{sec:validation}). Corollary \ref{cor:SCBs_Interpolation} implies that $\doublewidehatCL{\operatorname{SCB}}{}^{\inf,-}_{1-\alpha}(t)$ and $\doublewidehatCL{\operatorname{SCB}}{}^{\inf,+}_{1-\alpha}(t)$ have the correct size:

\begin{align*}
\lim_{n\to\infty}\quad & \sup_{H_0}P\left(\beta(t)\not\in\doublewidehatCL{\operatorname{SCB}}{}^{\inf,\square}_{1-\alpha}(t)\;\text{for all}\; t \in [-T_{pre},t_A] \right) \le \alpha\quad\text{for each}\quad\square\in\{+,-\}.
\end{align*}

\textbf{Second}, for our equivalence testing, the rejection region is actually the intersection of the rejection regions for the two one-sided tests. By the Intersection-Union Test (IUT) principle \citepappendix[see][Theorem 8.3.23]{Casella_Berger_2024_app}, we can show the correct level for our equivalence testing:

\spacingset{1}
\begin{align*}
    &\lim_{n\to\infty} \sup_{H_0} P\Big(\doublewidehatCL{\operatorname{SCB}}{}^{\inf}_{1-2\alpha}(t)\subsetneqq \left[\Delta_\ell(t), \Delta_u(t)\right]\quad\text{for all}\quad t\in[-T_{pre}, t_A]\Big)\\
    \le & \lim_{n\to\infty}\sup_{H_0}P\left(\beta(t)\not\in\doublewidehatCL{\operatorname{SCB}}{}^{\inf,\square}_{1-\alpha}(t)\;\text{for all}\; t \in [-T_{pre},t_A] \right)  \quad \text{for any} \quad \square\in\{+,-\} \\
    \le &\alpha.
\end{align*}
This completes the proof.
\end{proof}
\spacingset{1.5}

\newpage
%%%%%%%%%%%%%%%%%%%%%%%%%%%%%%%%%%%%%%%%%%%%%%%%%%%%%%%%%%%%%%%%%%
\section{Algorithms}\label{sec:ALGORITHMS} 
%%%%%%%%%%%%%%%%%%%%%%%%%%%%%%%%%%%%%%%%%%%%%%%%%%%%%%%%%%%%%%%%%%

\subsection{Parametric Bootstrap (Sup)}\label{ssec:SCB_PB_SUP}

\begin{algorithm}
\caption{Parametric Bootstrap for Supremum-based Simultaneous Confidence Band}
\begin{algorithmic}[1] \label{algo:SCB_PB_SUP}  
    \STATE Estimate $\widehat{\beta}_n(t)$ and $\widehat{C}_{\beta,n}(s,t)$ in \eqref{eq:BetaHat} and \eqref{eq:cov_est} for each observable time points $s,t\in\{0,\dots, T_{post}\}$ using the data.
    \STATE Interpolate $\widehat{\beta}_n(t)$ and $\widehat{C}_{\beta,n}(s,t)$ by \eqref{eq:BetaHatHat} and \eqref{eq:CovHatHat} to obtain natural cubic spline interpolations $\doublewidehat{\beta}_n(t)$ and $\doublewidehatCL{C}_{\beta,n}(s,t)$ over all $s,t \in[0, T_{post}]$.
    \FOR{$b = 1$ to $B$}
        \STATE Draw a random realization of event study coefficient estimate, $\widehat{\beta}^{*(b)}(t)$, from the multivariate normal distribution $\mathcal{GP}\left(\widehat{\beta}_n(t), \widehat{C}_{\beta,n}(s,t)/n\right)$ for $s, t\in\{0,\dots, T_{post}\}$.
        \STATE Interpolate $\widehat{\beta}^{*(b)}(t)$ by \eqref{eq:BetaHatHat} to obtain natural cubic spline interpolation $\doublewidehat{\beta}^{*(b)}(t)$ for all $t\in[0, T_{post}]$.
        \STATE Compute bootstrap statistic $T^{*(b)}= \sup_{t\in[0, T_{post}]} \left[\sqrt{n}\Big(\doublewidehat{\beta}^{*(b)}(t)-\doublewidehat{\beta}_n(t)\Big)\big/ \left(\doublewidehatCL{C}_{\beta,n}(t,t)\right)^{1/2} \right]$.
    \ENDFOR
    \STATE Compute the empirical $(1-\alpha/2)\times100\%$ quantile of $\left\{T^{*(b)}\right\}_{b=1}^B$, denoted as $\doublewidehatSL{u}^{\sup}_{1-\alpha/2}$.
    \STATE Construct the supremum-based $(1-\alpha)\times100\%$ simultaneous confidence band as $\doublewidehatCL{\operatorname{SCB}}{}^{\sup}_{1-\alpha}(t)=\left[\doublewidehat{\beta}_n(t) \pm \doublewidehatSL{u}^{\sup}_{1-\alpha/2} \sqrt{\doublewidehatCL{C}_{\beta,n}(t,t)/n}\right]$ for $t\in[0, T_{post}]$.
\end{algorithmic}
\end{algorithm}

\subsection{Parametric Bootstrap (Inf)}\label{ssec:SCB_PB_INF}

The algorithm of parametric bootstrap for infimum-based $(1-2\alpha)\times 100\%$ simultaneous confidence band $\doublewidehatCL{\operatorname{SCB}}{}^{\inf}_{1-2\alpha}(t)$ is identical to Algorithm \ref{algo:SCB_PB_SUP}, except that the bootstrap statistic in Line~6 replaces the supremum operator with infimum; the time span of interest changes from $s,t\in[0, T_{post}]$ to $s,t\in[-T_{pre},-1]$; and the significance level shifts from $\alpha$ to $2\alpha$.

\newpage
\subsection{Multiplier Bootstrap (Sup)}\label{ssec:SCB_MB_SUP}

\begin{algorithm}
\caption{Multiplier Bootstrap for Supremum-based Simultaneous Confidence Band} 
\begin{algorithmic}[1]\label{algo:SCB_MB_SUP}
    \STATE Estimate $\widehat{\beta}_n(t)$ and $\widehat{C}_{\beta,n}(s,t)$ in \eqref{eq:BetaHat} and \eqref{eq:cov_est} for each observable time points $s,t\in\{0,\dots, T_{post}\}$ using the data.
    \STATE Interpolate $\widehat{\beta}_n(t)$ and $\widehat{C}_{\beta,n}(s,t)$ by \eqref{eq:BetaHatHat} and \eqref{eq:CovHatHat} to obtain natural cubic spline interpolations $\doublewidehat{\beta}_n(t)$ and $\doublewidehatCL{C}_{\beta,n}(s,t)$ over all $s,t \in[0, T_{post}]$.
        \STATE Calculate the residuals $\Delta_0\dot{Y}_i(t)=(\dot{Y}_i(t)-\dot{Y}_i(0))-\widehat{\beta}_n(t)\dot{D}_i$ for each $i=1,\dots,n$ and $t\in\{0,\dots, T_{post}\}$.
    \FOR{$b = 1$ to $B$}
        \STATE Draw $n$ random realizations of variable $\Delta_0^*\dot{Y}_i(t)$ from binary distribution:
        $$
        P\left( \Delta_0^*\dot{Y}_i(t) = \frac{1-\sqrt{5}}{2} \Delta_0\dot{Y}_i(t) \right)=\Pi
        \text{ and }
        P\left( \Delta_0^*\dot{Y}_i(t) = \frac{1+\sqrt{5}}{2} \Delta_0\dot{Y}_i(t) \right)=1-\Pi,
        $$
        for $i=1,\dots,n$, where $\Pi=\frac{5+\sqrt{5}}{10}$. 
        \STATE Duplicate the $n$ realizations fixed for all $t\in\{0,\dots, T_{post}\}$ to retain the temporal correlation for generating the panel structure.
        \STATE Calculate $(\dot{Y}^*_i(t)-\dot{Y}^*_i(0))=\widehat{\beta}_n(t) \dot{D}_i+\Delta_0^*\dot{Y}_i(t)$ for $i=1,\dots,n$ and $t\in\{0,\dots, T_{post}\}$.
        \STATE Compute $\widehat{\beta}^{(b)}_n(t)= \left(\frac{1}{n}\sum_{i=1}^n \dot{D}_i^2\right)^{-1}\left(\frac{1}{n}\sum_{i=1}^n \dot{D}_i (\dot{Y}^*_i(t) - \dot{Y}^*_i(0))\right)$ for $i=1,\dots,n$ and $t\in\{0,\dots, T_{post}\}$.
         \STATE Interpolate $\widehat{\beta}^{*(b)}(t)$ by \eqref{eq:BetaHatHat} to obtain natural cubic spline interpolation $\doublewidehat{\beta}^{*(b)}(t)$ for all $t\in[0, T_{post}]$.
        \STATE Compute bootstrap statistic $T^{*(b)}= \sup_{t\in[0, T_{post}]} \left[\sqrt{n}\Big(\doublewidehat{\beta}^{*(b)}(t)-\doublewidehat{\beta}_n(t)\Big)\big/ \left(\doublewidehatCL{C}_{\beta,n}(t,t)\right)^{1/2} \right]$.
    \ENDFOR
    \STATE Compute the empirical $(1-\alpha/2)\times100\%$ quantile of $\left\{T^{*(b)}\right\}_{b=1}^B$, denoted as $\doublewidehatSL{u}^{\sup}_{1-\alpha/2}$.
    \STATE Construct the supremum-based $(1-\alpha)\times100\%$ simultaneous confidence band as $\doublewidehatCL{\operatorname{SCB}}{}^{\sup}_{1-\alpha}(t)=\left[\doublewidehat{\beta}_n(t) \pm \doublewidehatSL{u}^{\sup}_{1-\alpha/2} \sqrt{\doublewidehatCL{C}_{\beta,n}(t,t)/n}\right]$ for $t\in[0, T_{post}]$.
\end{algorithmic}
\end{algorithm}

\subsection{Multiplier Bootstrap (Inf)}\label{ssec:SCB_MB_INF}

The algorithm of multiplier bootstrap for infimum-based $(1-2\alpha)\times 100\%$ simultaneous confidence band $\doublewidehatCL{\operatorname{SCB}}{}^{\inf}_{1-2\alpha}(t)$ is identical to Algorithm \ref{algo:SCB_MB_SUP}, except that the bootstrap statistic in Line~10 replaces the supremum operator with infimum; the time span of interest changes from $s,t\in[0, T_{post}]$ to $s,t\in[-T_{pre},-1]$; and the significance level shifts from $\alpha$ to $2\alpha$.

\subsection{Kac-Rice Formula (Sup)}\label{ssec:SCB_KR}

\begin{algorithm}
\caption{Kac-Rice Formula for Supremum-based Simultaneous Confidence Band} 
\begin{algorithmic}[1]\label{algo:SCB_KR_SUP}
    \STATE Estimate $\widehat{\beta}_n(t)$ and $\widehat{C}_{\beta,n}(s,t)$ in \eqref{eq:BetaHat} and \eqref{eq:cov_est} for each observable time points $s,t\in\{0,\dots, T_{post}\}$ using the data.
    \STATE Interpolate $\widehat{\beta}_n(t)$ and $\widehat{C}_{\beta,n}(s,t)$ by \eqref{eq:BetaHatHat} and \eqref{eq:CovHatHat} to obtain natural cubic spline interpolations $\doublewidehat{\beta}_n(t)$ and $\doublewidehatCL{C}_{\beta,n}(s,t)$ over all $s,t \in[0, T_{post}]$.
    \STATE Estimate the roughness parameter $\tau(t)$ of the empirical correlation function along its diagonal as
        $$
            \doublewidehatSL{\tau}_{n}(t) = \left( \left.\frac{\partial^2}{\partial s \partial t } \doublewidehatCL{\operatorname{Corr}}_{\beta,n}(s,t) \right|_{(s,t)=(t,t)} \right)^{1/2}\quad\text{with}\quad \doublewidehatCL{\operatorname{Corr}}_{\beta,n}(s, t) = \frac{\doublewidehatCL{C}_{\beta,n}(s, t)}{\sqrt{\doublewidehatCL{C}_{\beta,n}(s, s)\doublewidehatCL{C}_{\beta,n}(t, t)}},
        $$
        for $s,t\in[0, T_{post}]$.
    \STATE Apply Corollary 3.3 (b) in \citeappendix{Liebl_Reimherr_2023_app} and determine $\doublewidehatSL{u}^{\sup}_{1-\alpha/2}>0$ as the solution to the equation
    $$
        F(-\doublewidehatSL{u}^{\sup}_{1-\alpha/2};\text{df})+\frac{1}{2\pi} \int_{t\in[0,T_{post}]}|\doublewidehatSL{\tau}_{n}(t)|dt\left(1+\frac{\doublewidehatSL{u}^{\sup}_{1-\alpha/2}}{\text{df}}\right)^{\text{df}/2} = \frac{\alpha}{2}
    $$
    with $F(\cdot;\text{df})$ denoting the cdf of the $t$-distribution with $\text{df}=n-1$ degrees of freedom.
    \STATE Construct the supremum-based $(1-\alpha)\times100\%$ simultaneous confidence band as $\doublewidehatCL{\operatorname{SCB}}{}^{\sup}_{1-\alpha}(t)=\left[\doublewidehat{\beta}_n(t) \pm \doublewidehatSL{u}^{\sup}_{1-\alpha/2} \sqrt{\doublewidehatCL{C}_{\beta,n}(t,t)/n}\right]$ for $t\in[0, T_{post}]$.
\end{algorithmic}
\end{algorithm}

\vspace*{-1em}
%%%%%%%%%%%%%%%%%%%%%%%%%%%%%%%%%%%%%%%%%%%%%%%%%%%%%%%%%%%%%%%%%%
\section{Additional Simulation Results}\label{sec:SIM_MORE} 
%%%%%%%%%%%%%%%%%%%%%%%%%%%%%%%%%%%%%%%%%%%%%%%%%%%%%%%%%%%%%%%%%%

\subsection{Classical Causal Inference in the Post-Treatment Period}\label{SIM:classic_hypothesis_test}

Under Assumptions I (No Anticipation) and Assumption II (Parallel Trends), we test
\begin{equation*}
    H_0\colon \beta(t)=0, \quad  \forall \; t\in[0,T_{post}] \quad \text{vs.} \quad H_1\colon\exists \;t\in[0,T_{post}] \quad \text{s.t.} \quad \beta(t) \not=0.
\end{equation*}
Data under $H_0$ are generated by setting $a=0$ in ATT1 and ATT2, while $|a|>0$ generate alternatives for producing power curves. For each $(n,T)$-combination, we run 500 simulations under the specified data generation process. In each run, we check whether $\doublewidehatCL{\operatorname{SCB}}^{\sup}_{1-\alpha}(t)$, with $\alpha=0.05$, covers zero at a grid of 101 equidistant time points over $t\in[0,T_{post}]$. We reject $H_0$ if zero is excluded at least once. We compare the three supremum-based $\doublewidehatCL{\operatorname{SCB}}^{\sup}_{1-\alpha}(t)$ with Naive band $\doublewidehatCL{\operatorname{CI}}{}^{\!\text{Naive}}_{1-\alpha}(t)$ and Bonferroni band $\doublewidehatCL{\operatorname{CI}}{}^{\!\text{Bonf}}_{1-\alpha}(t)$ with $\alpha/101$ as Bonferroni correction.

Figure \ref{fig:N200ATT1PowerCurves_SI} shows power curves for ATT1 with $n=200$ and $T=21$. The Naive band is anti-conservative and invalid, while the Bonferroni band is overly conservative and uniformly less powerful than the three $\doublewidehatCL{\operatorname{SCB}}^{\sup}_{1-\alpha}(t)$ bands. Results for ATT2 in Figure \ref{fig:N200ATT2PowerCurves_SI} are qualitatively equivalent.

\begin{figure}[ht!]
	\centering
	\includegraphics[width=0.98\linewidth]{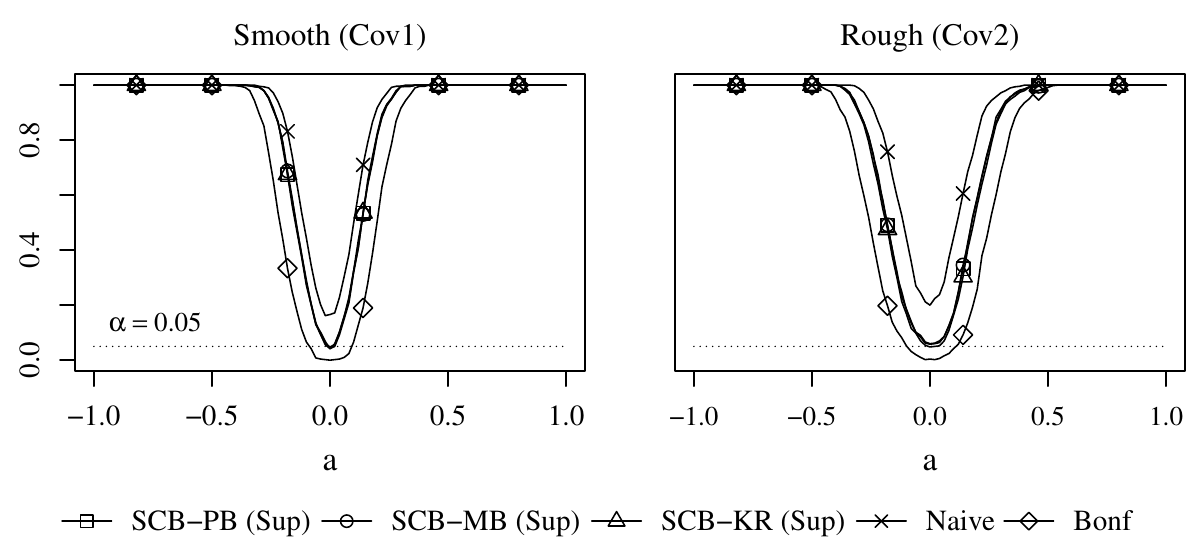}
	\caption[]{Power curves, under Assumptions I and II, for ATT1, $n=200$, and $T=21$.}
	\label{fig:N200ATT1PowerCurves_SI}
\end{figure}

\begin{figure}[ht!]
	\centering
	\includegraphics[trim={0 0 0 0}, clip, width=0.98\linewidth]{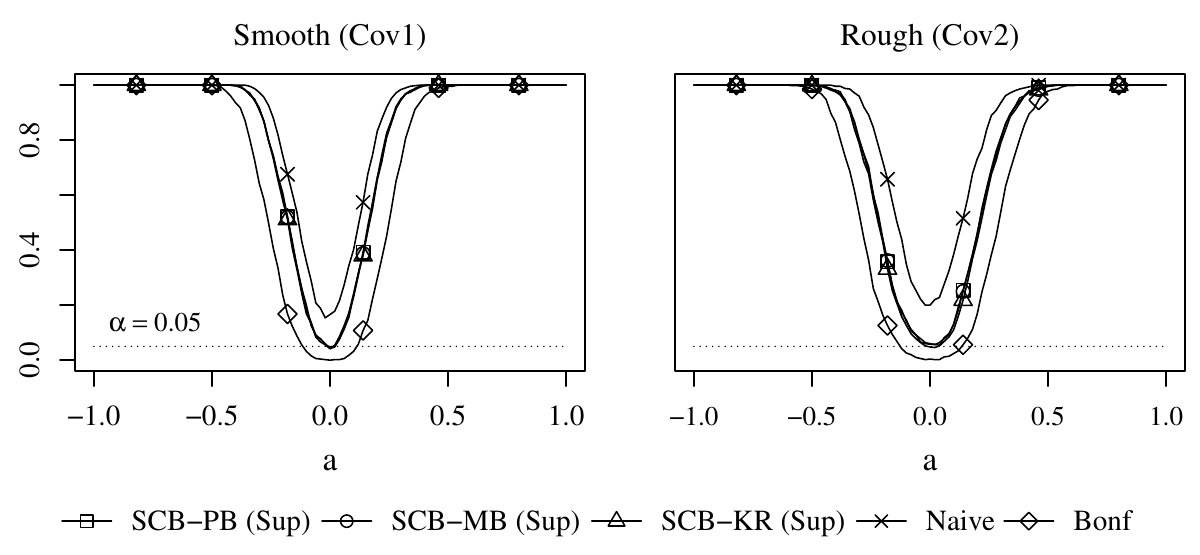}
	\caption[]{Power curves, under Assumption I and II, for ATT2, $n=200$, and $T=21$.}
	\label{fig:N200ATT2PowerCurves_SI}
\end{figure}

\newpage

\subsection{Honest Hypothesis Testing in the Post-Treatment Period}\label{SIM:honest_hypothesis_test}

\begin{figure}[ht!]
	\centering
	\includegraphics[trim={0 0 0 0}, clip, width=0.98\linewidth]{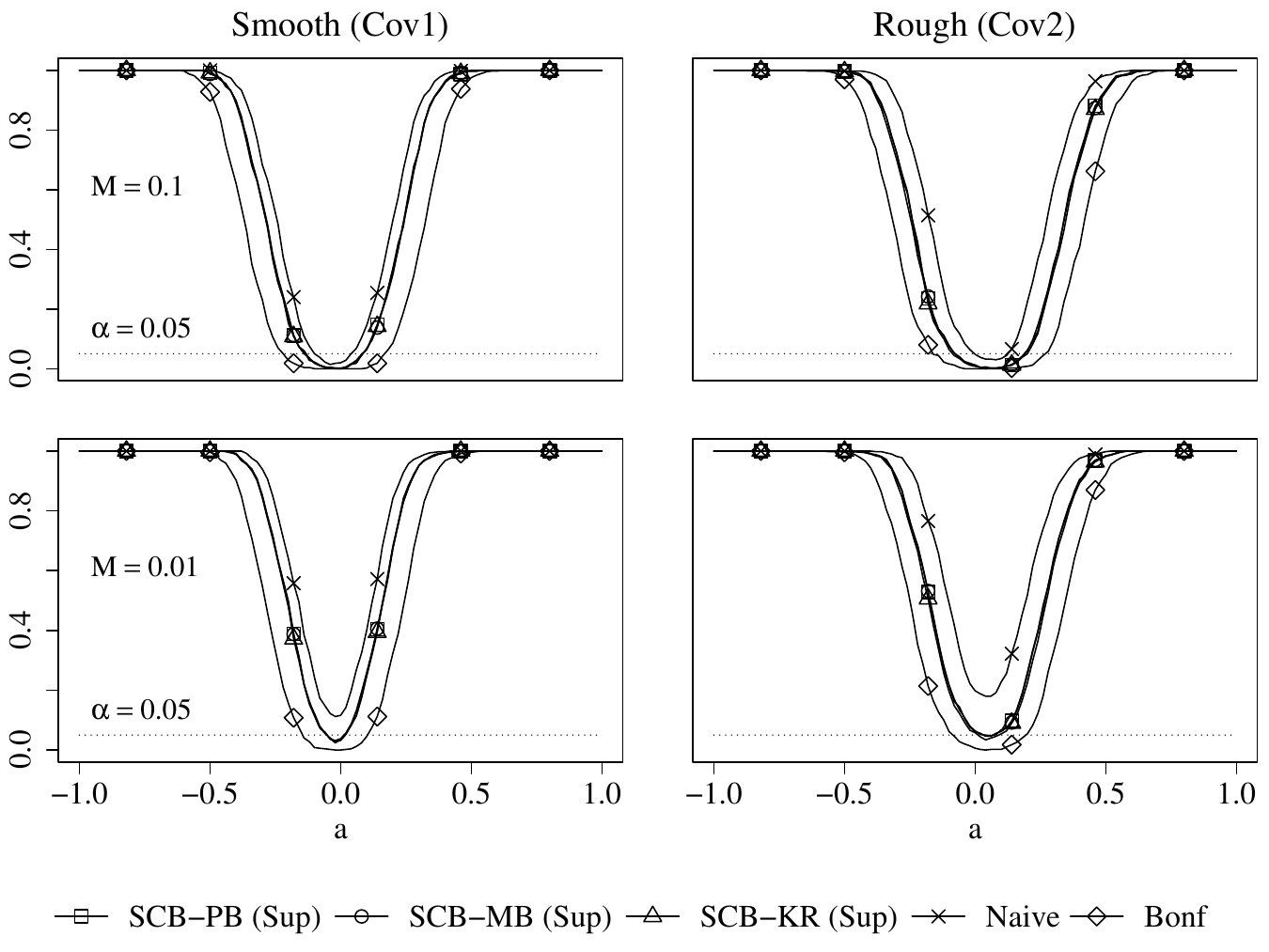}
	\caption[]{Power curves, under violated Assumption II, for ATT2, $n=200$, and $T=21$.}
	\label{fig:N200ATT2PowerCurves_PT}
\end{figure}

\newpage

\subsection{Validating Reference Bands in the Pre-Anticipation Period}\label{SIM:reference_band_validating}
\begin{figure}[ht!]
	\centering
	\includegraphics[width=0.98\linewidth]{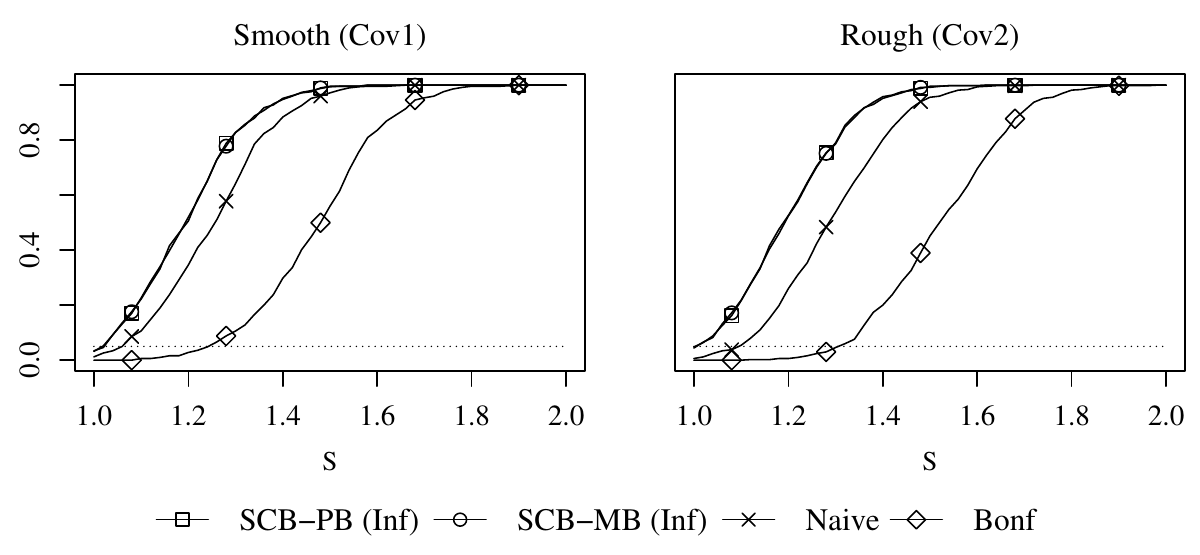}
	\caption[]{Power curves, under violated Assumption I, for $\text{ATT2}^*$, $n=200$, and $T=21$.}
	\label{fig:N200ATT2starPowerCurves_RB}
\end{figure}

\endgroup

%%%%%%%%%%%%%%%%%%%%%%%%%%%%%%%%%%%%%%%%
% References
%%%%%%%%%%%%%%%%%%%%%%%%%%%%%%%%%%%%%%%%
\spacingset{0.9}
\bibliographystyleappendix{Chicago}
\bibliographyappendix{bibfileappendix}
\spacingset{1.5}
%%%%%%%%%%%%%%%%%%%%%%%%%%%%%%%%%%%%%%%%

\end{document}